\documentclass[journal,draftcls,onecolumn,10pt]{IEEEtran}
\usepackage{psfig,amsmath,amssymb,amsthm}
\usepackage{subfigure,cite,graphicx,psfig,amsmath,amssymb,eufrak,mathrsfs,epsf,epsfig}
\usepackage[active]{srcltx}
\usepackage{tikz}
\usepackage{array}
\usepackage{multirow}
\usepackage{physics}
\usetikzlibrary{matrix,arrows,decorations.pathmorphing}

\begin{document}

\newcommand{\mc}[1]{\mathcal{#1}}
\newcommand{\up}[1]{\textup{#1}}
\newcommand{\hbf}[1]{\hat{\mathbf{#1}}}
\newcommand{\obf}[1]{\overline{\mathbf{#1}}}
\newcommand{\bmc}[1]{\boldsymbol{\mathcal{#1}}}
\newcommand{\mf}[1]{\mathfrak{#1}}
\newcommand{\kbt}{\up{k}_B\up{T}}
\providecommand{\tabularnewline}{\\}

\newcommand{\Highlight}{\color{magenta}}
\newcommand{\Red}{\color{red}}
\newcommand{\Blue}{\color{blue}}
\newcommand{\Green}{\color{green}}

\theoremstyle{remark}
\newtheorem*{claim}{\textbf{Claim}}
\newtheorem*{remark}{\textbf{Remark}}

\theoremstyle{definition}
\newtheorem{defi}{Definition}
\newtheorem{definition}{Definition}
\newtheorem{theorem}{Theorem}
\newtheorem{thm}{Theorem}
\newtheorem{lemma}{Lemma}
\newtheorem{prop}{Proposition}
\newtheorem{corollary}{Corollary}

\makeatletter
\newcommand*{\rom}[1]{\expandafter\@slowromancap\romannumeral #1@}

\sloppy

\title{
Capacity of Continuous-Space Electromagnetic Channels with Lossy Transceiver
\footnote
{
W. Jeon and S.-Y. Chung are with the School of Electrical Engineering, KAIST, Daejeon, South Korea (e-mail: wonsjeon@kaist.ac.kr, sychung@ee.kaist.ac.kr). The material in this paper was presented in
part at IEEE ISIT 2013~\cite{jeon2013capacity}, ITA 2015~\cite{jeon2015noise} and  IEEE ISIT 2015~\cite{jeon2015improving}.
}
}
\author{
  \IEEEauthorblockN{Wonseok Jeon, \emph{Student Member, IEEE}, and Sae-Young Chung, \emph{Senior Member, IEEE}}
}

\maketitle

\begin{abstract}
In this paper, the capacity of continuous-space electromagnetic channels, where transceivers are confined in given lossy regions, is analyzed.
First of all, the regions confining the transceivers are assumed to be filled with dielectric, which is either lossy or lossless.
Then, for capacity analysis, we use the exact power consumption that takes into account the electromagnetic interaction between the field and the source.
In addition, the exact noise model followed from the fluctuation-dissipation theorem in thermodynamics is used at the receive side.
The contribution of our work is summarized as follows.
First, we characterize the channel capacity as a function of the size and the physical property of the regions confining the transceivers and analytically show how the radiation efficiency affects the capacity.
We also show that the outgoing channel at the transmit side and the incoming channel at the receive side are information-theoretically equivalent, and thus, the capacities of both channels are the same.
Additionally, the quality factor, which is inversely proportional to the bandwidth, is theoretically derived, and the relationship between the spatial degrees of freedom of the channel and the quality factor is analyzed.
Besides, we consider how the power consumption is affected by the backscattered waves and compare the recent experimental demonstration with our work by solving the gain-optimization problem with the constraint on the quality factor.
\end{abstract}

{
\keywords
Continuous-space electromagnetic channel, spatial correlation, fluctuation-dissipation theorem, electromagnetic channel capacity, quality factor, spatial degrees of freedom
}
%
%\section{references}
%\begin{itemize}
%\item
%Discrete-space analysis
%\begin{itemize}
%\item
%WJ04 \cite{wallace2004mutual}
%\item
%MJW05 \cite{morris2005superdirectivity}
%\item
%BJ07 \cite{bikhazi2007relationship}
%\item
%IN09,10,11 \cite{ivrlac2009physical,ivrlac2010toward,ivrlac2011gaussian}
%\item
%MGSH10 \cite{muller2012channel}
%\item
%TH 16\cite{termos2016capacity}
%\end{itemize}
%\item
%Continuous-space analysis
%\begin{itemize}
%\item
%HF03, 06 \cite{hanlen2003capacity,hanlen2006wireless}
%\item
%PBT05 \cite{poon2005degrees}, PT11, 15 \cite{poon2011degree,poon2015does}
%\item
%XJ06\cite{xu2006electromagnetic}
%\item
%JW08 \cite{jensen2008capacity}
%\item
%GM08 \cite{gruber2008new}
%\item
%M06, 06, 08 \cite{migliore2006intuitive,migliore2006role,migliore2008electromagnetics}
%\end{itemize}
%
%\end{itemize}
%---------------------------------------------------------------------------------------------------------------------------
\section{Introduction}

The fundamental limit on the information transmission using electromagnetic waves has long been a major interest in electromagnetic theory and information theory. 
Related works on such limit in recent years can be classified into either discrete-space analysis or continuous-space analysis.
First of all, in the discrete-space analysis, the point sources are usually assumed to form an array structure and the multiport network theory is mainly utilized to model and analyze the electromagnetic system. For example, there have been some studies on the impact of antenna mutual coupling~\cite{wallace2004mutual,ivrlac2009physical,ivrlac2010toward,ivrlac2011gaussian,muller2012channel,termos2016capacity} and antenna superdirectivity~\cite{morris2005superdirectivity,bikhazi2007relationship} on the information-theoretic capacity of electromagnetic channels. 
In contrast with the discrete-space analysis, the continuous-space analysis assumes that the source is continuously distributed inside a limited space called the source region.
For example, 
\cite{hanlen2003capacity,poon2005degrees,hanlen2006wireless,xu2006electromagnetic,migliore2006intuitive,migliore2006role,migliore2008electromagnetics,poon2011degree,poon2015does}
studied the effect of
the size of the source region on the spatial degrees of freedom (DoF)
by using the continuous-space analysis.

Meanwhile, the relationship between the physical loss of electromagnetic system and the channel capacity has been considered for both discrete-space and continuous-space approaches. 
In the discrete-space analysis, the circuit-theoretic loss resistances were assumed to be placed at each antenna port, and the resultant reduction of the channel capacity was derived~\cite{ivrlac2009physical,ivrlac2010toward}.
In the continuous-space analysis, the effect of loss on the channel capacity was considered by assuming the loss on the electromagnetic channel~\cite{jensen2008capacity}.
However, even though the actual loss of the system is deeply related to the loss of the medium at the transceivers, the impact of material loss on the channel capacity has not been analyzed in the literature. 

To address the above issue, we analyze the effect of lossy medium on the electromagnetic channel capacity by using the continuous-space approach.
We summarize some existing results on continuous-space electromagnetic channels as follows.
Poon \emph{et al.}~\cite{poon2005degrees} analyzed the relationship between the size of the source region and the spatial DoF by assuming linear, circular and spherical free-space source regions. 
Later, Poon and Tse~\cite{poon2011degree} extended the methodology of~\cite{poon2005degrees} to the vector antennas and considered the extra DoF from polarization diversity.
Hanlen and Fu~\cite{hanlen2006wireless} suggested the scatter channel model and analyzed the spatial DoF. 
Xu and Janaswamy~\cite{xu2006electromagnetic} considered the DoF of electromagnetic channels when the scattering occurs in a two-dimensional region and the current strength is restricted. 
Migliore~\cite{migliore2006role} theoretically analyzed the relationship between the DoF of electromagnetic channels and the effective DoF of multi-antenna channels. 
Jensen and Wallace~\cite{jensen2008capacity} suggested a new framework using the constraint on the radiation power and background noise and compared this new framework with the conventional framework that restricts the current strength and uses the i.i.d field noise. 
\cite{jensen2008capacity} extended the research on the superdirectivity in discrete-space approach~\cite{morris2005superdirectivity} to the continuous-space approach. 
Also, the authors of \cite{jensen2008capacity} decomposed the electromagnetic channels into multiple independent sub-channels and assumed the artificial loss on each sub-channel. Then, the authors considered how those artificial loss affect the channel capacity. 
Gruber and Marengo~\cite{gruber2008new} mathematically derived the channel capacity when the source constraint is given for both the radiation power and the current strength.
In addition, the channel capacity was analyzed in~\cite{gruber2008new} by comparing the narrowband and the broadband scenarios. 
Recently, Poon and Tse~\cite{poon2015does} used the radiation power constraint and considered the relationship between the fractional bandwidth and the channel capacity.

Compared to the existing works above, the framework of our work is described as follows.
First of all, we assume the medium of the regions confining the transceivers as dielectric whose characteristic can be described by electric permittivity.
In addition, the noise model in our work assumes the i.i.d. charge fluctuation, whereas the conventional works mainly assumed the i.i.d. field fluctuation or background noise~\cite{xu2006electromagnetic,jensen2008capacity,gruber2008new,poon2015does}.
Such noise model in our work is followed exactly from the fluctuation-dissipation theorem in thermodynamics, which relates the loss of the physical system to the statistical property of the thermal noise.
Also, we use the exact power consumption that considers the electromagnetic interaction between the field and the source, whereas others mainly restricted the current strength or the radiation power.

We summarize the contribution of this paper as follows.
First, we characterize the capacity of continuous-space electromagnetic channels by considering the physical property of the regions confining the transceivers.
As a result, the channel capacity can be represented as a function of both the size and the physical property of the confining regions.
Also, we show how the radiation efficiency affects the channel capacity.
Second, we show that the outgoing channel at the transmit side and the incoming channel at the receive side are information-theoretically equivalent, and thus, the capacities of both channels are the same.
Note that a similar equivalence can also be found in~\cite{jensen2008capacity}, where the radiation power constraint is used and the isotropic background noise is assumed.
However, the equivalence in~\cite{jensen2008capacity} differs from ours since we use the constraint on the actual power consumption and the exact thermal noise due to the material loss that occurs internally at the receiver. 
Third, we derive the quality factor that is inversely proportional to the bandwidth and numerically analyze the spatial DoF of the channel under the constraint on the maximum quality factor.
Besides, we consider the impact on the near-field backscattering on the power consumption and solved the gain-optimization problem by restricting the maximum quality factor and compared our result to the recent experimental work~\cite{krasnok2014experimental}, which uses dielectric resonators and achieves high directivity and high efficiency with practically usable bandwidth.

The remainder of this paper is organized as follows.
In Section \rom{2}, some preliminaries on the electromagnetic channel, i.e., the dyadic Green function, the exact power consumption of the source, the noise model followed from the fluctuation-dissipation theorem and the physical definition of the reactive near-field region are introduced. 
In Section \rom{3}, we analyze the capacity of two different channels, i.e., the forward channel, which is the outgoing channel at the transmit side, and the reverse channel, which is the incoming channel at the receive side.
In Section \rom{4}, the quality factor is derived and the effect of backscattering on the power consumption is considered.
In Section \rom{5}, numerical results and the comparison with the existing works are given.
Finally, we conclude our paper in Section \rom{6}.

\textbf{Notation: }
In this paper, boldface letters are used for vectors or field quantities ($\mathbf{J}, \mathbf{E}, \mathbf{X},...$), and overlined boldface letters are used for matrices or operators ($\obf{G}$, $\obf{H}$, ...).
The superscript `$*$', `$H$', `$T$' denote element-wise complex conjugate (or complex conjugate for scalar quantity), conjugate-transpose and transpose, respectively.
$\Re\{A\}\triangleq(A+A^*)/2$ and $\Im\{A\}\triangleq(A-A^*)/2i$ are the real part and the imaginary part of the scalar quantity $A$, respectively. 
Similarly, $\Re\{\obf{K}\}\triangleq(\obf{K}+\obf{K}^*)/2$ and $\Im\{\obf{K}\}\triangleq(\obf{K}-\obf{K}^*)/2i$ are the real part and the imaginary part of the matrix $\obf{K}$, respectively. 
Throughout the paper, we assume the narrowband communication with steady-state variation $\exp(-i\omega t)$, where $\omega$ is the radial frequency. Also, let $k'\triangleq\Re\{k\}$ and $k''\triangleq\Im\{k\}$ for any wave number $k$ of the medium.

%---------------------------------------------------------------------------------------------------------------------------
\section{Preliminaries}\label{sec_preliminaries}

In this section, some preliminaries on the electromagnetic theory are given, which are needed to analyze the electromagnetic channel.
We first introduce the dyadic Green function, which relates the current density to the electric field, and its decomposition.
Second, the power consumption and the statistical property of noise followed from the fluctuation-dissipation theorem are given.
Third, we give the definition of near-field and far-field regions and their property at the end of this section.
Finally, we introduce how our framework is related to the framework in discrete-space analysis.

%---------------------------------------------------------------------------------------------------------------------------
\subsection{The dyadic Green function and its decomposition}\label{sec_preliminaries_dyadic}

In electromagnetic theory, the dyadic Green function (DGF) $\obf{G}$ is the kernel, which relates the current density $\mathbf{J}$ and electric field $\mathbf{E}$ as
\begin{align}
\mathbf{E}(\mathbf{r})=
i\omega\mu(\mathbf{r})\int\obf{G}(\mathbf{r},\mathbf{r}')\mathbf{J}(\mathbf{r}')d\mathbf{r}',
\label{sec_preliminaries_dyadic_1}
\end{align}
where $\omega$ is the radial frequency and the permeability at $\mathbf{r}$ is equal tothe free-space permeability $\mu_0$\footnote{In this paper, we only consider the dielectric media, which satisfies the condition $\mu(\mathbf{r})=\mu_0$.}.
The decomposition of DGF is given in~\cite{chew1995waves,li1994electromagnetic}, and we summarize it as follows:
The spherical vector waves, which are used for the decomposition of $\obf	{G}$, are defined as
\begin{align}
\mathbf{U}_{nm1}(k,\mathbf{r})
&\triangleq\nabla\times\mathbf{r}h_n^{(1)}(kr)Y_{nm}(\theta,\phi),\\
\mathbf{V}_{nm1}(k,\mathbf{r})
&\triangleq\nabla\times\mathbf{r}j_n(kr)Y_{nm}(\theta,\phi),\\
\mathbf{W}_{nm1}(k,\mathbf{r})
&\triangleq\nabla\times\mathbf{r}y_n(kr)Y_{nm}(\theta,\phi),\\
\mathbf{U}_{nm2}(k,\mathbf{r})
&\triangleq\frac{1}{k}\nabla\times\mathbf{U}_{nm1}(k,\mathbf{r}),\\
\mathbf{V}_{nm2}(k,\mathbf{r})
&\triangleq\frac{1}{k}\nabla\times\mathbf{V}_{nm1}(k,\mathbf{r}),\\
\mathbf{W}_{nm2}(k,\mathbf{r})
&\triangleq\frac{1}{k}\nabla\times\mathbf{W}_{nm1}(k,\mathbf{r})
\end{align}
for integers $n\ge1, -n\le m \le n, l=1,2$,  where $k\triangleq\omega\sqrt{\mu\epsilon}$ is the wavenumber for the permittivity $\epsilon$ and the permeability $\mu$ of the medium, $j_n,y_n,h_n^{(1)}$ are the spherical Bessel function of the first, the second and the third kind, respectively, and $Y_{nm}$ is the spherical harmonics.
Note that the value of $k$ is a complex number in general depending on the property of medium.
Similarly, the conjugate wave functions are defined as 
\begin{align}
\mathbf{U}_{nm1}^\star(k,\mathbf{r})
&\triangleq
\nabla\times \mathbf{r}h_n^{(1)}(kr)^*Y_{nm}(\theta,\phi),\\
\mathbf{V}_{nm1}^\star(k,\mathbf{r})
&\triangleq
\nabla\times \mathbf{r}j_n(kr)^*Y_{nm}(\theta,\phi),\\
\mathbf{W}_{nm1}^\star(k,\mathbf{r})
&\triangleq
\nabla\times \mathbf{r}y_n(kr)^*Y_{nm}(\theta,\phi),\\
\mathbf{U}_{nm2}^\star(k,\mathbf{r})
&\triangleq
\frac{1}{k^*}\nabla\times \mathbf{U}_{nm1}^\star(k,\mathbf{r}),\\
\mathbf{V}_{nm2}^\star(k,\mathbf{r})
&\triangleq
\frac{1}{k^*}\nabla\times \mathbf{V}_{nm1}^\star(k,\mathbf{r}),\\
\mathbf{W}_{nm2}^\star(k,\mathbf{r})
&\triangleq
\frac{1}{k^*}\nabla\times \mathbf{W}_{nm1}^\star(k,\mathbf{r})
\end{align}
for all $n,m,l$. We explicitly derive the spherical vector waves and their properties in Appendix~\ref{sec_appendix_definitions}.
\begin{figure}[t]
\begin{center}
\begin{tikzpicture}
\node at (0,0) {\includegraphics[page=1,height=2in]{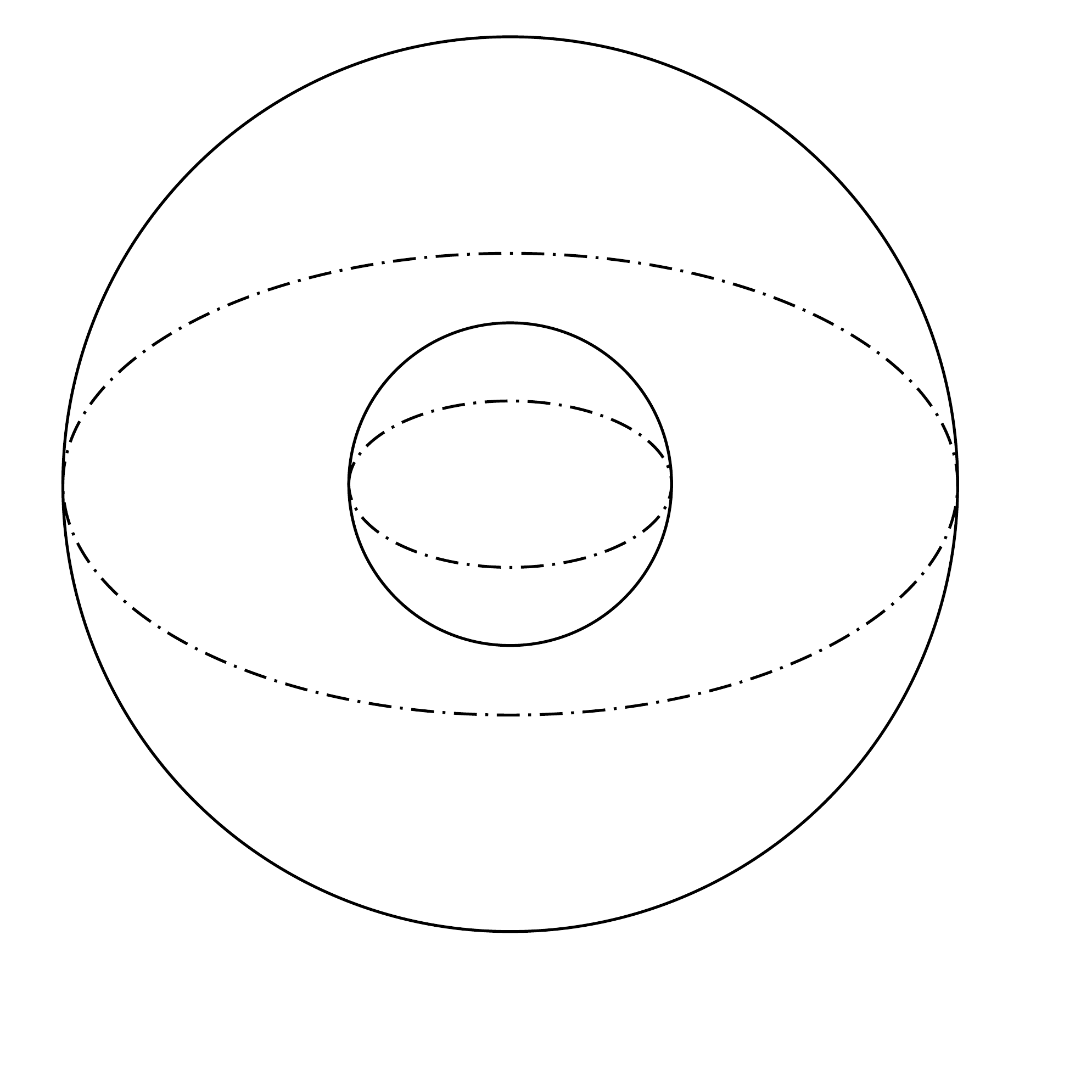}};
\node at (-0.8in,0.8in) {$S$};
\node at (-0.3in,0.3in) {$V$};
\end{tikzpicture}
\vspace{-0.2in}
\caption{Two concentric spheres $V$ and $S$}\label{figure1}
\end{center}
\vspace{-0.4in}
\end{figure}

Now, assume there are concentric spheres $V$ and $S$ with radius $R_1$ and $R_2>R_1$, respectively (Fig.~\ref{figure1}).
In addition, the wave numbers of the region $V$ and the region $V^C$, the outside of $V$, are assumed to be $k_1\triangleq\omega\sqrt{\mu_0\epsilon_1}$ and $k_2=k_0\triangleq\omega\sqrt{\mu_0\epsilon_0}$, where $\epsilon_1$ and $\epsilon_0$ are the permittivity of $V$ and the free-space permittivity, respectively. 
We define the normalized vector wave functions as
\begin{align}
\mathbf{v}_{nml}(k,\mathbf{r})
\triangleq
\frac{\mathbf{V}_{nml}(k,\mathbf{r})}{\mathcal{N}_{V,\mathbf{V}_{nml}}(k)},
\mathbf{v}_{nml}^\star(k,\mathbf{r})
\triangleq
\frac{\mathbf{V}_{nml}^\star(k,\mathbf{r})}{\mathcal{N}_{V,\mathbf{V}_{nml}^\star}(k)},
\mathbf{u}_{nml}(k,\mathbf{r})
\triangleq
\frac{\mathbf{U}_{nml}(k,\mathbf{r})}{\mathcal{N}_{S,\mathbf{U}_{nml}}(k)},
\mathbf{u}_{nml}^\star(k,\mathbf{r})
\triangleq
\frac{\mathbf{U}_{nml}^\star(k,\mathbf{r})}{\mathcal{N}_{S,\mathbf{U}_{nml}^\star}(k)},
\end{align}
where the normalization coefficients are
\begin{align}
\mathcal{N}_{V,\mathbf{V}_{nml}}(k)
&\triangleq 
\sqrt{\left\langle\mathbf{V}_{nml}(k,\cdot),\mathbf{V}_{nml}(k,\cdot)\right\rangle_V}=
\mathcal{N}_{V,\mathbf{V}_{nml}^\star}(k)
,\forall n,m,l\\
\mathcal{N}_{S,\mathbf{U}_{nml}}(k)
&\triangleq
\sqrt{\left\langle\mathbf{U}_{nml}(k,\cdot),\mathbf{U}_{nml}(k,\cdot)\right\rangle_S}=
\mathcal{N}_{S,\mathbf{U}_{nml}^\star}(k),\forall n,m,l,
\end{align}
and the inner products between two spherical vector waves $\mathbf{F}_1$ and $\mathbf{F}_2$ are
\begin{align}
&\left\langle
\mathbf{F}_1(k,\cdot),\mathbf{F}_2(k,\cdot)
\right\rangle_V\triangleq
\int_V\mathbf{F}_1(k,\mathbf{r})^H \mathbf{F}_2(k,\mathbf{r}) d \mathbf{r},\\
&\left\langle
\mathbf{F}_1(k,\cdot),\mathbf{F}_2(k,\cdot)
\right\rangle_S\triangleq
\int\mathbf{F}_1(k,(R_S,\Omega))^H\mathbf{F}_2(k,(R_S,\Omega))d\Omega,
\end{align}
for the angular position $\Omega\triangleq(\theta,\phi)$. 
The normalization coefficients are derived in Appendix~\ref{sec_appendix_definitions} and they are independent of $m$ and only depend on $(n,l)$. 
Thus, let $\mathcal{N}_{n,l}^S\triangleq \mathcal{N}_{S,\mathbf{U}_{n,0,l}}(k_0), \mathcal{N}_{n,l}^V\triangleq \mathcal{N}_{V,\mathbf{V}_{n,0,l}}(k_1)$.
By using the vector wave functions and the conjugate vector wave functions, the dyadic Green function $\obf{G}$ can be decomposed as
\begin{align}
\obf{G}(\mathbf{r},\mathbf{r}')
=
ik(\mathbf{r}')
\sum_{n,l}\frac{1}{n(n+1)}\obf{g}_{nl}(\mathbf{r},\mathbf{r}')-\frac{\hat{\mathbf{r}}\hat{\mathbf{r}}^T}{k(\mathbf{r}')^2}\delta(\mathbf{r}-\mathbf{r}'),
\end{align}
where $\hbf{r}\triangleq\mathbf{r}/\left\|\mathbf{r}\right\|$ and $\mathbf{g}_{nl}(\mathbf{r},\mathbf{r}')$ is defined as follows\footnote{The case for $\mathbf{r},\mathbf{r}'\in V^C$ is not given since the decomposition for that case is not used in this paper.}:
for $\mathbf{r},\mathbf{r}'\in V$,
\begin{align}
\obf{g}_{nl}(\mathbf{r},\mathbf{r}')=
\begin{cases}
\sum_{m}
\left[\mathbf{U}_{nml}(k_1,\mathbf{r})+\mathcal{R}_{n,l}\mathbf{V}_{nml}(k_1,\mathbf{r})\right]
\mathbf{V}_{nml}^\star(k_1,\mathbf{r}')^H,
&\text{~if~}r\ge r',\\
\sum_{m}
\mathbf{V}_{nml}(k_1,\mathbf{r})
\left[\mathbf{U}_{nml}^\star(k_1,\mathbf{r}')^H+\mathcal{R}_{n,l}\mathbf{V}_{nml}^\star(k_1,\mathbf{r}')^H\right],
&\text{~if~}r\le r',
\end{cases}
\end{align}
for $\mathbf{r}\in V^C$ and $\mathbf{r}'\in V$, 
\begin{align}
\obf{g}_{nl}(\mathbf{r},\mathbf{r}')
&=
\mathcal{T}_{n,l}(k_1)
\sum_{m}
\mathbf{U}_{nml}(k_0,\mathbf{r})
\mathbf{V}_{nml}^\star(k_1,\mathbf{r}')^H,
\end{align}
and for $\mathbf{r}\in V$ and $\mathbf{r}'\in V^C$, 
\begin{align}
\obf{g}_{nl}(\mathbf{r},\mathbf{r}')
&=
\mathcal{C}(k_1)\mathcal{T}_{n,l}(k_1)
\sum_{m}
\mathbf{V}_{nml}(k_1,\mathbf{r})
\mathbf{U}_{nml}^\star(k_0,\mathbf{r}')^H
\end{align}
for $\mathcal{C}(k)=k/k_0$.
Here, $\mathcal{R}_{n,l}$ and $\mathcal{T}_{n,l}$ are the scattering coefficients defined in~\cite{chew1995waves,li1994electromagnetic}, and we derived those coefficients as 
\begin{align}
\mathcal{R}_{n,1}(k)
&=
\frac{
\mathcal{C}
\hat{H}_n^{(1)}(z)
\hat{H}_n^{(1)'}(\mathcal{C}z)
-
\hat{H}_n^{(1)'}(z)
\hat{H}_n^{(1)}(\mathcal{C}z)
}
{
\hat{J}_n(\mathcal{C}z)
\hat{H}_n^{(1)'}(z)
-
\mathcal{C}
\hat{J}_n'(\mathcal{C}z)
\hat{H}_n^{(1)}(z)
}
,
\\
\mathcal{R}_{n,2}(k)
&=
\frac{
\hat{H}_n^{(1)}(z)
\hat{H}_n^{(1)'}(\mathcal{C}z)
-
\mathcal{C}
\hat{H}_n^{(1)'}(z)
\hat{H}_n^{(1)}(\mathcal{C}z)
}
{
\mathcal{C}
\hat{J}_n(\mathcal{C}z)
\hat{H}_n^{(1)'}(z)
-
\hat{J}_n'(\mathcal{C}z)
\hat{H}_n^{(1)}(z)
},\\
\mathcal{T}_{n,1}(k)
&=
\frac{
i
}
{
\hat{J}_n(\mathcal{C}z)
\hat{H}_n^{(1)'}(z)
-
\mathcal{C}
\hat{J}_n'(\mathcal{C}z)
\hat{H}_n^{(1)}(z)
},
\\
\mathcal{T}_{n,2}(k)
&=
\frac{
i
}
{
\mathcal{C}
\hat{J}_n(\mathcal{C}z)
\hat{H}_n^{(1)'}(z)
-
\hat{J}_n'(\mathcal{C}z)
\hat{H}_n^{(1)}(z)
}
\end{align}
for all $n$, where $\hat{H}_n^{(1)}(\rho)\triangleq \rho h_n^{(1)}(\rho)$, $\hat{J}_n(\rho)\triangleq \rho j_n(\rho)$, $\hat{H}_n^{(1)'}(\rho)\triangleq\frac{d}{d\rho}\hat{H}_n^{(1)}(\rho)$, $\hat{J}_n'(\rho)\triangleq \frac{d}{d\rho}\hat{J}_n(\rho)$, $z\triangleq k_0R_1$ and the argument $k$ of $\mathcal{C}(k)$ is omitted. Throughout the paper, the argument of scattering coefficients are omitted if $k=k_1$. 

%---------------------------------------------------------------------------------------------------------------------------
\subsection{Electromagnetic power consumption}\label{sec_preliminaries_electromagnetic}

In~\cite{harrington1961time,chew1995waves}, the complex power of the current source is equal to
\begin{align}
-\frac{1}{2}
\left\langle
\mathbf{J},\mathbf{E}
\right\rangle_V,
\end{align}
where $V$ is the source region having arbitary shape.
By using the definition of DGF \eqref{sec_preliminaries_dyadic_1}, it is equal to
\begin{align}
-\frac{1}{2}
\left\langle\mathbf{J},i\omega\mu_0\obf{G},\mathbf{J}\right\rangle_V
&=
-\frac{i\omega\mu_0}{2}
\left\langle\mathbf{J},\obf{G},\mathbf{J}\right\rangle_V\triangleq
-\frac{i\omega\mu_0}{2}A,
\end{align}
where 
\begin{align}
\left\langle\mathbf{F}_1,\obf{T},\mathbf{F}_2\right\rangle_V
\triangleq
\int_V
\int_V
\mathbf{F}_1(\mathbf{r})^H\obf{T}(\mathbf{r},\mathbf{r}')\mathbf{F}_2(\mathbf{r}')
d\mathbf{r}d\mathbf{r}'
\end{align}
for a linear kernel $\obf{T}$ and two vectors $\mathbf{F}_1,\mathbf{F}_2$.
The real part of the complex power is the power that is consumed by the source, i.e.,
\begin{align}
\frac{1}{2}\Re\{-i\omega\mu_0A\}=\frac{\omega\mu_0}{2}\frac{A-A^*}{2i}.
\end{align}
Define the operator $\mathrm{S}$ that exchanges the argument, i.e., for all $\mathbf{r},\mathbf{r}'$, $
\mathrm{S}(\obf{G})(\mathbf{r},\mathbf{r}')\triangleq\obf{G}(\mathbf{r}',\mathbf{r})$.
Then, the complex conjugate $A^*$ of $A$ is equal to $\left\langle\mathbf{J},\obf{G}^*,\mathbf{J}\right\rangle_V$ since 
\begin{align}
A^*
%&=
%\left(
%\left\langle\mathbf{J}(\mathbf{r}),\omega\mu_0\obf{G}(\mathbf{r},\mathbf{r}'),\mathbf{J}(\mathbf{r}')\right\rangle_V
%\right)^*\\
=
\left\langle\mathbf{J},\obf{G},\mathbf{J}\right\rangle_V^*
%&=
%\left(
%\{
%\left\langle\mathbf{J}(\mathbf{r}),\omega\mu_0\obf{G}(\mathbf{r},\mathbf{r}'),\mathbf{J}(\mathbf{r}')\right\rangle_V\}^T
%\right)^*\\
=
\left(
\left\langle\mathbf{J},\obf{G},\mathbf{J}\right\rangle_V^T
\right)^*
%&=
%\left(
%\left\langle\mathbf{J}(\mathbf{r}')^*,\omega\mu_0\obf{G}(\mathbf{r},\mathbf{r}')^T,\mathbf{J}(\mathbf{r})^*\right\rangle_V
%\right)^*\\
=
\left(
\left\langle\mathbf{J}^*,\mathcal{S}(\obf{G})^T,\mathbf{J}^*\right\rangle_V
\right)^*
%&=
%\left(
%\left\langle\mathbf{J}(\mathbf{r}')^*,\omega\mu_0\obf{G}(\mathbf{r}',\mathbf{r}),\mathbf{J}(\mathbf{r})^*\right\rangle_V
%\right)^*\\
%&=
%\left\langle\mathbf{J}(\mathbf{r}'),\omega\mu_0\obf{G}(\mathbf{r}',\mathbf{r})^*,\mathbf{J}(\mathbf{r})\right\rangle_V\\
%&=
%\left\langle\mathbf{J}(\mathbf{r}),\omega\mu_0\obf{G}(\mathbf{r},\mathbf{r}')^*,\mathbf{J}(\mathbf{r}')\right\rangle_V\\
=
\left(
\left\langle\mathbf{J}^*,\obf{G},\mathbf{J}^*\right\rangle_V
\right)^*
=
\left\langle\mathbf{J},\obf{G}^*,\mathbf{J}\right\rangle_V,
\end{align}
where the fourth equality holds due to the reciprocity of DGF~\cite{chew1995waves}. As a result, the power consumption due to $\mathbf{J}$ can be represented as follows:
\begin{align}
\frac{\omega\mu_0}{2}\Re\{-i\left\langle\mathbf{J},\obf{G},\mathbf{J}\right\rangle_V\}
%&=
%\frac{\omega\mu_0}{2}
%\left\langle
%\mathbf{J},
%\frac{\obf{G}-\obf{G}^*}{2i},
%\mathbf{J}
%\right\rangle_\mathcal{R}\\
&=
\frac{\omega\mu_0}{2}
\left\langle
\mathbf{J},
\Im\{\obf{G}\},
\mathbf{J}\right\rangle_V.
\end{align}

%---------------------------------------------------------------------------------------------------------------------------
\subsection{Thermal noise: fluctuation-dissipation theorem}\label{sec_preliminaries_thermal}

The statistical property of thermal noise was first analyzed by Johnson and Nyquist~\cite{PhysRev.32.97, PhysRev.32.110}.
The authors show that the thermal current fluctuation $i_{\mathrm{noise}}$ across the conductor is inversely proportional to the resistance $\mathrm{R}$ of the conductor, i.e.,
\begin{align}
E\{i_{\mathrm{noise}}^2\}=\frac{4k_BTB}{\mathrm{R}},
\end{align}
or equivalently, for the thermal voltage fluctuation $v_{\mathrm{noise}}$ across the conductor, 
\begin{align}
E\{v_{\mathrm{noise}}^2\}=4k_BTB\mathrm{R},
\end{align}
where $k_B$ is the Boltzmann constant, $T$ is the temperature of the conductor, and $B$ is the bandwidth. 
Later in statistical physics, Johnson-Nyquist thermal noise was generalized for various physical systems~\cite{PhysRev.83.34}, where \emph{the energy dissipation in thermodynamical system is shown to be related to the thermal fluctuation}, which is called the \emph{fluctuation-dissipation theorem (FDT)}\footnote{In~\cite{novotny2006principles}, FDT is introduced in view of classical electromagnetic theory.}~\cite{PhysRev.83.34}. 

In this paper, we take a noise model based on the FDT.
We first assume there is a dielectric sphere $V$ in which the physical loss can be characterized by the complex permittivity $\epsilon$ of $V$. 
Then, from~\cite{shchegrov2000near,novotny2006principles,rytov1989principles}, the statistical property of the noise current density satisfies the following property:
\begin{lemma}[FDT for narrowband analysis]
At temperature $T$, the statistics of the thermal charge fluctuation $\mathbf{J}_{\mathrm{noise}}$ in a dielectric medium $V$ follows
\begin{align}
E
\left\{
\mathbf{J}_{\mathrm{noise}}(\mathbf{r})
\mathbf{J}_{\mathrm{noise}}(\mathbf{r}')^H
\right\}
=4k_BTB\omega\epsilon'' \obf{I}\delta(\mathbf{r}-\mathbf{r}'),
\end{align}
where $k_B$ is the Boltzmann constant, $B$ is bandwidth, $\epsilon''\triangleq \Im\left\{\epsilon\right\}$ for the complex permittivity $\epsilon$ of the medium in which the thermal fluctuation exists, and $\mathbf{I}$ is the identity operator. 
\end{lemma}
Since $\sigma=\omega\epsilon''$, where $\sigma$ is the dielectric conductivity of the medium, the FDT states that the thermal charge fluctuation is proportional to the temperature and the conductivity of the medium.
This is the exact generalization of the Johnson-Nyquist thermal noise at the resistor stated above.
In addition, the FDT gives the pointwise independence of the thermal charge fluctuation.
By using the FDT and the definition of the dyadic Green function, the following corollary can be derived~\cite{novotny2006principles,dung1998three}:
\begin{corollary}\label{sec_preliminaries_thermal_corollary}
At temperature $T$, the statistical property of the thermal electric field $\mathbf{E}_{\mathrm{noise}}$ due to the charge fluctuation $\mathbf{J}_{\mathrm{noise}}$ follows
\begin{align}
E
\left\{
\mathbf{E}_{\mathrm{noise}}(\mathbf{r})
\mathbf{E}_{\mathrm{noise}}(\mathbf{r}')^H
\right\}
=
4k_BTB\omega\mu_0 \Im\left\{\obf{G}(\mathbf{r},\mathbf{r}')\right\},
\end{align}
where $k_B$ is the Boltzmann constant and $B$ is bandwidth.
\end{corollary}
Note that the conductivity is not explicitly shown in the corollary since $\omega\mu_0 \Im\{\overline{\mathbf{G}}\}$ works as the resistivity, which directly comes from the definition of the dyadic Green function. Again, this is the exact generalization of the Johnson-Nyquist thermal noise at the resistor, which states that the variance of the voltage noise across the resistor is proportional to the resistance. However, FDT states the electric field fluctuation is spatially correlated over different points by the imaginary part of the dyadic Green function.

%---------------------------------------------------------------------------------------------------------------------------
\subsection{Near-field region for the spherical vector waves}\label{sec_preliminaries_near}

In~\cite{potter1967application,hansen1988spherical}, the spherical vector wave with its order $n$ has the following physical regions:
\begin{itemize}
\item
Far-field region (Fraunhofer region), $k_0r \gtrsim 4n^2/\pi  $: The amplitude of the wave is inversely proportional to $r$ and shows the phase variation as $\exp(ikr)$.
\item
Radiative near-field region (Fresnel region), $n \lesssim k_0r \lesssim 4n^2/\pi $: The amplitude of the wave is inversely proportional to $r$. However, the phase variation differs by the order of the wave. 
\item
Reactive near-field region (Evanescent region), $k_0r \lesssim n$: The amplitude increases exponentially as $r$ decreases. In addition, the phase variation differs by the order of the wave. 
\end{itemize}
This shows the physical size of the reactive near-field region becomes larger as the order $n$ of the orthogonal bases increases.

%---------------------------------------------------------------------------------------------------------------------------
\subsection{Connection to discrete-space analysis}\label{sec_preliminaries_connection}

\begin{figure}[t]
\begin{center}
\begin{tikzpicture}
\node at (0,0) {\includegraphics[page=1,height=2in]{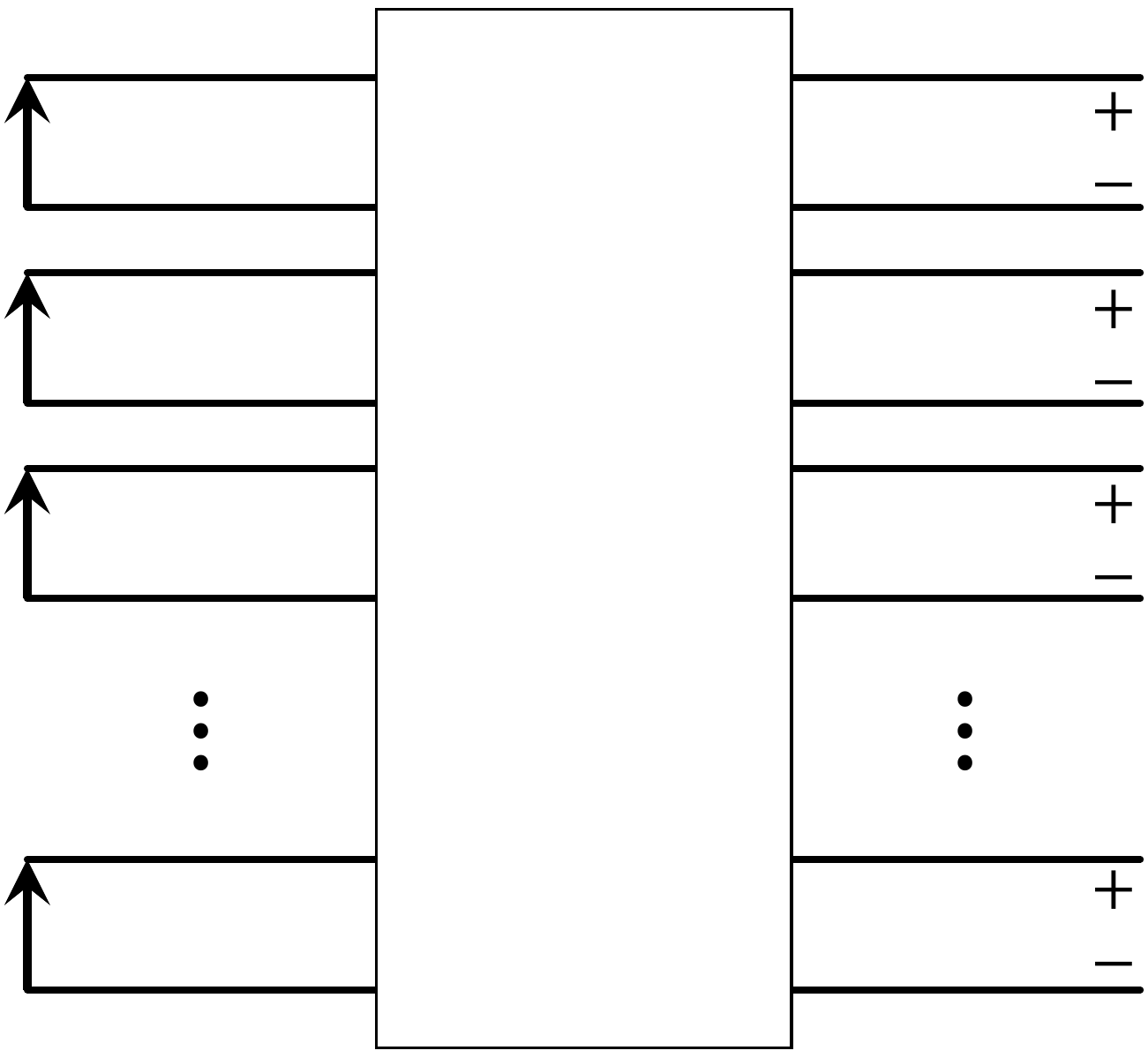}};
\node at (0,0) {$\overline{\mathbf{Z}}$};
\node at (-1.2in,+0.70in) {$i_1$};
\node at (-1.2in,+0.35in) {$i_2$};
\node at (-1.2in,+0.00in) {$i_3$};
\node at (-1.2in,-0.75in) {$i_N$};
\node at (+1.2in,+0.70in) {$v_1$};
\node at (+1.2in,+0.35in) {$v_2$};
\node at (+1.2in,+0.00in) {$v_3$};
\node at (+1.2in,-0.75in) {$v_M$};
\end{tikzpicture}
\vspace{-0.2in}
\caption{Linear multiport model for multi-antenna channel}\label{figure2page1}
\end{center}
\vspace{-0.4in}
\end{figure}
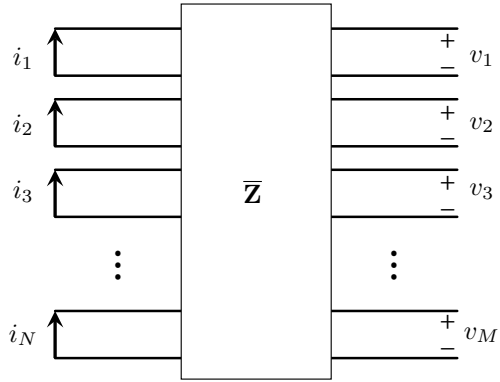

From~\cite{wallace2004mutual, ivrlac2010toward}, the multi-antenna radio channel can be modeled as an equivalent linear multiport network.
For example, if $N$ idealized current sources $\mathbf{i}_i\triangleq[i_1,...,i_N]^T$ and $M$ open-circuit voltages $\mathbf{v}_o\triangleq[v_1,...,v_M]^T$ are used as the channel input and output, respectively, there is an equivalent multiport system with an impedance matrix $\obf{Z}$ as shown in Fig.~\ref{figure2page1}.
Formally, if the voltage across the channel input is $\mathbf{v}_i$ and the current at the output port is $\mathbf{i}_o$, which is zero when output port is open, the multiport model is obtained as 
\begin{align}
\begin{bmatrix}
\mathbf{v}_i\\
\mathbf{v}_o
\end{bmatrix}
=
\obf{Z}
\begin{bmatrix}
\mathbf{i}_i\\
\mathbf{i}_o
\end{bmatrix}
\triangleq
\begin{bmatrix}
\obf{Z}_T 		&\obf{Z}_{TR} 		\\
\obf{Z}_{RT} 		&\obf{Z}_R 
\end{bmatrix}
\begin{bmatrix}
\mathbf{i}_i\\
\mathbf{i}_o
\end{bmatrix},
\end{align}
where $\obf{Z}$ is partitioned into four matrices: transmit and receive impedance matrix $\obf{Z}_T\in\mathbb{C}^{N\times N}$ and $\obf{Z}_R\in\mathbb{C}^{M\times M}$, the channel from the transmitter to the receiver $\obf{Z}_{RT}\in\mathbb{C}^{M\times N}$, and the reverse channel from the receiver to the transmitter $\obf{Z}_{TR}\in\mathbb{C}^{N\times M}$. 
When the receiver is sufficiently far apart from the transmitter, $\obf{Z}_{TR}$ is negligible at the transmitter relative to $\obf{Z}_{T}$, and the power consumption is equal to
\begin{align}
\frac{1}{2}\mathbf{i}_i^H \Re\{\obf{Z}_T\}\mathbf{i}_i.
\label{sec_preliminaries_connection_1}
\end{align}
In addition, for the noise temperature $T$ of the antennas, the open-circuit noise voltage at the output port satisfies
\begin{align}
E\left\{\mathbf{v}_{o,N}\mathbf{v}_{o,N}^H\right\}=4k_B T B \Re\{\obf{Z}_R\}.
\label{sec_preliminaries_connection_2}
\end{align}
The power consumption and the noise statistics in this paper are related to \eqref{sec_preliminaries_connection_1} and \eqref{sec_preliminaries_connection_2}.
As an example, for the antenna array with $N$ Hertzian dipoles, the power consumption and the noise statistics in the multi-antenna channel above can be derived from the power consumption and the noise statistics in this paper.

\begin{figure}[t]
\begin{center}
\begin{tikzpicture}
\node at (0,0) {\includegraphics[page=2,height=1in]{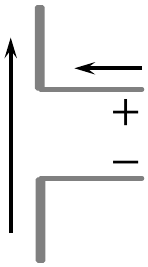}};
\node at (-0.4in,+0.0in) {$\hbf{e}_i$};
\node at (+0.4in,+0.0in) {$V_i$};
\node at (+0.2in,+0.4in) {$I_i$};
\end{tikzpicture}
\caption{Definitions on current direction, voltage sign and the orientation of Hertzian dipole}\label{figure2page2}
\end{center}
\end{figure}
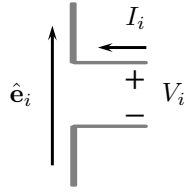

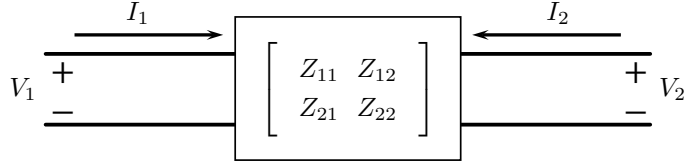
\begin{figure}
\begin{center}
\begin{tikzpicture}
\node at (0,0) {\includegraphics[page=3,height=0.8in]{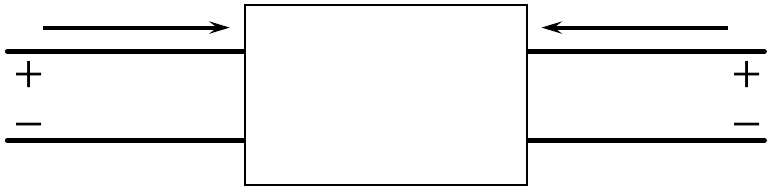}};
\node at (-1.7in,+0.0in) {$V_1$};
\node at (+1.7in,+0.0in) {$V_2$};
\node at (-1.1in,+0.4in) {$I_1$};
\node at (+1.1in,+0.4in) {$I_2$};
\matrix [matrix of math nodes, left delimiter  = {[}, right delimiter = {]}] at (0,0)
{
  Z_{11} & Z_{12} \\
  Z_{21} & Z_{22} \\
};
\end{tikzpicture}
\end{center}
\vspace{-0.1in}
\caption{Equivalent two-port network between two dipoles}\label{figure2page3}
\vspace{-0.1in}
\end{figure}
\begin{corollary}
For the transmitter with $N$ lossless Hertzian dipoles with their input current $\mathbf{i}\triangleq[I_1,...,I_N]^T$ and mutual impedance $\obf{Z}_T\triangleq[\mathrm{Z}_{ij}]_{1\le i,j\le N}$, the power consumption is
\begin{align}
\frac{1}{2}\mathbf{i}^H\Re\{\obf{Z}_T\}\mathbf{i}.
\end{align}
\end{corollary}
\begin{proof}
Assume two Hertzian dipoles located at $\mathbf{r}_1,\mathbf{r}_2$ with $\mathbf{l}_i=\hbf{e}_iL_i,i=1,2,$ where $\hbf{e}_i$ and $L_i$ are the orientation and the length of dipoles, respectively. In this paper, the direction of current $I_i,i=1,2,$ and the sign of the open-circuit voltage $V_i,i=1,2,$ are defined along the orientation as described in Fig. \ref{figure2page3}, and the current density on the dipole is
\begin{align}
\mathbf{J}_i(\mathbf{r})=\mathbf{l}_iI_i\delta(\mathbf{r}-\mathbf{r}_i), i=1,2.
\end{align}
For the electric field $\mathbf{E}$ on each dipole, the direction of current flow should be the same as that of the field, which implies the signs of $\mathbf{l}_i^T\mathbf{E}, \mathbf{l}_i^T\mathbf{J}_i$ and $I_i$ should be the same. On the other hand, $V_i$ should have the opposite sign of $I_i$ under our description, i.e.,
\begin{align}
V_i\triangleq-\mathbf{l}_i^T\mathbf{E}(\mathbf{r}_i).
\end{align}
In addition, the spatial correlation between two dipoles can be equivalently modeled as a two-port network as described in Fig. \ref{figure2page3}, where the impedance is defined as
\begin{align}
\mathrm{Z}_{ij}\triangleq\left.\frac{V_i}{I_j}\right|_{I_{k}=0}, i,j,k=1,2, k\ne j.
\end{align}
Then, from the definitions above, we have
$
\mathrm{Z}_{12}=-i\omega\mu_0\mathbf{l}_1^T\obf{G}(\mathbf{r}_1,\mathbf{r}_2)\mathbf{l}_2,
$
which follows since $\mathbf{E}(\mathbf{r}_1)=i\omega\mu_0\obf{G}(\mathbf{r}_1,\mathbf{r}_2)\mathbf{l}_2I_2$.
The real part of impedance $\mathrm{Z}_{12}$ is
\begin{align}
\Re\{\mathrm{Z}_{12}\}=\Im\left\{\omega\mu_0\mathbf{l}_1^T\obf{G}(\mathbf{r}_1,\mathbf{r}_2)\mathbf{l}_2\right\}.
\end{align}
For the current density $\mathbf{J}_1+\mathbf{J}_2$, the power consumption is equal to
\begin{align}
\frac{\omega\mu_0}{2}\sum_{i=1}^2\sum_{j=1}^2I_i^*I_j\Im\{\mathbf{l}_i^T\obf{G}(\mathbf{r}_i,\mathbf{r}_j)\mathbf{l}_j\}
=\frac{1}{2}\mathbf{i}^H\Re\{\obf{Z}\}\mathbf{i},
\end{align}
where $\mathbf{i}\triangleq[I_1,I_2]^T, \obf{Z}\triangleq[\mathrm{Z}_{ij}]_{i,j=1,2}$. The generalization for $N$-port antenna array is straightforward. 
\end{proof}
%
%\begin{figure}[t]
%\begin{center}
%\begin{tikzpicture}
%\node at (0,0) {\includegraphics[height=1.6in]{fig5.pdf}};
%\node at (-3.8,1.2)  [rotate=345] {$V_1$};
%\node at (3.8,1.5)  [rotate=30] {$V_2$};
%\node at (1.5,0.5) {$\mathbf{E}_N$};
%\node at (-1.4,0.4) {$\mathbf{E}_N$};
%\node at (1.6,0) {$~~~$};
%\end{tikzpicture}
%\end{center}
%\vspace{-0.1in}
%\caption{Illustration on the background noise}\label{fig:dipole}
%\vspace{-0.1in}
%\end{figure}

\begin{corollary}
For the receiver with $N$ lossless Hertzian dipoles with their open-circuit noise voltage $\mathbf{v}_{\mathrm{noise}}\triangleq[v_{\mathrm{noise},1},...,v_{\mathrm{noise},N}]^T$ and mutual impedance $\obf{Z}_R\triangleq[\mathrm{Z}_{ij}]_{1\le i,j\le N}$, the noise voltages satisfies 
\begin{align}
E\left\{\mathbf{v}_{\mathrm{noise}}\mathbf{v}_{\mathrm{noise}}^H\right\}
=
4k_BTB\Re\left\{\obf{Z}_R\right\}.
\end{align}
\end{corollary}
\begin{proof}
Assume there are two open-circuit Hertzian dipoles.
For the background noise field $\mathbf{E}_{\mathrm{noise}}$ on each dipole, the open-circuit noise voltage across the dipole is
\begin{align}
v_{\mathrm{noise},i}=-\mathbf{l}_i^T\mathbf{E}_{\mathrm{noise}}(\mathbf{r}_i), i=1,2.
\end{align}
From FDT, we have
\begin{align}
E\{v_{\mathrm{noise},1}v_{\mathrm{noise},2}^*\}
&=
4k_B T B \Re\{\mathrm{Z}_{12}\}
\end{align}
since $\Re\{\mathrm{Z}_{12}\}=\mathbf{l}_1^TE\{\mathbf{E}_{\mathrm{noise}}(\mathbf{r}_1)\mathbf{E}_{\mathrm{noise}}^H(\mathbf{r}_2)\}\mathbf{l}_2$.
In addition, if $\mathbf{r}_2=\mathbf{r}_1$, 
\begin{align}
E\{|v_{\mathrm{noise},1}|^2\}=4k_BTB\Re\{\mathrm{Z}_{11}\},
\end{align}
which follows since $\Re\{\mathrm{Z}_{11}\}=\omega\mu_0\mathbf{l}_1^T\Im\{\obf{G}(\mathbf{r}_1,\mathbf{r}_1)\}\mathbf{l}_2=20k^2L_1^2$, $\Im\left\{\obf{G}(\mathbf{r}_1,\mathbf{r}_1)\right\}=(k/(6\pi))\obf{I}$ and $\omega\mu_0=k_0\sqrt{\mu_0/\epsilon_0}=120\pi k_0$.
Note that $\Re\{\mathrm{Z}_{11}\}$ is equal to the radiation resistance of the Hertzian dipole~\cite{balanis2005antenna}, and the statistics is equal to the formula of the Johnson-Nyquist thermal noise where the resistance is equal to $\Re\{\mathrm{Z}_{11}\}$. The generalization for the multiport is trivial. 
\end{proof}
We summarize the connection between the discrete-space analysis and our framework in Table~\ref{table1}. 
\begin{table}
\begin{center}
\begin{tabular}{c | c}
Discrete-space analysis						&		Continuous-space analysis							\\
\hline\hline
Current input									&		Current density at the transmitter					\\
$\mathbf{i}$									&		$\mathbf{J}$										\\
\hline
Open-circuit voltage output						&		Electric field at the receiver 						\\
$\mathbf{v}$									&		$\mathbf{E}			$							\\
\hline
Ohm's law 									&		Role of dyadic Green function 						\\
$\mathbf{v}=\obf{Z}\mathbf{i}$				
& 		
$\mathbf{E}(\mathbf{r})=i\omega\mu_0\int_\mathcal{R}\obf{G}(\mathbf{r},\mathbf{r}')\mathbf{J}(\mathbf{r}')d \mathbf{r}'$		\\
\hline
Impedance matrix								&		Dyadic Green function with constant 					\\
$\obf{Z}$									
&		
$-i\omega\mu_0\obf{G}$																					\\
\hline
Mutual resistance								&		Imaginary part of dyadic Green function with constant 	\\
$\Re\{\overline{\mathbf{Z}}\}$					&		$\Im\{\omega\mu_0\overline{\mathbf{G}}\}$ \\
\hline
Average power consumption						&		Average power consumption							\\
$\frac{1}{2}\mathbf{i}^H\Re\{\overline{\mathbf{Z}}\}\mathbf{i}$
&			
$\frac{1}{2}\left\langle\mathbf{J},\Im\{\omega\mu_0\obf{G}\},\mathbf{J}\right\rangle_\mathcal{R}$\\
\hline
Johnson-Nyquist noise  						&		FDT					\\
$E\{\mathbf{v}_{\mathrm{noise}}\mathbf{v}_{\mathrm{noise}}^H\}=4k_B T B \Re\{\overline{\mathbf{Z}}\}$
&
$E\{\mathbf{E}_{\mathrm{noise}}(\mathbf{r})\mathbf{E}_{\mathrm{noise}}^H(\mathbf{r}')\}=4k_B T B  \Im\{\omega\mu_0\overline{\mathbf{G}}(\mathbf{r},\mathbf{r}')\}$
\end{tabular}
\caption{Comparison between the frameworks of discrete-space analysis and continuous-space analysis in this paper}\label{table1}
\end{center}
\vspace{-0.4in}
\end{table}

%---------------------------------------------------------------------------------------------------------------------------
\section{Main result}\label{sec_single}

Assume the concentric spheres $V$ and $S$ with radius $R_1$ and $R_2$, and the property of media, i.e., wave numbers, are defined in the same way as that in the previous section. 
Also, assume $K$ different points, $\mathbf{s}_1,...,\mathbf{s}_K$, are distributed on $S$.
For the proof of our theorem, we use the following definition to select $K$ points on $S$:
\begin{definition}[uniform distribution]
A sequence of sets of angular positions $\boldsymbol{\Theta}_K\triangleq\{\theta_{K,j}\in
\boldsymbol{\Omega}\triangleq\left[0,\pi\right]\times\left[0,2\pi\right)| j=1,...,K\}, K\in\mathbb{Z}^+$ is said to be \emph{uniformly distributed} if
\begin{align}
\frac{|\boldsymbol{\Phi}|}{4\pi}=\lim_{K\rightarrow \infty}\frac{\sum_{j=1}^K\mathbb{I}\{\theta_{K,j}\in\boldsymbol{\Phi}\}}{K},
\end{align}
for all conic solid angles $\boldsymbol{\Phi}\subseteq\boldsymbol{\Omega}$, where $|\boldsymbol{\Phi}|\triangleq\int_{\boldsymbol{\Phi}}d\Omega$ and $\mathbb{I}$ is the indicator function, and for $\beta>0$, 
\begin{align}
d(\Omega_{K,j_1},\Omega_{K,j_2})\ge\frac{\beta}{\sqrt{K}}, j_1,j_2=1,...,K, 
\end{align}
for all $j_1\ne j_2$ and for all $K$, where $d(\cdot,\cdot)$ is the Euclidean distance between two angular positions on the unit sphere. 
\end{definition}
Any $\beta<\sqrt{8\pi/\sqrt{3}}$ is expected to work for our definition~\cite{clare1991optimal}.
For uniformly distributed $\boldsymbol{\Theta}_1,\boldsymbol{\Theta}_2,...$, let us define $\mathbf{s}_{K,j}\triangleq(R_2,\theta_{K,j})$ for $\theta_{K,j}\in\boldsymbol{\Theta}_K$ for all $K,j$, where $k_0R_2\triangleq\sqrt{K/\alpha}$ for some $0<\alpha<1$.
Here, the constant $\alpha$ is related to the density of the sampled points on $S$.
Note that if we increase the number $K$ of sampled points for given sampling density $\alpha$, the radius $R_2$ of $S$ also increases.
Also, the electromagnetic interaction among the sampled points $\mathbf{s}_{K,j}$'s becomes negligible by choosing sufficiently small sampling density $\alpha$.
If there is no confusion, we will use $\mathbf{s}_j$ instead of $\mathbf{s}_{K,j}$.
In this section, we analyze the capacity of following two channels.
\begin{figure}[t]
\begin{center}
\begin{tikzpicture}
\node at (0,0) {\includegraphics[page=3,height=2in]{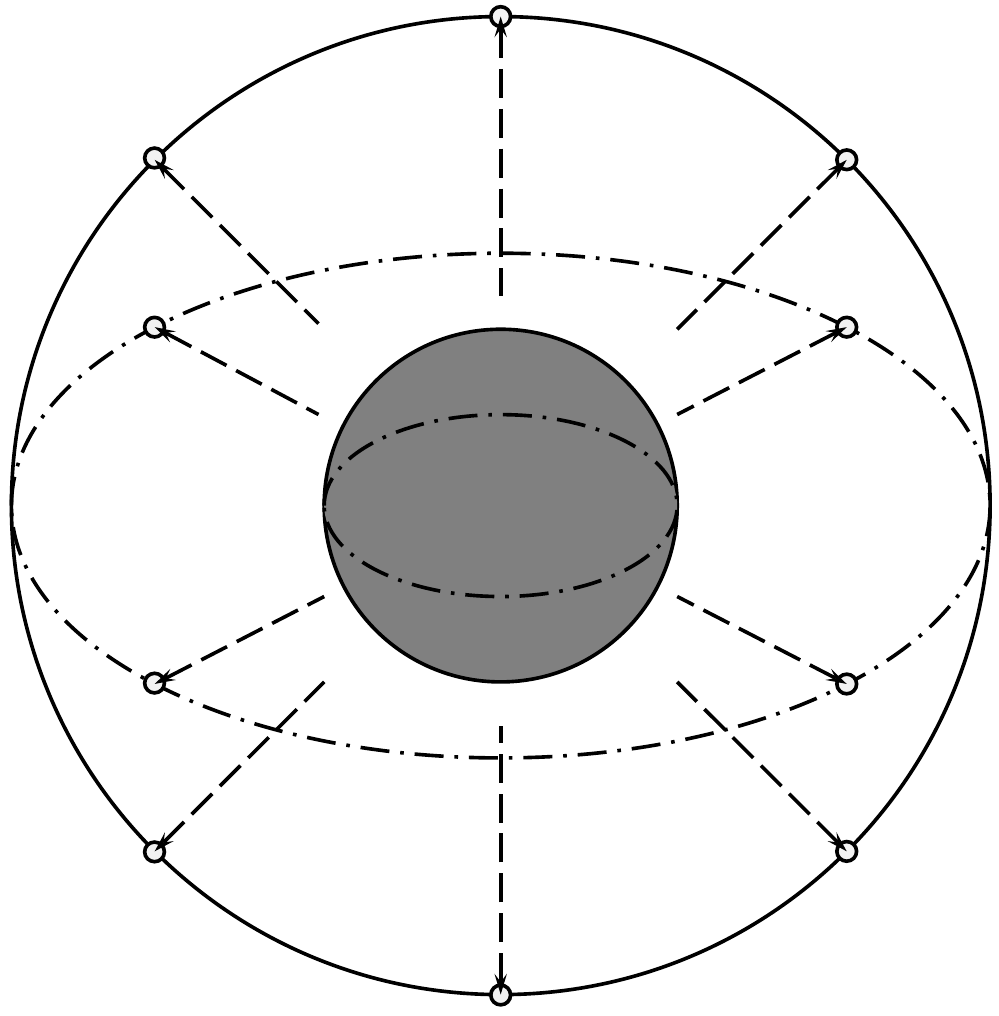}};
\node at (0,0) {$V$};
\node at (0.8in,0.8in) {$S$};
\end{tikzpicture}
\begin{tikzpicture}
\node at (0,0) {\includegraphics[page=2,height=2in]{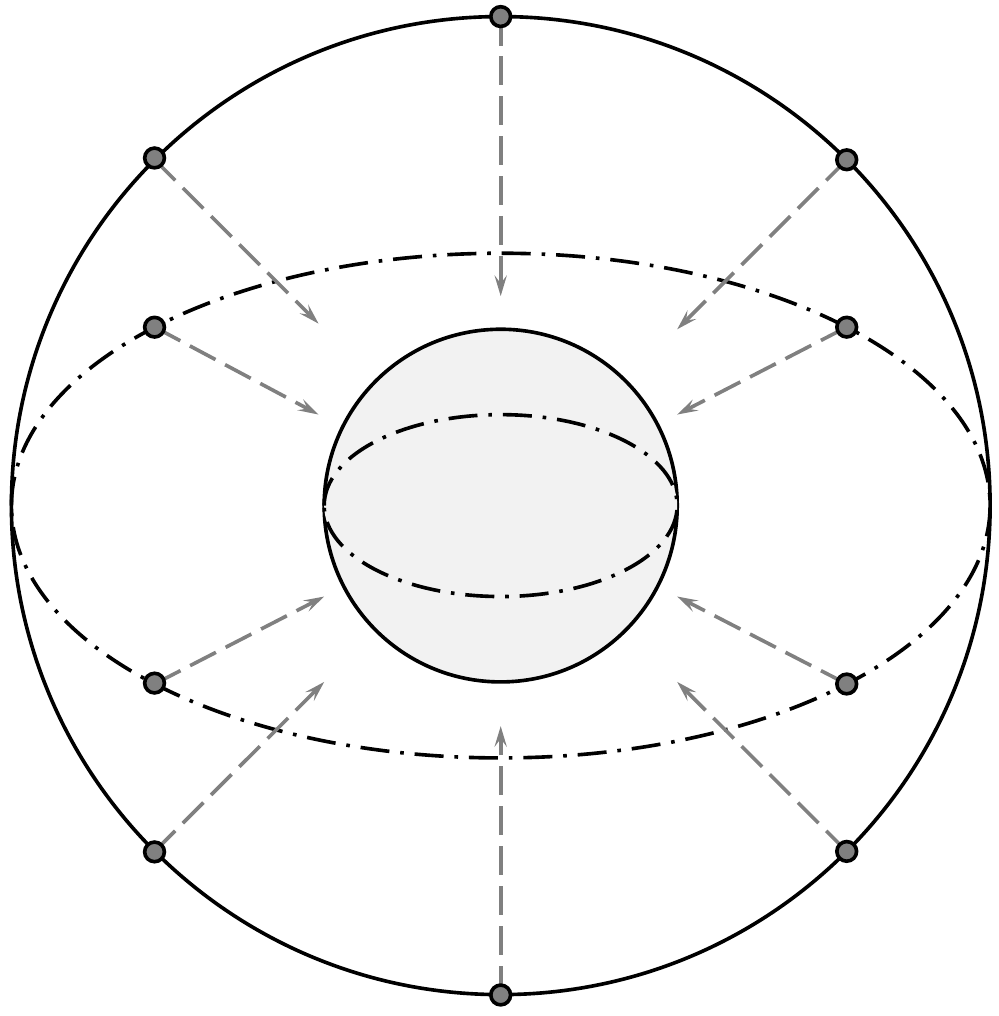}};
\node at (0,0) {$V$};
\node at (0.8in,0.8in) {$S$};
\end{tikzpicture}
\vspace{-0.2in}
\caption{Illustration on the single-user channel: a forward channel (left) and a reverse channel (right)}\label{figure2page2}
\end{center}
\vspace{-0.4in}
\end{figure}
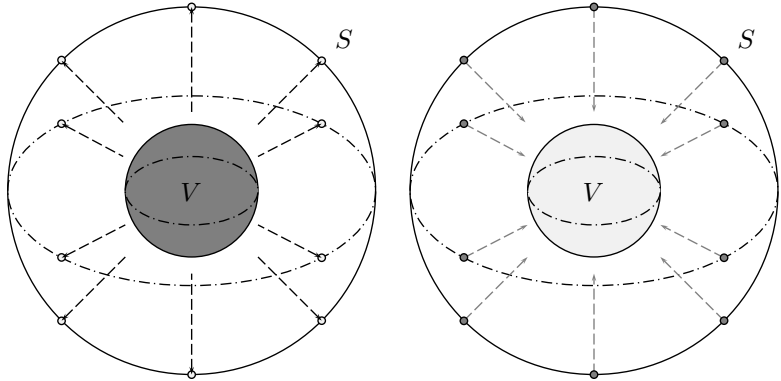
\begin{itemize}
\item
In the \emph{forward channel}, the transmitter is allowed to generate the source current inside $V$ that satisfies the power constraint $\mathrm{P}$. Also, the receiver measures the electric field on $\mathbf{s}_1,...,\mathbf{s}_K$, where the electric field is the sum of the field due to the source and the thermal noise. The transmitter is allowed to use the spherical waves with order $n\le N$. 
\item
In the \emph{reverse channel}, the transmitter is allowed to generate the source current on $\mathbf{s}_1,...,\mathbf{s}_K$ under the power constraint $\mathrm{P}$. Also, the receiver measures the electric field inside $V$, where the electric field is the sum of the electric field generated from the source and the thermal noise. The receiver is allowed to use the spherical waves with order $n\le N$. 
\end{itemize}
In Section~\ref{sec_single_equivalent}, we derive the capacity of the forward and the reverse channels.
During our analysis, the electromagnetic power consumption and the statistical property of the thermal noise based on FDT are utilized, which are introduced in the previous section.
As a result, we show the capacity of each channel is determined by the efficiency of the orthogonally decomposed channels.
In Section~\ref{sec_single_lossless}, the forward and the reverse channels are considered when the dielectric sphere is assumed to be lossless, which can be regarded as the special case of~\ref{sec_single_equivalent}.
For this case, it can be shown that the capacity does not depend on the size of the spherical region.

%---------------------------------------------------------------------------------------------------------------------------
\subsection{Capacity analysis}\label{sec_single_equivalent}

The capacity of both forward and reverse channels is given as follows:
\begin{definition}[efficiency]
The efficiency $\eta_{nml}$ of the channel with index $(n,m,l)$ is defined as
\begin{align}
\eta_{nml}
\triangleq
\frac{\rho_{nml}}{\tau_{nml}},
\end{align}
where
\begin{align}
\rho_{nml}
\triangleq
\frac{I_{n,l}^{jj^*}|k_1\mathcal{T}_{n,l}|^2}{k_0},
\tau_{nml}
\triangleq
\dfrac{\Re\{\mathcal{F}_{n,l}I_{n,l}^{jj}\}}{I_{n,l}^{jj^*}}+\dfrac{1}{4k_1'k_1''},
\forall n,m,l
\end{align}
for $\mathcal{F}_{n,l}(k)\triangleq \mathcal{D}_{n,l}(k)-\mathcal{E}_{n,l}(k)/(4k'k'')$ i.e., 
\begin{align}
\mathcal{D}_{n,l}(k)
&\triangleq
k[(1+\mathcal{R}_{n,l})I_{n,l}^{jj^*}(k,R_1)+iI_{n,l}^{yj^*}(k,R_1)],\\
\mathcal{E}_{n,l}(k)
&\triangleq
\begin{cases}
k^{-n}(k^*)^n,
&\text{~if~}l=1,\\
\dfrac{n+1}{2n+1}\mathcal{E}_{n-1,1}(k)+\dfrac{n}{2n+1}\mathcal{E}_{n+1,1}(k),
&\text{~if~}l=2,
\end{cases}
\end{align}
for all $n,l$ and the argument $k_1$ of $\mathcal{F}_{n,l}$ is omitted for simplicity.
$I_{n,l}^{jj^*}$ and $I_{n,l}^{yj^*}$ are defined in Appendix~\ref{sec_appendix_definitions}, and those without arguments are assumed to have their arguments as $(k_1,R_1)$. 
\end{definition}
\begin{theorem}\label{sec_single_equivalent_theorem}
For both the forward channel and the reverse channel defined as above, the capacity of each channel is equal to
\begin{align}
\sum_{(n,m,l)\in\boldsymbol{\Upsilon}}
\log 
\left(1+p_{nml}h_{nml}^2\right)
\end{align}
for $\boldsymbol{\Upsilon}\triangleq\{(n,m,l)\in \mathbb{Z}^3:1\le n \le N, -n\le m\le n, l=1,2\}$, where $h_{nml}\triangleq\sqrt{3\alpha}\sqrt{\frac{\eta_{nml}}{4k_BTB}}$ for all $n,m,l$ and 
\begin{align*}
p_{nml}\triangleq
\max\left(\frac{1}{\lambda}-\frac{1}{h_{nml}^2},0\right)
\end{align*}
for $\lambda>0$ satisfying $\sum_{(n,m,l)\in\boldsymbol{\Upsilon}}p_{nml}=\mathrm{P}$. 
\end{theorem}
\begin{remark}
For any $n,l$, the efficiency $\eta_{nml}$ decreases rapidly after some threshold as $n$ increases when the lossy dielectric is used, i.e., $k_1''\ne 0$. 
Thus, we can always find an integer $N_0$ such that the capacity in \textbf{Theorem \ref{sec_single_equivalent_theorem}} increases for $N<N_0$ and remains as a constant for $N\ge N_0$ as $N$ increases.
\end{remark}
The wave numbers of the spherical waves and the normalized spherical vector waves are omitted if there is no confusion. Note that $m$ is chosen to be zero since the normalization coefficients only depend on the order $n$ and $l$ of the spherical vector waves. 

%---------------------------------------------------------------------------------------------------------------------------
\subsubsection{Forward channel}\label{sec_single_equivalent_forward}
We derive the capacity of the forward channel as follows.
~\newline
\textbf{\emph{Transmitter.}}
Assume the current source $\mathbf{J}$ in $V$, which can be represented as
\begin{align}
\mathbf{J}
&=
\sum_{(n,m,l)\in\boldsymbol{\Upsilon}} J_{nml}\mathbf{v}_{nml}^\star,
\end{align}
where $J_{nml}\triangleq\left\langle\mathbf{v}_{nml}^\star,\mathbf{J}\right\rangle_V, \forall n,m,l$. 
From Section \ref{sec_preliminaries_dyadic}, DGF from $\mathbf{r}'\in V$ to $\mathbf{r}\in S$ can be decomposed as
\begin{align}
\obf{G}(\mathbf{r},\mathbf{r}')
=
ik_1\sum_{n,m,l}
\frac{\mathcal{T}_{n,l}\mathcal{N}_{n,l}^S\mathcal{N}_{n,l}^V}{n(n+1)}
\mathbf{u}_{nml}(\mathbf{r})
\mathbf{v}_{nml}^\star(\mathbf{r}')^H.
\end{align}
By using the decomposition of DGF and the definition of the source, the electric field on $S$ generated from the source is 
\begin{align}
\mathbf{E}
%&=
%-\omega\mu_0k_1\sum_{(n,m,l)\in\mathcal{M}(N)}
%\mathbf{u}_{nml}
%\\
%&\times\left[
%\frac{\mathcal{T}_{n,l}^{f}N_{S,nl}N_{V,nl}}{n(n+1)}
%\left\langle\mathbf{v}_{nml}^\star,\mathbf{v}_{nml}^\star\right\rangle_VJ_{nml}
%\right]\\
%&=
%\sum_{(n,m,l)\in\mathcal{M}(N)}
%\left[
%-\omega\mu_0k_1
%\frac{\mathcal{T}_{n,l}^{f}N_{S,nl}N_{V,nl}}{n(n+1)}
%\frac{\left\langle\mathbf{V}_{nml}^\star,\mathbf{V}_{nml}\right\rangle_V}{N_{V,nl}^2}
%J_{nml}
%\right]
%\mathbf{u}_{nml}\\
&=
\sum_{\boldsymbol{\Upsilon}}\mathcal{G}_{n,l}J_{nml}\mathbf{u}_{nml}
\label{single_equivalent_forward_1}
\end{align}
for $\mathcal{G}_{n,l}\triangleq-\omega\mu_0k_1\mathcal{T}_{n,l}\sqrt{I_{n,l}^{jj^*}}\mathcal{N}_{n,l}^S/\sqrt{n(n+1)}$,
where the equalities follow by using the orthogonality of the spherical vector waves and $\mathcal{N}_{n,l}^V=\sqrt{n(n+1)I_{n,l}^{jj^*}}$. 
~\newline
\textbf{\emph{Receiver.}}
Suppose the receiver measures the electric field on $\mathbf{s}_1,...,\mathbf{s}_K$, which is the sum of the electric field generated from the source and the thermal noise, i.e.,
\begin{align}
\mathbf{E}(\mathbf{s}_j)+\mathbf{E}_{\mathrm{noise}}(\mathbf{s}_j),j=1,...,K.
\end{align}
Then, the channel output is defined by using the electric field on $\mathbf{s}_1,...,\mathbf{s}_K$ as
\begin{align}
E_{nml}
&\triangleq
\sqrt{\frac{4\pi}{K}}
\sum_{j=1}^K
\mathbf{u}_{nml}(\mathbf{s}_j)^H[\mathbf{E}(\mathbf{s}_j)+\mathbf{E}_{\mathrm{noise}}(\mathbf{s}_j)]
\end{align}
for all $(n,m,l)\in\boldsymbol{\Upsilon}$. Let us define
\begin{align}
E_{nml,\mathrm{signal}}
&\triangleq
\sqrt{\frac{4\pi}{K}}
\sum_{j=1}^K
\mathbf{u}_{nml}(\mathbf{s}_j)^H
\mathbf{E}(\mathbf{s}_j),\\
E_{nml,\mathrm{noise}}
&\triangleq
\sqrt{\frac{4\pi}{K}}
\sum_{j=1}^K
\mathbf{u}_{nml}(\mathbf{s}_j)^H
\mathbf{E}_{\mathrm{noise}}(\mathbf{s}_j),
\end{align}
which represent the contributions of the signal and the noise electric field on the channel output, respectively.
Then, $E_{nml,\mathrm{signal}}$ is equal to
\begin{align}
\sum_{(\tilde{n},\tilde{m},\tilde{l})\in\boldsymbol{\Upsilon}}\mathcal{G}_{\tilde{n},\tilde{l}}J_{\tilde{n}\tilde{m}\tilde{l}}
\left[
\sqrt{\frac{4\pi}{K}}\sum_{j=1}^K\mathbf{u}_{nml}(\mathbf{s}_j)^H\mathbf{u}_{\tilde{n}\tilde{m}\tilde{l}}(\mathbf{s}_j)
\right]
\end{align}
for all $(n,m,l)\in\boldsymbol{\Upsilon}$.
For the approximation of $E_{nml,\mathrm{signal}}$ for sufficiently large $K$, the following two approximations is used. 
The first approximation is 
\begin{align}
\sqrt{\frac{4\pi}{K}}\sum_{j=1}^K\mathbf{u}_{nml}(\mathbf{s}_j)^H\mathbf{u}_{\tilde{n}\tilde{m}\tilde{l}}(\mathbf{s}_j)
\approx
\sqrt{\frac{K}{4\pi}}
\left\langle\mathbf{u}_{nml},\mathbf{u}_{\tilde{n}\tilde{m}\tilde{l}}\right\rangle_S
=
\sqrt{\frac{K}{4\pi}}
\delta_{n,\tilde{n}}
\delta_{m,\tilde{m}}
\delta_{l,\tilde{l}}
\end{align}
for sufficiently large $K$, where $\delta_{a,b}=1$ if $a=b$ and $0$ if $a\ne b$. This follows since the angular positions of $\mathbf{s}_1,...,\mathbf{s}_K$ are uniformly distributed, and the summation can be approximately equal to the integral over all angular positions. 
The second approximation is done for $\mathcal{N}_{n,l}^S$ contained in $\mathcal{G}_{n,l}$, i.e., 
\begin{align}
\mathcal{N}_{n,l}^S\approx\sqrt{\frac{\alpha n(n+1)}{K}}
\end{align}
for all $(n,m,l)\in\boldsymbol{\Upsilon}$, where the approximation holds for $K\gg \alpha N^2$
since $\lvert h_n^{(1)}(k_0R_2)\rvert\approx 1/(k_0R_2)$ for $k_0R_2\gg n$ and $k_0R_2=\sqrt{K/\alpha}$ from the definition of the uniform distribution. 
%
%In addition, from Appendix~\ref{sec_appendix_definitions},
%\begin{align}
%N_{S,\mathbf{U}_{nml}}(k_0)^2
%\approx \frac{n(n+1)}{(k_0R_2)^2}, \forall (n,m,l)\in\mathcal{M}(N)
%\end{align}
%where the approximation holds for $k_0R_2\gg N$, since
%\begin{align}
%\left\lvert h_n^{(1)}(k_0R_2)\right\rvert\approx \frac{1}{k_0R_2}\text{~for~}k_0R_2\gg n.
%\end{align}
In other words, \emph{the approximation for $\mathcal{N}_{n,l}^S$ holds when $\mathbf{s}_1,...,\mathbf{s}_K$ are sufficiently far from the reactive near-field region of the spherical vector waves with order $N$.} 
As a result, by using those two approximations, the channel output is approximately equal to
\begin{align}
\mathcal{H}_{n,l}J_{nml}+E_{nml,\mathrm{noise}}, \forall (n,m,l)\in\boldsymbol{\Upsilon},
\end{align}
where $\mathcal{H}_{n, l}\triangleq-\omega\mu_0k_1\mathcal{T}_{n,l}\sqrt{\alpha I_{n,l}^{jj^*}/(4\pi)}$.
~\newline
\textbf{\emph{Power consumption.}}
The electric field generated from the source $\mathbf{v}_{nml}^\star$ is derived in Appendix~\ref{sec_appendix_field}. 
Then, by using the orthogonality of $\mathbf{A}_{nm}^{(1)}, \mathbf{A}_{nm}^{(2)}, \mathbf{A}_{nm}^{(3)}$, the complex power of the source $\mathbf{v}_{nml}^\star$ is equal to
\begin{align}
-\frac{1}{2}\left\langle\mathbf{v}_{nml}^\star,i\omega\mu_0\obf{G},\mathbf{v}_{nml}^\star\right\rangle_V
=
\frac{\omega\mu_0}{2}
\left(
\frac{\mathcal{F}_{n,l}I_{n,l}^{jj}}{I_{n,l}^{jj^*}}
+
\frac{1}{4k_1'k_1''}
\right).
\end{align}
The power consumption is the real part of the complex power, which is equal to
\begin{align}
\frac{\omega\mu_0}{2}\tau_{nml},
\end{align}
where $\tau_{nml}$ is defined in \textbf{Theorem~\ref{sec_single_equivalent_theorem}}.
Since the complex power due to the source $\mathbf{J}=\sum_{\boldsymbol{\Upsilon}}J_{nml}\mathbf{v}_{nml}^\star$ is equal to
\begin{align}
-\frac{1}{2}
\left\langle\mathbf{J},i\omega\mu_0\obf{G},\mathbf{J}\right\rangle_V
&=
-\frac{1}{2}
\sum_{\boldsymbol{\Upsilon}}
|J_{nml}|^2
\left\langle\mathbf{v}_{nml}^\star,i\omega\mu_0\obf{G},\mathbf{v}_{nml}^\star\right\rangle_V,
\end{align}
from the orthogonality of the spherical vector waves, the power consumed by the source $\mathbf{J}$ is
\begin{align}
\frac{\omega\mu_0}{2}
\sum_{\boldsymbol{\Upsilon}}|J_{nml}|^2
\tau_{nml}.
\end{align}
Here, the power consumption averaged over $J_{nml}$'s should be bounded above by the power constraint $\mathrm{P}$. 
~\newline
\textbf{\emph{Radiation power.}}
From~\cite{harrington1961time}, the radiation power is equal to
\begin{align}
\frac{1}{2}\oint_S\Re\{\mathbf{E}\times\mathbf{H}^*\}^T\hat{\mathbf{r}}R_2^2	d\Omega,
\end{align}
where $\mathbf{H}$ is the magnetic field generated from the source.
Since the radiated power should be the same for all $R_2>R_1$, it is equal to
\begin{align}
\frac{1}{2\mathrm{Z}_0}\lim_{R_2\rightarrow\infty}R_2^2\left\langle\mathbf{E},\mathbf{E}\right\rangle_S
\end{align}
for the free-space wave impedance $\mathrm{Z}_0\triangleq\sqrt{\mu_0/\epsilon_0}$.
By using the orthogonality of $\mathbf{u}_{nml}$'s, the radiation power is equal to
\begin{align}
\frac{1}{2\mathrm{Z}_0}
\sum_{\boldsymbol{\Upsilon}}|J_{nml}|^2
\frac{\omega^2\mu_0^2I_{n,l}^{jj^*}|k_1\mathcal{T}_{n,l}|^2}{n(n+1)}\lim_{R_2\rightarrow\infty}(R_2\mathcal{N}_{n,l}^S)^2.
\end{align}
Finally, since the limit is equal to$n(n+1)/(k_0R_2)^2$ and $\omega\mu_0=\mathrm{Z}_0k_0$,
the power radiated by the source $\mathbf{J}$ is equal to 
\begin{align}
\frac{\omega\mu_0}{2}\sum_{\boldsymbol{\Upsilon}}|J_{nml}|^2\rho_{nml}, 
\end{align}
where $\rho_{nml}$ is defined in \textbf{Theorem~\ref{sec_single_equivalent_theorem}}. 
Note that the radiation efficiency~\cite{balanis2005antenna} of the spherical vector wave with order $(n,m,l)$ is equal to 
$\eta_{nml}$, which is also defined in the theorem. 
~\newline
\textbf{\emph{Noise statistics.}}
The noise $E_{nml,\mathrm{noise}}$ satisfies
\begin{align}
E\{E_{nml,\mathrm{noise}}E_{n'm'l',\mathrm{noise}}^*\}
%&=
%\frac{4\pi}{K}
%\sum_{j=1}^K\sum_{\tilde{j}=1}^K
%\mathbf{u}_{nml}(\mathbf{s}_j)^H
%E\{
%\mathbf{E}_{\mathrm{noise}}(\mathbf{s}_j)
%\mathbf{E}_{\mathrm{noise}}(\mathbf{s}_{\tilde{j}})^H
%\}
%\mathbf{u}_{\tilde{n}\tilde{m}\tilde{l}}(\mathbf{s}_{\tilde{j}})\\
&=
4k_BTB\omega\mu_0
\left[
\frac{4\pi}{K}
\sum_{j=1}^K\sum_{\tilde{j}=1}^K
\mathbf{u}_{nml}(\mathbf{s}_j)^H
\Im\{\obf{G}(\mathbf{s}_j,\mathbf{s}_{\tilde{j}})\}
\mathbf{u}_{\tilde{n}\tilde{m}\tilde{l}}(\mathbf{s}_{\tilde{j}})
\right],
\end{align}
where the equality holds by using the noise statistics in \textbf{Corollary \ref{sec_preliminaries_thermal_corollary}}. 
Since $\mathbf{s}_j$'s are uniformly distributed, the electromagnetic interaction among those sampled points is ignorable and the following approximation holds:
\begin{align}
\Im\{\obf{G}(\mathbf{s}_j,\mathbf{s}_{j'})\}
&\approx
\Im\{\obf{G}_0(\mathbf{s}_j,\mathbf{s}_{j})\}\delta_{j,j'}=
\dfrac{k_0}{6\pi}\obf{I}\delta_{jj'},
\end{align}
where $\obf{G}_0$ is the free-space DGF and $\obf{I}$ is the identity operator~\cite{novotny2006principles}. 
Thus, the statistics can be approximated as 
\begin{align}
E\{E_{nml,\mathrm{noise}}E_{n'm'l',\mathrm{noise}}^*\}
&\approx
4k_BTB
\frac{\omega\mu_0k_0}{6\pi}
\left[
\frac{4\pi}{K}
\sum_{j=1}^K
\mathbf{u}_{nml}(\mathbf{s}_j)^H
\mathbf{u}_{\tilde{n}\tilde{m}\tilde{l}}(\mathbf{s}_j)
\right].
\end{align}
Also, by using the condition that $\mathbf{s}_j$'s are uniformly distributed, it can be again approximated as 
\begin{align}
4k_BTB \frac{\omega\mu_0k_0}{6\pi}
\left\langle\mathbf{u}_{nml},\mathbf{u}_{\tilde{n}\tilde{m}\tilde{l}}\right\rangle_S=
4k_BTB \frac{\omega\mu_0k_0}{6\pi}
\delta_{n,\tilde{n}}\delta_{m,\tilde{m}}\delta_{l,\tilde{l}}.
\end{align}
~\newline
\textbf{\emph{Information-theoretically equivalent channel.}}
The channel output, the power consumption and radiation from the source and the noise statistics were derived in the previous section for sufficiently large $K$, and the result is summarized as follows.
The channel output $E_{nml}$ is equal to
\begin{align}
E_{nml}=
\mathcal{H}_{n,l}J_{nml}+E_{nml,\mathrm{noise}},
\end{align}
for all $(n,m,l)\in\boldsymbol{\Upsilon}$ and $\mathcal{H}_{n,l}=-\omega\mu_0k_1\mathcal{T}_{n,l}\sqrt{\alpha I_{n,l}^{jj^*}/(4\pi)}$, where the
where the power constraint is
\begin{align}
\frac{\omega\mu_0}{2}\sum_{\boldsymbol{\Upsilon}}|J_{nml}|^2\tau_{nml}\le
\mathrm{P},
\end{align}
and the noise follows
\begin{align}
E\{E_{nml,\mathrm{noise}}E_{\tilde{n}\tilde{m}\tilde{l},\mathrm{noise}}^*\}=4k_BTB\frac{\omega\mu_0k_0}{6\pi}\delta_{n,\tilde{n}}\delta_{m,\tilde{m}}\delta_{l,\tilde{l}}.
\end{align}
Let the input, the output and the noise of the channel be
\begin{align}
X_{nml}\triangleq
\sqrt{\omega\mu_0\tau_{nml}}J_{nml},
Y_{nml}\triangleq
\sqrt{\frac{6\pi}{\omega\mu_0k_0}}E_{nml},
Z_{nml}\triangleq
\sqrt{\frac{6\pi}{\omega\mu_0k_0}}E_{nml,\mathrm{noise}}
\end{align}
for all $(n,m,l)\in\boldsymbol{\Upsilon}$. 
Then, the channel output satisfies
\begin{align}
Y_{nml}=\sqrt{\frac{3\alpha}{2}}\left(-k_1\mathcal{T}_{n,l}\sqrt{\frac{I_{n,l}^{jj^*}}{k_0\tau_{nml}}}\right)X_{nml}+Z_{nml},
\end{align}
where $(1/2)\sum_{\boldsymbol{\Upsilon}}|X_{nml}|^2\le\mathrm{P}$ and $E\{Z_{nml}Z_{\tilde{n}\tilde{m}\tilde{l}}^*\}=4k_BTB\delta_{n,\tilde{n}}\delta_{m,\tilde{m}}\delta_{l,\tilde{l}}$.
Without loss of generality, the channel gain can be regarded as
\begin{align}
\sqrt{\frac{3\alpha}{2}}\left\lvert -k_1\mathcal{T}_{n,l}\sqrt{\frac{I_{n,l}^{jj^*}}{k_0\tau_{nml}}}\right\rvert
\end{align}
for circularly symmetric Gaussian noise~\cite{tse2005fundamentals}, and this is equal to
\begin{align}
\sqrt{\frac{3\alpha}{2}}\sqrt{\eta_{nml}}.
\end{align}
The derivation of the capacity of this channel is straightforward by using the waterfilling power allocation in multi-antenna channel~\cite{tse2005fundamentals}.

%---------------------------------------------------------------------------------------------------------------------------
\subsubsection{Reverse channel}\label{sec_single_equivalent_reverse}

Now, let us derive the capacity of the reverse channel. Here, we can find the exact duality between the forward and the reverse channel, i.e., the power consumption (the noise statistics) of the reverse channel is related to the noise statistics (the power consumption) of the forward channel. Due to such duality, it can be shown that the capacity of the reverse channel is equal to that of the forward channel, which is derived as follows:
~\newline
\textbf{\emph{Transmitter.}}
Assume $K$ point sources are located on $\mathbf{s}_1,...,\mathbf{s}_K$. The current density for such point sources is
\begin{align}
\mathbf{J}(\mathbf{r})=\sum_{j=1}^K \mathbf{d}_j\delta(\mathbf{r}-\mathbf{s}_j)
\end{align}
for some vector $\mathbf{d}_j,j=1,...,K$. 
From Section \ref{sec_preliminaries_dyadic}, the DGF from the source point $\mathbf{r}'\in S$ to $\mathbf{r}\in V$ can be decomposed as 
\begin{align}
\obf{G}(\mathbf{r},\mathbf{r}')
=
ik_0\sum_{nml}
\frac{\mathcal{C}\mathcal{T}_{n,l}\mathcal{N}_{n,l}^V \mathcal{N}_{n,l}^S}{n(n+1)}
\mathbf{v}_{nml}(\mathbf{r})
\mathbf{u}_{nml}^\star(\mathbf{r}')^H.
\end{align}
From the definition of the dyadic Green function in Section \ref{sec_preliminaries_dyadic}, the electric field in $V$ generated from the source is equal to
\begin{align}
\mathbf{E}_{\mathrm{signal}}
=
\sum_{nml}
\mathcal{G}_{n,l}
\left[
\sum_{j=1}^K \mathbf{u}_{nml}^\star(\mathbf{s}_j)^H\mathbf{d}_j
\right]
\mathbf{v}_{nml}
,
\mathbf{r}\in V,
\end{align}
where $\mathcal{G}_{n,l}=-\omega\mu_0k_0\mathcal{C}\mathcal{T}_{n,l}\mathcal{N}_{n,l}^V\mathcal{N}_{n,l}^S/(n(n+1)),\forall n,m,l$. Here, note that $\mathcal{G}_{n,l}$ is equal to that in the forward channel since $k_1=k_0\mathcal{C}$ and $\mathcal{N}_{n,l}^V=\sqrt{n(n+1)I_{n,l}^{jj^*}}$.
For the reverse channel, we assume the channel input $X_{nml}$ is related to $\mathbf{d}_j$ as 
\begin{align}
\mathbf{d}_j
=
\sqrt{\frac{6\pi}{\omega\mu_0k_0}}
\left[
\sqrt{\frac{4\pi}{K}}
\sum_{\boldsymbol{\Upsilon}}
\mathbf{u}_{nml}^\star(\mathbf{s}_j)X_{nml}
\right]
\label{eq:sec_single_equivalent_reverse_1}
\end{align}
for all $j=1,...,K$. Then, the electric field in $V$ due to the source is equal to 
\begin{align}
\sqrt{\frac{6\pi}{\omega\mu_0k_0}}
\sum_{nml}\mathcal{G}_{n,l}
\mathbf{v}_{nml}
\sum_{(\tilde{n},\tilde{m},\tilde{l})\in\boldsymbol{\Upsilon}}X_{\tilde{n}\tilde{m}\tilde{l}}
\left[
\sqrt{\frac{4\pi}{K}}
\sum_{j=1}^K \mathbf{u}_{nml}^\star(\mathbf{s}_j)^H
\mathbf{u}_{\tilde{n}\tilde{m}\tilde{l}}^\star(\mathbf{s}_j)
\right].
\end{align}
Due to the condition that $\mathbf{s}_1,...,\mathbf{s}_K$ are uniformly distributed and the approximation for $\mathcal{G}_{n,l}$ outside the reactive near-field region, the electric field is approximately the same as
\begin{align}
\sqrt{\frac{6\pi}{\omega\mu_0k_0}}
\sum_{\boldsymbol{\Upsilon}}
\mathcal{H}_{n,l}X_{nml}
\mathbf{v}_{nml}
\end{align}
for sufficiently large $K$. Note that the approximation is similar to that in the forward channel. 
~\newline
\textbf{\emph{Receiver.}}
Suppose the receiver measures the electric field in $V$, which is the sum of the electric field generated from the source and the thermal noise, i.e.,
\begin{align}
\mathbf{E}_{\mathrm{signal}}+\mathbf{E}_{\mathrm{noise}}.
\end{align}
The channel output is defined by using the values of the electric field in $V$ as
\begin{align}
E_{nml}
\triangleq
\left\langle\mathbf{v}_{nml},\mathbf{E}_{\mathrm{signal}}+\mathbf{E}_{\mathrm{noise}}\right\rangle_V
\approx
\sqrt{\frac{6\pi}{\omega\mu_0k_0}}\mathcal{H}_{n,l}X_{nml}+E_{nml,\mathrm{noise}}
\end{align}
for all $(n,m,l)\in\boldsymbol{\Upsilon}$, where $E_{nml,\mathrm{noise}}\triangleq\left\langle\mathbf{v}_{nml},\mathbf{E}_{\mathrm{noise}}\right\rangle_V$ for all $n,m,l$.
~\newline
\textbf{\emph{Power consumption.}}
Let $V_t$ be the source region including the points $\mathbf{s}_1,...,\mathbf{s}_K$.
Then, the power consumption at the source is equal to 
\begin{align}
\frac{\omega\mu_0}{2}
\left\langle\mathbf{J},\Im\{\obf{G}\},\mathbf{J}\right\rangle_{V_t}
=
\frac{\omega\mu_0}{2}
\sum_{j=1}^K\sum_{\tilde{j}=1}^K
\left\langle\mathbf{d}_j\Delta_j,\Im\{\obf{G}\},\mathbf{d}_{\tilde{j}}\Delta_{\tilde{j}}\right\rangle_{V_t}
\end{align}
for $\Delta_j(\mathbf{r})\triangleq\delta(\mathbf{r}-\mathbf{s}_j), j=1,...,K$. 
Due to the condition that $\mathbf{s}_1,...,\mathbf{s}_K$ are uniformly distributed, the electromagnetic interaction among those points are negligible i.e.,
\begin{align}
\Im\{\obf	{G}(\mathbf{s}_j,\mathbf{s}_{\tilde{j}})\}
\approx\Im\{\obf{G}_0(\mathbf{s}_j,\mathbf{s}_j)\}\delta_{j,\tilde{j}}
=\frac{k_0}{6\pi}\obf{I}\delta_{j,\tilde{j}}
\end{align}
for the free-space DGF $\obf{G}_0$. Thus, the power consumption is approximately equal to
\begin{align}
\frac{\omega\mu_0k_0}{12\pi}
\sum_{j=1}^K
\left\lVert\mathbf{d}_j\right\rVert^2.
\end{align}
By substituting \eqref{eq:sec_single_equivalent_reverse_1}, the approximation is equal to
\begin{align}
\frac{1}{2}\sum_{(n,m,l)\in\boldsymbol{\Upsilon}}
\sum_{(\tilde{n},\tilde{m},\tilde{l})\in\boldsymbol{\Upsilon}}
X_{nml}^*X_{\tilde{n}\tilde{m}\tilde{l}}\left[
\frac{4\pi}{K}
\sum_{j=1}^K
\mathbf{u}_{nml}^\star(\mathbf{s}_j)^H
\mathbf{u}_{\tilde{n}\tilde{m}\tilde{l}}^\star(\mathbf{s}_j)
\right].
\label{eq:sec_single_equivalent_reverse_2}
\end{align}
In addition, since the term in the bracket in~\eqref{eq:sec_single_equivalent_reverse_2} can be approximated as $\left\langle\mathbf{u}_{nml}^\star,\mathbf{u}_{\tilde{n}\tilde{m}\tilde{l}}^\star\right\rangle_S$ due to the condition of the uniform distribution, it is approximately equal to 
\begin{align}
\frac{1}{2}\sum_{\boldsymbol{\Upsilon}}|X_{nml}|^2.
\end{align}
~\newline
\textbf{\emph{Noise statistics.}}
From the noise model based on FDT, the noise $E_{nml}$ satisfies
\begin{align}
E\{E_{nml,\mathrm{noise}}E_{\tilde{n}\tilde{m}\tilde{l},\mathrm{noise}}^*\}
=
4k_BTB\omega\mu_0
\left\langle\mathbf{v}_{nml},\Im\{\obf{G}\},\mathbf{v}_{\tilde{n}\tilde{m}\tilde{l}}\right\rangle_V,
\end{align}
where the term in the bracket is equal to$-\frac{i}{2}\left(\left\langle\mathbf{v}_{nml},\obf{G},\mathbf{v}_{n'm'l'}\right\rangle_V-\left\langle\mathbf{v}_{nml},\obf{G}^*,\mathbf{v}_{n'm'l'}\right\rangle_V\right)$.
Since
\begin{align}
\left\langle\mathbf{v}_{nml},\obf{G},\mathbf{v}_{\tilde{n}\tilde{m}\tilde{l}}\right\rangle_V
=
\left\langle\mathbf{v}_{nml},\obf{G},\mathbf{v}_{nml}\right\rangle_V
\delta_{n,\tilde{n}}\delta_{m,\tilde{m}}\delta_{l,\tilde{l}},
\end{align}
and
\begin{align}
\left\langle\mathbf{v}_{nml},\obf{G}^*,\mathbf{v}_{\tilde{n}\tilde{m}\tilde{l}}\right\rangle_V
%&=
%\left(\left\langle\mathbf{v}_{nml}^*,\obf{G},\mathbf{v}_{n'm'l'}^*\right\rangle_V\right)^*\\
%&=
%\left(\left\langle\mathbf{v}_{nml}^*,\mathrm{S}(\obf{G})^T,\mathbf{v}_{n'm'l'}^*\right\rangle_V\right)^*\\
%&=
%\left(\left\langle\mathbf{v}_{n'm'l'},\obf{G},\mathbf{v}_{nml}\right\rangle_V^T\right)^*\\
%&=
%\left\langle\mathbf{v}_{n'm'l'},\obf{G},\mathbf{v}_{nml}\right\rangle_V^*\\
&=
\left\langle\mathbf{v}_{nml},\obf{G},\mathbf{v}_{nml}\right\rangle_V^*
\delta_{n,\tilde{n}}\delta_{m,\tilde{m}}\delta_{l,\tilde{l}},
\end{align}
which follows from the orthogonality among spherical vector waves, the noise statistics is equal to
\begin{align}
E\{E_{nml,\mathrm{noise}}E_{\tilde{n}\tilde{m}\tilde{l},\mathrm{noise}}^*\}=4k_BTB\omega\mu_0\Re\{\left\langle\mathbf{v}_{nml},-i\obf{G},\mathbf{v}_{nml}\right\rangle_V\}\delta_{n,\tilde{n}}\delta_{m,\tilde{m}}\delta_{l,\tilde{l}}.
\end{align}
As a result, it becomes
\begin{align}
4k_BTB\omega\mu_0\tau_{nml}\delta_{n,\tilde{n}}\delta_{m,\tilde{m}}\delta_{l,\tilde{l}}.
\end{align}
since $\left\langle\mathbf{v}_{nml},-i\obf{G},\mathbf{v}_{nml}\right\rangle_V=\left\langle\mathbf{v}_{nml}^\star,-i\obf{G},\mathbf{v}_{nml}^\star\right\rangle_V$ and its real part is equal to$\tau_{nml}$. The proof for $\left\langle\mathbf{v}_{nml},-i\obf{G},\mathbf{v}_{nml}\right\rangle_V=\left\langle\mathbf{v}_{nml}^\star,-i\obf{G},\mathbf{v}_{nml}^\star\right\rangle_V$  is available in Appendix~\ref{sec_appendix_noise}.
~\newline
\textbf{\emph{Information-theoretically equivalent channel.}}
In summary, the channel output for sufficiently large $K$ is approximately equal to
\begin{align}
E_{nml}=
\sqrt{\frac{6\pi}{\omega\mu_0k_0}}
\mathcal{H}_{n,l}X_{nml}+E_{nml,\mathrm{noise}},
\end{align}
for all $(n,m,l)\in\boldsymbol{\Upsilon}$ and $\mathcal{H}_{n,l}=-\omega\mu_0k_1\mathcal{T}_{n,l}\sqrt{\alpha I_{n,l}^{jj^*}/(4\pi)}$, where the
where the power constraint is
\begin{align}
\frac{1}{2}\sum_{\boldsymbol{\Upsilon}}|X_{nml}|^2\le\mathrm{P},
\end{align}
and the noise follows
\begin{align}
E\{E_{nml,\mathrm{noise}}E_{\tilde{n}\tilde{m}\tilde{l},\mathrm{noise}}^*\}
=
4k_BTB\omega\mu_0\tau_{nml}\delta_{n,\tilde{n}}\delta_{m,\tilde{m}}\delta_{l,\tilde{l}}.
\end{align}
Let the input, the output and the noise of the channel be
\begin{align}
Y_{nml}
\triangleq
\frac{1}{\sqrt{\omega\mu_0\tau_{nml}}}E_{nml},
Z_{nml}
\triangleq
\frac{1}{\sqrt{\omega\mu_0\tau_{nml}}}E_{nml,\mathrm{noise}}.
\end{align}
Then, the channel output of the channel is equal to
\begin{align}
Y_{nml}=\sqrt{\frac{6\pi}{\omega\mu_0k_0}}\frac{\mathcal{H}_{n,l}}{\sqrt{\omega\mu_0\tau_{nml}}} X_{nml}+Z_{nml},
\end{align}
where $(1/2)\sum_{\boldsymbol{\Upsilon}}|X_{nml}|^2\le\mathrm{P}$ and $E\{Z_{nml}Z_{n'm'l'}^*\}=4k_BTB\delta_{nml,n'm'l'}$. The absolute value of the channel gain of the channel is equal to
\begin{align}
\left\lvert\sqrt{\frac{3\alpha}{2}}\left(-k_1\mathcal{T}_{n,l}\sqrt{\frac{I_{n,l}^{jj^*}}{k_0\tau_{nml}}}\right)\right\rvert
=
\sqrt{\frac{3\alpha}{2}}\sqrt{\eta_{nml}},
\end{align}
which is the same as the gain of the forward channel derive before.
Similar to the forward channel, the derivation of the capacity of this channel is straightforward by using the waterfilling power allocation in multi-antenna channel~\cite{tse2005fundamentals}.

%---------------------------------------------------------------------------------------------------------------------------
\subsection{Lossless dielectric sphere}\label{sec_single_lossless}

Consider the forward channel and assume the sphere filled with lossless dielectric, i.e., $k_1''\rightarrow0$.
Assume there is the source in $V$, i.e., 
\begin{align}
\mathbf{J}=\sum_{\boldsymbol{\Upsilon}}J_{nml}\mathbf{v}_{nml},
\end{align}
where $J_{nml}\triangleq\left\langle\mathbf{v}_{nml},\mathbf{J}\right\rangle_V, \forall n,m,l$. Note that the orthogonal basis is $\mathbf{v}_{nml}$ since the spherical Bessel function $j_n(k_1r)$ of the first kind is a real-valued function for $r>0$ and $\mathbf{v}_{nml}^\star=\mathbf{v}_{nml}$ for lossless case. 
The electric field generated from the source $\mathbf{v}_{nml}$ can be derived similar to that for the lossy dielectric sphere, and is equal to
\begin{align}
\mathbf{E}_{nml}
&=
-\frac{\omega\mu_0}{\sqrt{n(n+1)I_{n,l}^{jj}}}
\bigg\{
\frac{k_1}{n(n+1)}
\mathbf{C}_{nml}
+
\frac{i}{k_1^2}[\hat{\mathbf{r}}^T \mathbf{V}_{nml}]\hat{\mathbf{r}}
\bigg\},
\end{align}
where $\mathbf{C}_{nml}\triangleq\mathbf{C}_{1,nml}+i\mathbf{C}_{2,nml}$ for
\begin{align}
\mathbf{C}_{1,nml}
&\triangleq
(1+\mathcal{R}_{n,l})\left\langle\mathbf{V}_{nml},\mathbf{V}_{nml}\right\rangle_V\mathbf{V}_{nml},
\\
\mathbf{C}_{2,nml}
&\triangleq
\left\langle\mathbf{V}_{nml},\mathbf{V}_{nml}\right\rangle_{V(r)}\mathbf{W}_{nml}
+
[
\left\langle\mathbf{W}_{nml},\mathbf{V}_{nml}\right\rangle_V
-
\left\langle\mathbf{W}_{nml},\mathbf{V}_{nml}\right\rangle_{V(r)}
]
\mathbf{V}_{nml},
\end{align}
where $V(r)$ is a sphere with radius $r$ centered at the origin. 
%By using the definitions and properties in Appendix~\ref{sec_appendix_definitions}, we have
%\begin{align}
%\mathbf{C}_{1,nml}(\mathbf{r})
%&=
%\begin{cases}%%%%%%%%%%
%\left(n(n+1)\right)^{3/2}
%\mathbf{A}_{nm}^{(1)}(\theta,\phi)
%\left(
%(1+\mathcal{R}_{n,1})I_{n,1}^{jj}j_n^{(1)}(k_1r)
%\right),
%&\text{~if $l=1$},\\
%\dfrac{\left(n(n+1)\right)^2}{2n+1}
%\mathbf{A}_{nm}^{(2)}(\theta,\phi)
%\left(
%(1+\mathcal{R}_{n,2})I_{n,2}^{jj}j_n^{(2)}(k_1r)
%\right)\\
%+
%\left(n(n+1)\right)^{3/2}
%\mathbf{A}_{nm}^{(3)}(\theta,\phi)
%\left(
%(1+\mathcal{R}_{n,2})I_{n,2}^{jj}j_n^{(3)}(k_1r)
%\right),
%&\text{~if $l=2$}.
%\end{cases}
%\end{align}
%and
%\begin{align}
%\mathbf{C}_{2,nml}(\mathbf{r})
%&=
%\begin{cases}%%%%%%%%%%
%\left(n(n+1)\right)^{3/2}
%\mathbf{A}_{nm}^{(1)}(\theta,\phi)
%\left(
%I_{n,1}^{yj}j_n^{(1)}(k_1r)
%+
%[I_{n,1}^{jj}(k_1,r)y_n^{(1)}(k_1r)-I_{n,1}^{yj}(k_1,r)j_n^{(1)}(k_1r)]
%\right),
%&\text{~if $l=1$},\\
%\dfrac{\left(n(n+1)\right)^2}{2n+1}
%\mathbf{A}_{nm}^{(2)}(\theta,\phi)
%\left(
%I_{n,2}^{yj}j_n^{(2)}(k_1r)
%+
%[I_{n,2}^{jj}(k_1,r)y_n^{(2)}(k_1r)-I_{n,2}^{yj}(k_1,r)j_n^{(2)}(k_1r)]
%\right)\\
%+
%\left(n(n+1)\right)^{3/2}
%\mathbf{A}_{nm}^{(3)}(\theta,\phi)
%\left(
%I_{n,2}^{yj}j_n^{(3)}(k_1r)
%+
%[I_{n,2}^{jj}(k_1,r)y_n^{(3)}(k_1r)-I_{n,2}^{yj}(k_1,r)j_n^{(3)}(k_1r)]
%\right),
%&\text{~if $l=2$},
%\end{cases}
%\end{align}
%where the argument $(k_1,R_1)$ of $I_{n,l}$'s are omitted if there is no confusion.
The power consumed by the source $\mathbf{J}$ is equal to 
\begin{align}
\Re\left\{-\frac{1}{2}\left\langle\mathbf{J},\sum_{\boldsymbol{\Upsilon}}\mathbf{E}_{nml}\right\rangle_V\right\}
=
\frac{1}{2}
\sum_{\boldsymbol{\Upsilon}}
|J_{nml}|^2
\Re
\left\{-\left\langle\mathbf{v}_{nml},\mathbf{E}_{nml}\right\rangle_V\right\},
\end{align}
due to the orthogonality of the spherical vector waves. Since the inner product $\left\langle\mathbf{v}_{nml},\mathbf{C}_{2,nml}\right\rangle_V$ and $\left\langle\mathbf{v}_{nml},[\hat{\mathbf{r}}^T\mathbf{V}_{nml}]\hat{\mathbf{r}}\right\rangle_V$ are real-valued, the power consumption is equal to
\begin{align}
\frac{\omega\mu_0}{2}
\sum_{\boldsymbol{\Upsilon}}
\frac{|J_{nml}|^2}{(\mathcal{N}_{n,l}^V)^2}
\frac{k_1}{n(n+1)}
\Re
\left\{\left\langle\mathbf{V}_{nml},
\mathbf{C}_{1,nml}\right\rangle_V\right\}.
\end{align}
Since $\left\langle\mathbf{V}_{nml},\mathbf{C}_{1,nml}\right\rangle_V=(1+\mathcal{R}_{n,l})(\mathcal{N}_{n,l}^V)^4$, the power consumption is equal to
\begin{align}
\frac{\omega\mu_0}{2}
\sum_{\boldsymbol{\Upsilon}}|J_{nml}|^2
k_1 I_{n,l}^{jj}\Re\{1+\mathcal{R}_{n,l}\},
\end{align}
since $(\mathcal{N}_{n,l}^V)^2=n(n+1)I_{n,l}^{jj}$.
In addition, the radiation power is equal to
\begin{align}
\frac{\omega\mu_0}{2}
\sum_{\boldsymbol{\Upsilon}}|J_{nml}|^2
\frac{I_{n,l}^{jj}k_1^2|\mathcal{T}_{n,l}|^2}{k_0}.
\end{align}
Therefore, the radiation efficiency of each spherical wave is
\begin{align}
\frac{\mathcal{C}|\mathcal{T}_{n,l}|^2}{\Re\{1+\mathcal{R}_{n,l}\}},
\end{align}
which is shown to be 1 for any real-valued $k_1$. The proof is in Appendix~\ref{sec_appendix_lossless_efficiency}. As a result, we derive the following theorem by using a similar approach as in the previous section:
\begin{theorem}\label{sec_single_lossless_corollary}
For both the forward channel and the reverse channel defined as above, the capacity of each channel is equal to
\begin{align}
\sum_{(n,m,l)\in\boldsymbol{\Upsilon}}
\log 
\left(1+\frac{\mathrm{P}}{2N(N+2)}h_{nml}^2\right)
\end{align}
for $\boldsymbol{\Upsilon}\triangleq\{(n,m,l)\in \mathbb{Z}^3:1\le n \le N, -n\le m\le n, l=1,2\}$, where $h_{nml}\triangleq\sqrt{\frac{3\alpha}{4k_BTB}}$ for all $n,m,l$. 
\end{theorem}

The theorem states that if there is no loss, $2N(N+2)$ channels are of equal quality. 
A similar result was reported in~\cite{gruber2008new} that considered the capacity for the free-space source region.
Such result is related to the superdirective antenna arrays~\cite{oseen1922einsteinsche,schelkunoff1943mathematical,bouwkamp1945problem,di1956directivity} that can achieve desired directivity irrespective of the size of the antenna array.
Thus, it is possible to generate extremely narrow beam and achieve arbitrarily high DoF if the source region is filled with the lossless medium.
However, there are some practical issues on superdirectivity~\cite{jensen2008capacity}, and this will be considered in the next section. 

\section{Other considerations}\label{sec_practical}

In Section~\ref{sec_practical_Q}, the Q factor, which is inversely proportional to the fractional bandwidth, is derived. 
Also, in Section~\ref{sec_practical_effect}, we argue the reactive near-field region should be carefully considered since the communication in that region may critically affect the power consumption of the source, which is the reason why the order of the spherical waves is bounded above as $n\le N$ in the previous section.

%---------------------------------------------------------------------------------------------------------------------------
\subsection{Q factor}\label{sec_practical_Q}

The Q factor of the source is defined as the ratio of the energy stored in the field to the power consumption~\cite{chu1948physical,collin1964evaluation,fante1969quality,mclean1996re,thal2006new,yaghjian2005impedance,hansen2012stored,hansen2014properties}.
The Q factor is practically important since the fractional bandwidth of the resonant antenna is inversely proportional to Q when Q is sufficiently larger than $1$. 
In this section, we derive the Q factor for the lossy dielectric sphere\footnote{The Q factor for the dielectric sphere was also derived in~\cite{hansen2012stored,hansen2014properties}. In those works, Q was derived by assuming the surface current sources and using the boundary condition.
On the other hand, we derive Q of the volume current sources by using the spherical vector wave expansion of DGF and show that the results is not the same as those in the previous works. 
}. 
Let us assume that the source $\mathbf{v}_{nml}^\star$ generates the electric field $\mathbf{E}_{nml}$ and the magnetic field $\mathbf{H}_{nml}$.
By using the electromagnetic field, the electric and magnetic energy stored in $V$ are defined as
\begin{align}
W_{nml}^{e,in}
&=
\frac{\epsilon_1'}{4},
\int_{0<r<R_1}\left\lVert\mathbf{E}_{nml}(\mathbf{r})\right\rVert^2d\mathbf{r},\\
W_{nml}^{m,in}
&=
\frac{\mu_0}{4}
\int_{0<r<R_1}\left\lVert\mathbf{H}_{nml}(\mathbf{r})\right\rVert^2d\mathbf{r}
\end{align}
for $\epsilon_1'\triangleq\Re\{\epsilon_1\}$.
Also, define $W_{nml}^{e,out}$ and $W_{nml}^{m,out}$ as the stored electric and magnetic energy outside $V$ except the contribution of radiated energy outside $V$~\cite{collin1964evaluation}.
Then, the Q factor is defined as 
\begin{align*}
Q_{nml}\triangleq\max\left\{Q_{nml}^m,Q_{nml}^e\right\},
\end{align*}
where
\begin{align}
Q_{nml}^m
\triangleq
\dfrac{2\omega(W_{nml}^{m,in}+W_{nml}^{m,out})}{P_{tot,nml}},
Q_{nml}^e
\triangleq
\dfrac{2\omega(W_{nml}^{e,in}+W_{nml}^{e,out})}{P_{tot,nml}}
\end{align}
for the power consumption $P_{tot,nml}$ of the source. By using the efficiency $\eta_{nml}$, 
\begin{align}
Q_{nml}=\eta_{nml}\tilde{Q}_{nml}
\end{align}
for $\tilde{Q}_{nml}\triangleq\max\{\tilde{Q}_{nml}^m,\tilde{Q}_{nml}^e\}$ called radiation Q factor, where 
\begin{align}
\tilde{Q}_{nml}^m
\triangleq
\dfrac{2\omega W_{nml}^{m,in}}{P_{rad,nml}}
+
\dfrac{2\omega W_{nml}^{m,out}}{P_{rad,nml}},
\tilde{Q}_{nml}^e
\triangleq
\dfrac{2\omega W_{nml}^{e,in}}{P_{rad,nml}}
+
\dfrac{2\omega W_{nml}^{e,out}}{P_{rad,nml}}
\end{align}
for the radiated power $P_{rad,nml}$ from the source.
From~\cite{collin1964evaluation}, we have
\begin{align}
\dfrac{2\omega W_{nm2}^{m,out}}{P_{rad,nml}}
&=\dfrac{2\omega W_{nm1}^{e,out}}{P_{rad,nml}}=\mathcal{A}_n(R_1),\\
\dfrac{2\omega W_{nm1}^{m,out}}{P_{rad,nml}}
&=\dfrac{2\omega W_{nm2}^{e,out}}{P_{rad,nml}}
=\frac{n+1}{2n+1}\mathcal{A}_{n-1}(R_1)+\frac{n}{2n+1}\mathcal{A}_{n+1}(R_1)
\end{align}
for $\mathcal{A}_n$ defined as
\begin{align}
\mathcal{A}_n(r)
\triangleq
-\frac{(k_0r)^3}{2}
\bigg(\left\lvert h_n^{(1)}(k_0r)\right\rvert^2-j_{n+1}(k_0r)j_{n-1}(k_0r)
-y_{n+1}(k_0r)y_{n-1}(k_0r)-\frac{2}{(k_0r)^2}\bigg)
\end{align}
for all $n$. In addition, by using the electric field generated from the source derived in Appendix~\ref{sec_appendix_field}, the electric energy stored inside $V$ is derived as
\begin{align}
W_{nml}^{e,in}
=\frac{\epsilon_1'}{4}\omega^2\mu_0^2
\bigg(
\left\lvert
F_{n,l}
\right\rvert^2
+
\left(\frac{1}{4k_1'k_1''}\right)^2
+
\frac{\Re\{F_{n,l}I_{n,l}^{jj}\}}{2k_1'k_1''I_{n,l}^{jj^*}}
\bigg),
\end{align}
where $F_{n,l}\triangleq D_{n,l}-E_{n,l}/(4k_1'k_1''), \forall n,l$.
Also, by using the magnetic field generated from the source derived in Appendix~\ref{sec_appendix_field}, the magnetic energy stored inside $V$ is derived as
\begin{align}
W_{nml}^{m,in}
=\frac{\mu_0}{4}
\bigg(
\bigg[
\left\lvert
k_1F_{n,l}
\right\rvert^2
+
\left\lvert\frac{k_1}{4k_1'k_1''}\right\rvert^2
\bigg]
\frac{I_{n,3-l}^{jj^*}}{I_{n,l}^{jj^*}}
+
\frac{\Re\{k_1^2F_{n,l}I_{n,3-l}^{jj}\}}{2k_1'k_1''I_{n,l}^{jj^*}}
\bigg).
\end{align}

%---------------------------------------------------------------------------------------------------------------------------

%---------------------------------------------------------------------------------------------------------------------------
\subsection{Power consumption considering the near-field scattering}\label{sec_practical_effect}
\begin{figure}[t]
\begin{center}
\includegraphics[height=2.4in]{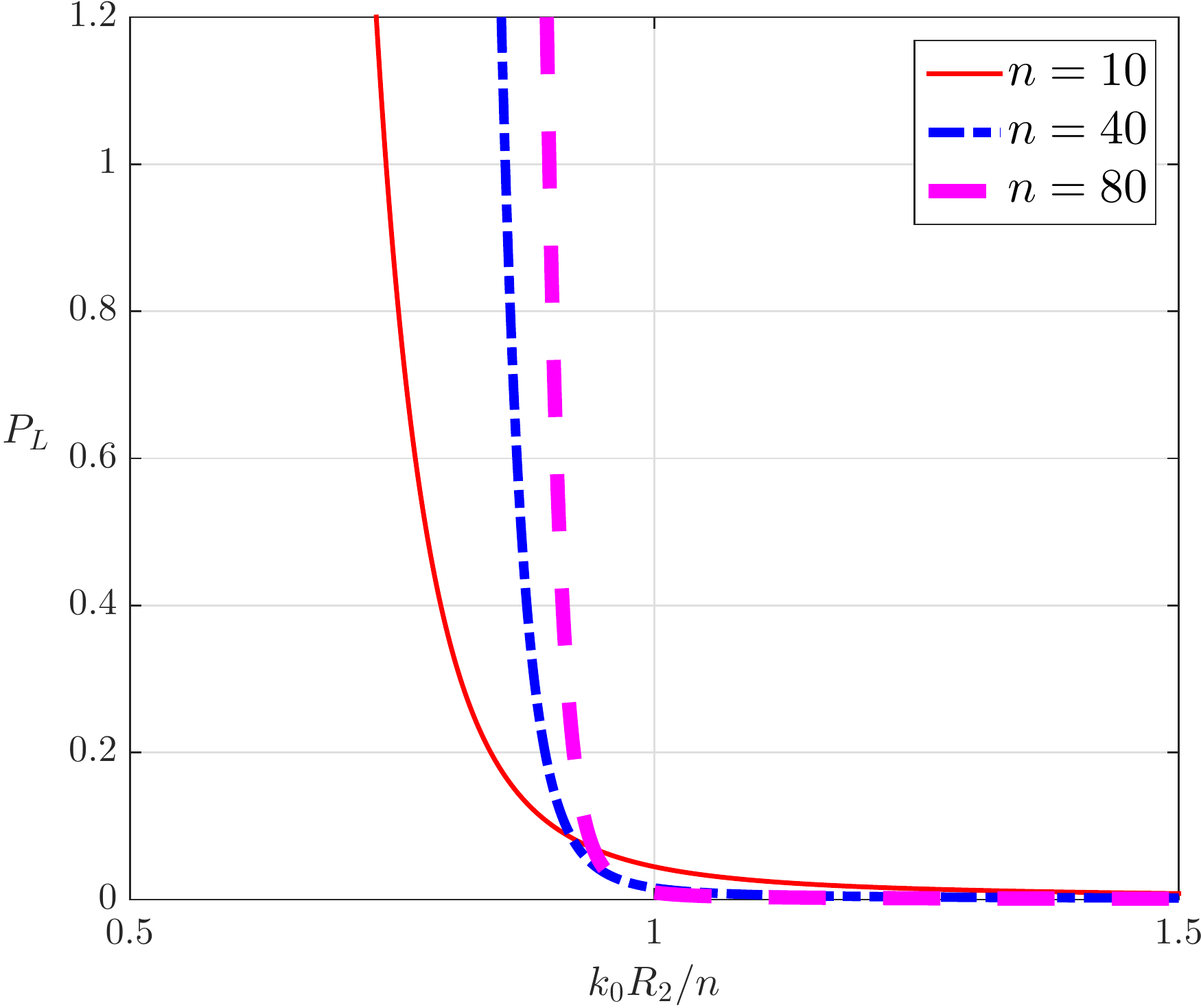}
\includegraphics[height=2.4in]{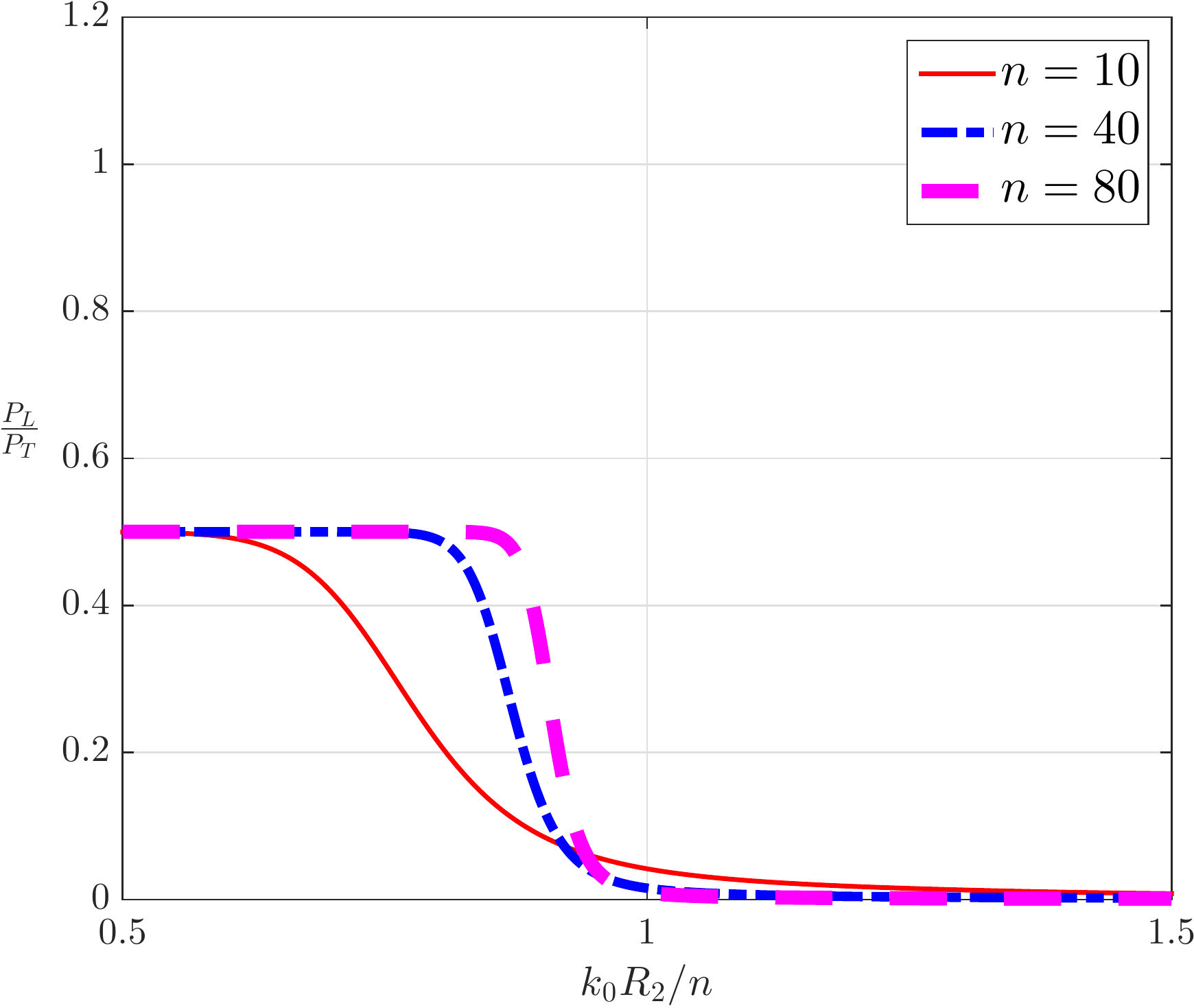}
\end{center}
\vspace{-0.1in}
\caption{The power transfered to the load (left) and the ratio of the transfered power to the power consumption (right)}\label{fig:powerconsumption}
\vspace{-0.1in}
\end{figure}

In many researches on multi-antenna communication, the electromagnetic interaction between the transmitter and the receiver is usually ignored for the calculation of the power consumption. 
This is due to the assumption that the distance between the transceivers is sufficiently far apart~\cite{ivrlac2010toward}.
As we introduced in Section \ref{sec_preliminaries_near}, however, the size of the near-field region becomes larger as the order $n$ of the spherical waves increases.
Thus, for given distance between the transceiver, the maximum order $n$ of the spherical waves should be suitably bounded, i.e., the receiver is far from the reactive near-field region of the transmitter.
In this subsection, we show the power consumption at the transmitter is critically affected due to the back-scattered waves when the receiver is inside the reactive near-field region of the transmitter. 
Formally, assume the current source $\mathbf{J}$ inside a free-space sphere $V$, i.e., $k_1=k_0$, and that source generates the electric field $\mathbf{E}$ on a sphere $S$, where the spheres $V$, $S$ and their corresponding parameters are the same as those in the previous section~\footnote{Even though the free-space sphere is assumed for simple analysis in this section, a similar analysis can be done for general dielectric sphere.}.
In addition, assume the receiver measures the electric field on $\mathbf{s}\triangleq(R_2,0,\cdot)$ by using a single dipole antenna with length $L$, which is perfectly matched to the load impedance~\cite{balanis2005antenna} and has its orientation as $\hbf{x}$.
Then, as we mentioned in the gain-maximization problem in Section~\ref{sec_discussion_numerical}, it is sufficient to consider the spherical vector waves with indices $(n,m,l)\in\boldsymbol{\Upsilon}'\triangleq\left\{(n,m,l):1\le n\le N, m=\pm1,l=1,2\right\}$.

For simplicity, suppose the transmitter uses the source generating the $(n,1,1)$ mode, i.e., 
\begin{align}
\mathbf{J}=\sqrt{\frac{2}{\omega\mu_0\tau_{n11}}}\mathbf{v}_{n11},
\end{align}
where the power consumption at the source is equal to 1 if the effect of the backscattered field is ignored.
This source generates the electric field 
\begin{align}
\mathbf{E}
=-\frac{\sqrt{2\omega\mu_0k_0}}{\sqrt{n(n+1)}}\mathbf{U}_{n11}.
\end{align}
By using the derivations in~\cite{hansen1988spherical}, we have
\begin{align}
\mathbf{U}_{n11}(\mathbf{s})
=
-ih_n^{(1)}(kR_2)(\hbf{x}+i\hbf{y})\mathcal{U}_n, 
\mathbf{U}_{n11}^\star(\mathbf{s})
=
-ih_n^{(2)}(kR_2)(\hbf{x}+i\hbf{y})\mathcal{U}_n, 
\end{align}
where $\mathcal{U}_n=\sqrt{n(n+1)(2n+1)/(16\pi)}$ and $h_n^{(2)}$ is the spherical Hankel function of the second kind. 
Then, the open-circuit voltage $V$ at the dipole is
\begin{align}
V
\triangleq-L\hbf{x}^T\mathbf{E}(\mathbf{s})
=-\sqrt{15}ikL	h_n^{(1)}(kR_2)\sqrt{2n+1},\
\end{align}
by using $\omega\mu_0k_0=\eta_0k_0^2$~\cite{hansen1988spherical}.
From~\cite{balanis2005antenna}, the power $P_L$ transferred to the load resistance and captured by the receiver is
\begin{align}
P_L
=
\frac{|V|^2}{8\mathrm{R}_r}=\frac{3}{32}(2n+1)\left(y_n^2(kR_2)+j_n^2(kR_2)\right),
\end{align}
where $\mathrm{R}_r\triangleq20k_0^2L^2$ is the radiation resistance of the dipole~\cite{balanis2005antenna} and $y_n$ is the spherical Bessel function of the second kind.
As shown in Fig.~\ref{fig:powerconsumption}, the power captured by the receiver in the reactive near-field region is greater than the power consumption calculated by ignoring the back-scattering effect, i.e., 
\begin{align}
P_L>1\text{~for~}kR_2\ll n.
\end{align}
Thus, ignoring the back-scattering effect on the power consumption contradicts the energy conservation law since the receiving power is greater than the radiation power. 

The power consumption considering the back scattering is calculated as follows.
The volume current density on the dipole is 
\begin{align}
\mathbf{J}_S(\mathbf{r})=\hbf{x}IL\delta(\mathbf{r}-\mathbf{s}).
\end{align}
since the current $I$ on the dipole satisfies $IL=-{VL}/{(2\mathrm{R}_r)}$.
By using the decomposition of the DGF for the inward direction in Section~\ref{sec_preliminaries_dyadic}, the scattered field $\mathbf{E}_S$ in $V$ due to $\mathbf{J}_S$ is
\begin{align}
\mathbf{E}_S
=
\sum_{nml}\mathcal{G}_{n,l}(\mathbf{u}_{nml}^\star(\mathbf{s})^H\hbf{x}IL)\mathbf{v}_{nml},
\end{align}
where $\mathcal{G}_{n,l}=-\omega\mu_0k_0\mathcal{N}_{n,l}^V\mathcal{N}_{n,l}^S/(n(n+1)),\forall n,m,l$.
Then, the power consumption $P_S$ in $V$ due to the scattered field $\mathbf{E}_S$ is
\begin{align}
P_S
&\triangleq\frac{1}{2}
\Re\left\{-
\left\langle\mathbf{E}_S,\mathbf{J}\right\rangle_V
\right\}=
\frac{1}{2}\Re\left\{-\mathcal{G}_{n,1}\sqrt{\frac{2}{\omega\mu_0\tau_{n11}}}\mathbf{u}_{nml}^\star(\mathbf{s})^H\hbf{x}IL\right\}\nonumber
=\frac{3}{16}(2n+1)\left(n_n^2(kR_2)-j_n^2(kR_2)\right),\label{eq:scatter_p}
\end{align}	
where the second equality holds due to the orthogonality of $\mathbf{v}_{nml}$'s.
Note that $P_S\rightarrow 0$ as $kR_2\rightarrow\infty$, which implies $P_S$ is negligible when the transmitter and the receiver are sufficiently far apart.
The power consumption $P_T$ at the transmitter considering the near-field back scattering is
\begin{align}
P_T=P_R+P_S,
\end{align}
where $P_R=1$ is the power consumption due to the field ignoring the back-scattered field. 
It is shown that for fixed $\beta\triangleq(kR_S)/n$, which is the boundary of the reactive near-field region of the transmitter mentioned in Section~\ref{sec_preliminaries_near}, the ratio of the receiving power to the power consumption at the transmitter is 
\begin{align}
\frac{P_L}{P_T}\rightarrow\frac{1}{2}
\label{eq:sec_single_effect_1}
\end{align}
as $n\rightarrow \infty$ if $\beta<1$ and 
\begin{align}
\frac{P_L}{P_T}\rightarrow 0
\label{eq:sec_single_effect_2}
\end{align}
as $n\rightarrow \infty$ if $\beta>1$ as shown in Fig. \ref{fig:powerconsumption}.
As a result, the results in~\eqref{eq:sec_single_effect_1} and~\eqref{eq:sec_single_effect_2} do not violate the energy conservation law.
The detailed proof is in Appendix \ref{appendix:scattering}.

\section{Discussion}

\subsection{Numerical results}\label{sec_discussion_numerical}

\begin{figure}[t]
\centering
\subfigure[]{\includegraphics[height=2.3in]{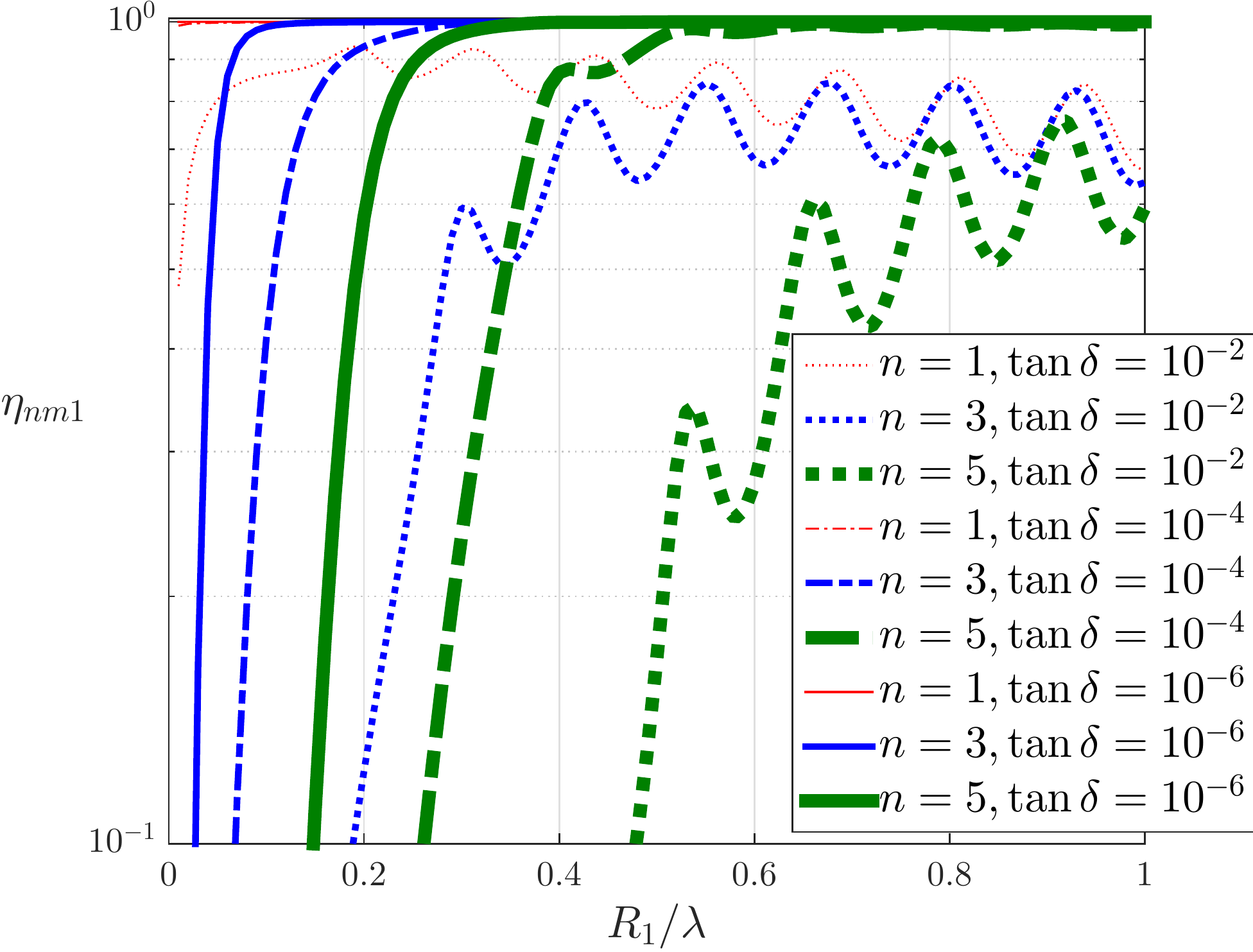}\hspace{0.2in}}
\subfigure[]{\includegraphics[height=2.3in]{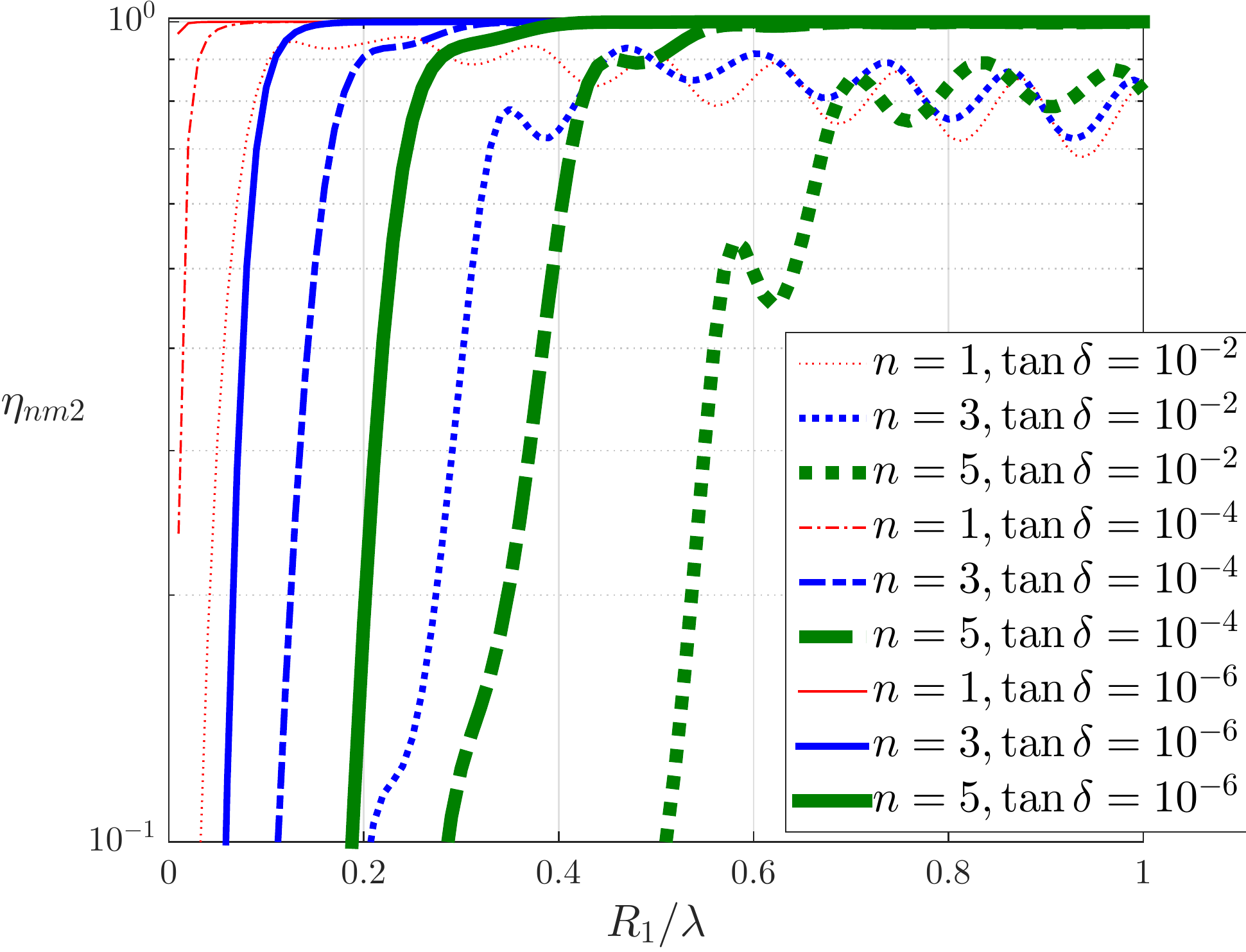}\hspace{0.2in}}
\caption{The efficiency $\eta_{nml}$ of the source $\mathbf{v}_{nml}^\star$ (a) for $l=1$ and (b) for $l=2$.}\label{figure5figure6}
\end{figure}

%%\begin{figure}[t]
%%\begin{center}
%%\includegraphics[height=2.4in]{figure9.eps}
%%\hspace{0.2in}
%%\includegraphics[height=2.4in]{figure10.eps}
%%\end{center}
%%\vspace{-0.2in}
%%\caption{The radiation quality factor $\tilde{Q}_{nml}$ of the source $\mathbf{v}_{nml}^\star$ for $l=1$ (left) and $l=2$ (right).}\label{figure9figure10}
%%\end{figure}

\begin{figure}[t]
\centering
\subfigure[]{\includegraphics[height=2.3in]{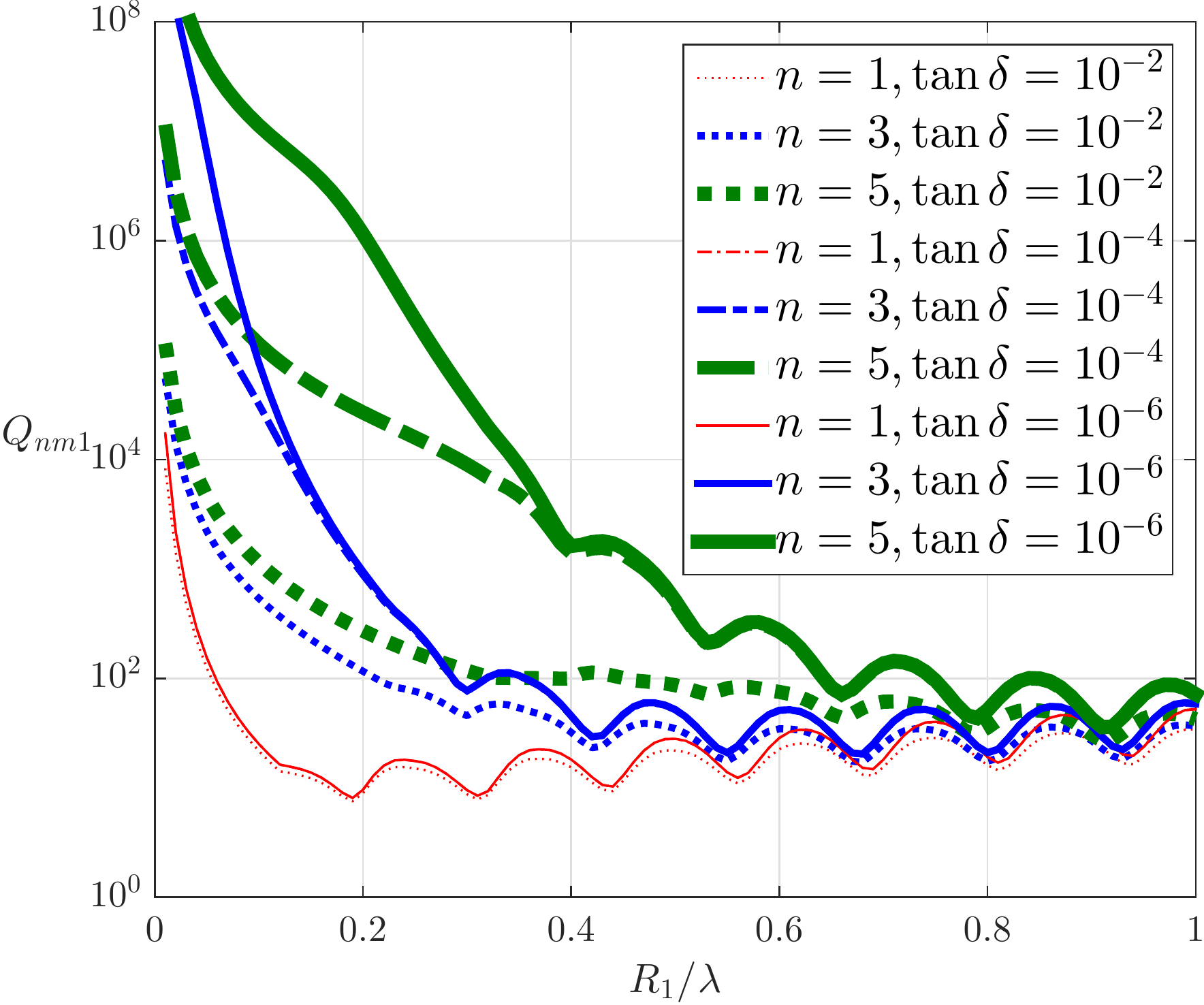}\hspace{0.4in}}
\subfigure[]{\includegraphics[height=2.3in]{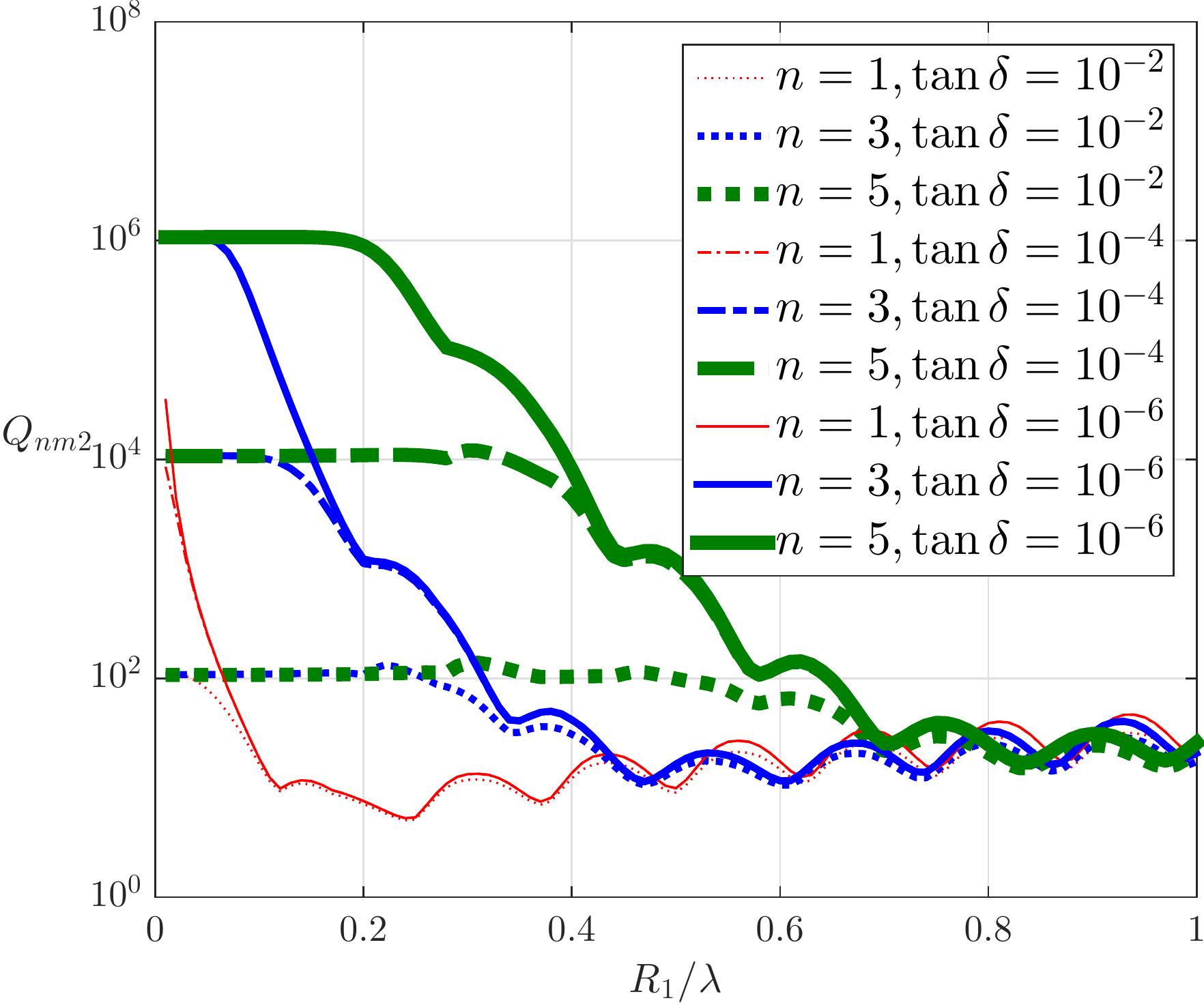}\hspace{0.4in}}
\caption{The quality factor $Q_{nml}$ of the source $\mathbf{v}_{nml}^\star$ (a) for $l=1$ and (b) for $l=2$.}\label{figure3figure4}
\end{figure}

\begin{figure}[t]
\begin{center}
\includegraphics[height=2.4in]{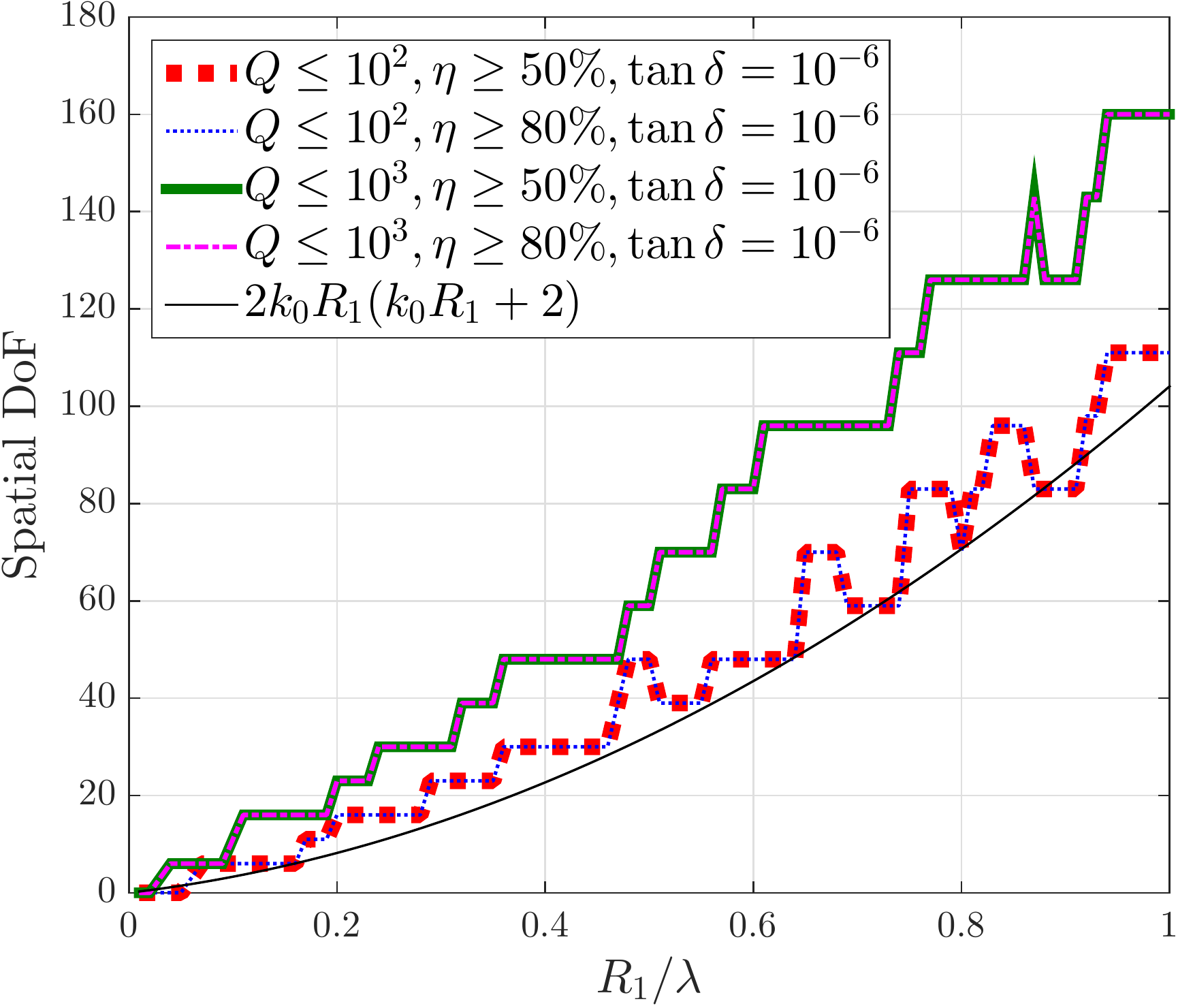}
\end{center}
\vspace{-0.2in}
\caption{The spatial DoF at single frequency when $\tan\delta=10^{-4}$.}\label{figure7figure8}
\end{figure}

Let us assume the carrier frequency $\omega=\omega_c\triangleq2\pi f_c$, where $f_c=16.8\mathrm{GHz}$.
Also, assume the permittivity of the dielectric sphere is $\epsilon_1=\epsilon_r\epsilon_0(1+i\tan\delta)$, where $\epsilon_r$ is called the relative permittivity and $\tan\delta$ is called the loss tangent that represents the lossy property of the dielectric.

For $\epsilon_r=16, n=1,3,5, \tan\delta=10^{-2},10^{-4},10^{-6}$, we have the following numerical results. 
First, in Fig.~\ref{figure5figure6}, the efficiency $\eta_{nml}$, which was derived in Section~\ref{sec_single_equivalent}, is plotted. In this figure, one can see that $\eta_{nml}$ gets smaller as (1) $n$ increases, (2) $R_1$ decreases, or (3) $\tan\delta$ gets larger.
Second, the quality factor $Q_{nml}$ derived in Section~\ref{sec_practical_Q} is plotted in Fig.~\ref{figure3figure4}.
An interesting behavior in this figure is that both $Q_{nm1}$ and $Q_{nm2}$ converge as $R_1/\lambda\rightarrow 0$, and $Q_{nm2}$ converges to a much smaller value than $Q_{nm1}$ does for the same $n$ and $\tan\delta$.
Also, both $Q_{nm1}$ and $Q_{nm2}$ converge to smaller values as $\tan\delta$ gets larger.
Third, the achievable spatial DoF is numerically plotted in Fig.~\ref{figure7figure8}.
Here, we use the upper bound for all $Q_{nml}$ and the lower bound for all $\eta_{nml}$, and the indices $n,m,l$ are omitted in the figure. 
Also, we plot the curve $2k_0R_1(k_0R_1+2)$, which is equal to the achievable DoF when the spherical waves with $n\le k_0R_1$ are usable. 
Here, note that the spherical waves with $n\le k_0R_1$ are mainly used when the source region is filled with the medium with high conductivity~\cite{harrington1960effect}.
As shown in Fig.~\ref{figure7figure8}, the achievable DoF is larger than $2k_0R_1(k_0R_1+2)$ when $\tan\delta=10^{-6}$ and the upper bound on $Q_{nml}$ becomes larger.
Thus, spatial DoF depends not only on the size $R_1$ of the sphere, but also on the lower bound on the efficiency and the upper bound on Q factor.
Note that the capacity may decrease due to the high $Q$ since the higher quality factor implies the narrower usable bandwidth near the carrier frequency $\omega_c$~\cite{poon2015does}.
However, if multiple carriers are utilized, higher Q and resulting reduction in bandwidth applies to each carrier independently since each signal is independent and the whole system is linear~\cite{gruber2008new}. Therefore, the total DoF and the capacity are not fundamentally affected by higher Q and the resulting reduction in bandwidth.

In addition, we compare our result to~\cite{krasnok2014experimental}, which demonstrated an antenna composed of the notched dielectric sphere and was possible to efficiently excite the higher order mode, i.e., $n\ge k_0R_1$.
For comparison with~\cite{krasnok2014experimental}, let us assume there is a dielectric sphere with $\epsilon_r=16$ and $\tan\delta=1.2\times 10^{-4}$ at $f_c=16.8\mathrm{GHz}$, which is similar to the property of $\mathrm{MgO}$-$\mathrm{TiO_2}$ at $f_c$. Also, assume the radius of the sphere is $R_1=5\mathrm{mm}$ and the target minimum bandwidth $\Delta f$ is equal to $0.5\mathrm{GHz}$. 
Under such assumption, we solve the gain optimization problem with the restriction on the maximum Q factor, i.e., 
\begin{equation*}
\begin{aligned}
& \underset{\mathbf{J}}{\text{maximize}}
& & G_{\mathbf{J}}(\theta,\phi) \\
& \text{subject to}
& & Q_{\mathbf{J}} \leq \bar{Q},
\end{aligned}
\end{equation*}
where $\mathbf{J}$ is a vector of all the components in $\{J_{nml}:(n,m,l)\in\boldsymbol{\Upsilon}\}$ that follows the ordering in \textbf{Theorem~\ref{sec_single_equivalent_theorem}}, $\bar{Q}\triangleq f_c/\Delta f$, and the gain $G_{\mathbf{J}}$ is defined as
\begin{align}
G_{\mathbf{J}}(\theta,\phi)
\triangleq
\frac{4\pi U_{\mathbf{J}}(\theta,\phi)}{P_{\mathbf{J}}}
\end{align}
for the radiation intensity $U_{\mathbf{J}}$ and the total power consumption $P_{\mathbf{J}}$ 
\begin{align}
U_{\mathbf{J}}(\theta,\phi)\triangleq\lim_{r\rightarrow\infty}\frac{\left\lVert r\mathbf{E}(\mathbf{r})\right\rVert^2}{2\mathrm{Z}_0},
P_{\mathbf{J}}\triangleq\frac{\omega\mu_0}{2}\sum_{\boldsymbol{\Upsilon}}|J_{nml}|^2\tau_{nml},
\end{align}
and 
\begin{align}
Q_{\mathbf{J}}\triangleq
\frac{(\omega\mu_0/2)\sum_{nml}Q_{nml}|J_{nml}|^2\tau_{nml}}{P_{\mathbf{J}}}
\end{align}
from the definition of the quality factor.
Let $(\theta,\phi)=(0,\cdot)$ without loss of generality by using the property of rotational invariance spherical vector waves~\cite{harrington1960effect,hansen2011small}.
Then, it is sufficient to consider the spherical vector waves with indices $(n,m,l)\in\boldsymbol{\Upsilon}'\triangleq\left\{(n,m,l):1\le n\le N, m=\pm1,l=1,2\right\}$ since only modes generate the field at $\theta=0$~\cite{hansen2011small}.
In addition, by using \eqref{single_equivalent_forward_1}, the explicit derivation of $\mathbf{U}_{nml}$ in Appendix~\ref{sec_appendix_definitions} and the far-field behavior of spherical Hankel functions~\cite{potter1967application,hansen1988spherical}, we have
\begin{align}
U_{\mathbf{J}}(0,\cdot)
=
\frac{1}{2\mathrm{Z}_0}
\left\lVert\sum_{(n,m,l)\in\boldsymbol{\Upsilon}'}
\mathcal{K}_{n,l}
J_{nml}\tilde{\mathbf{U}}_{nml}\right\rVert^2,
\end{align}
where $\mathcal{K}_{n,l}\triangleq-\omega\mu_0k_1\mathcal{T}_{n,l}\sqrt{{I_{n,l}^{jj^*}}/(n(n+1))}$ and
\begin{align}
\tilde{\mathbf{U}}_{nm1}
\triangleq
\frac{(-i)^{n+1}}{k_0}(-i\hbf{x}+m\hbf{y})\mathcal{U}_n,
\tilde{\mathbf{U}}_{nm2}
\triangleq
\frac{(-i)^{n}}{k_0}(-m\hbf{x}-i\hbf{y})\mathcal{U}_n
\end{align}
for $n\le N$, $m=\pm1$ and $\mathcal{U}_n\triangleq\sqrt{n(n+1)(2n+1)/(16\pi)}$.

\begin{figure}[t]
\begin{center}
\subfigure[]{\includegraphics[height=2.4in]{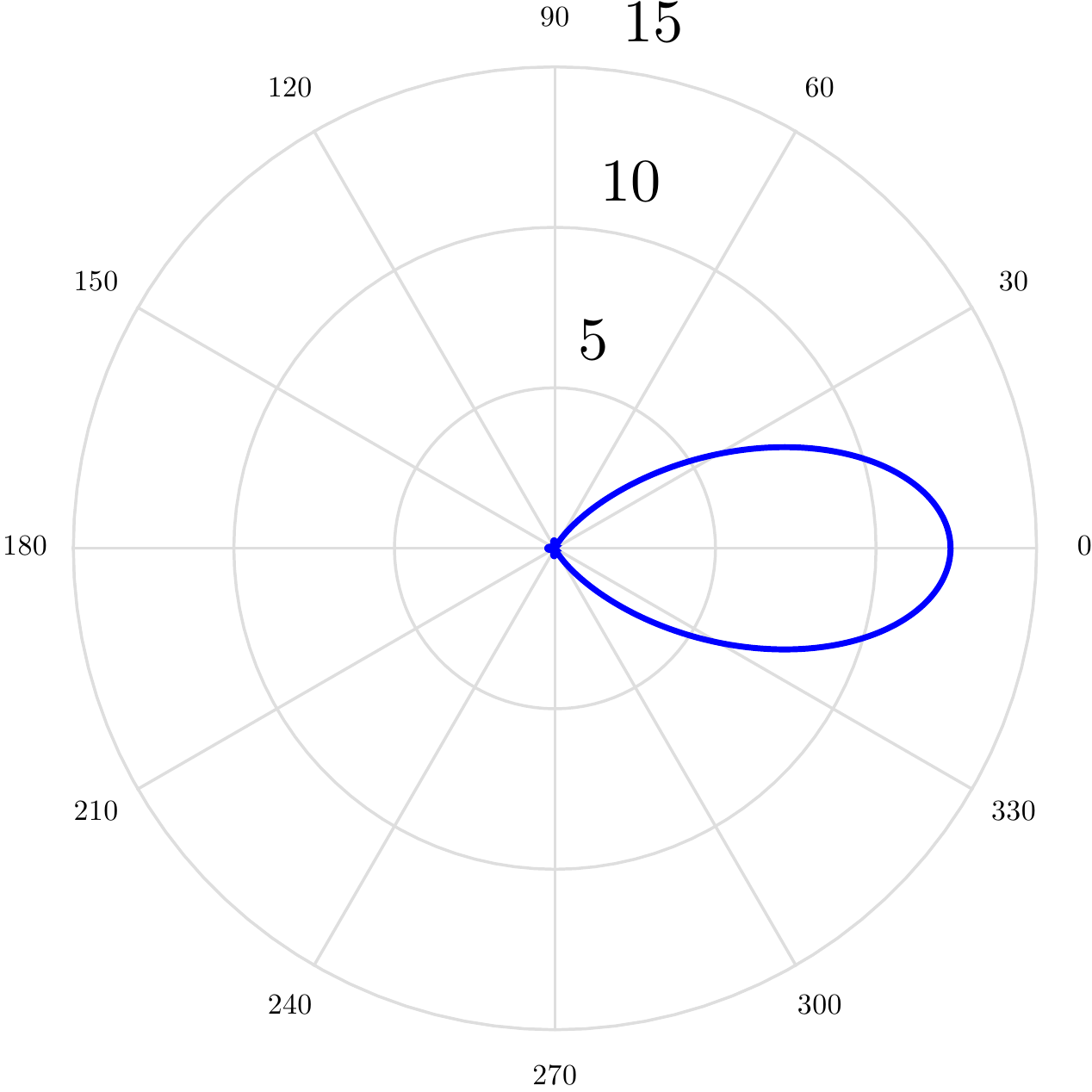}}
\subfigure[]{\includegraphics[height=2.4in]{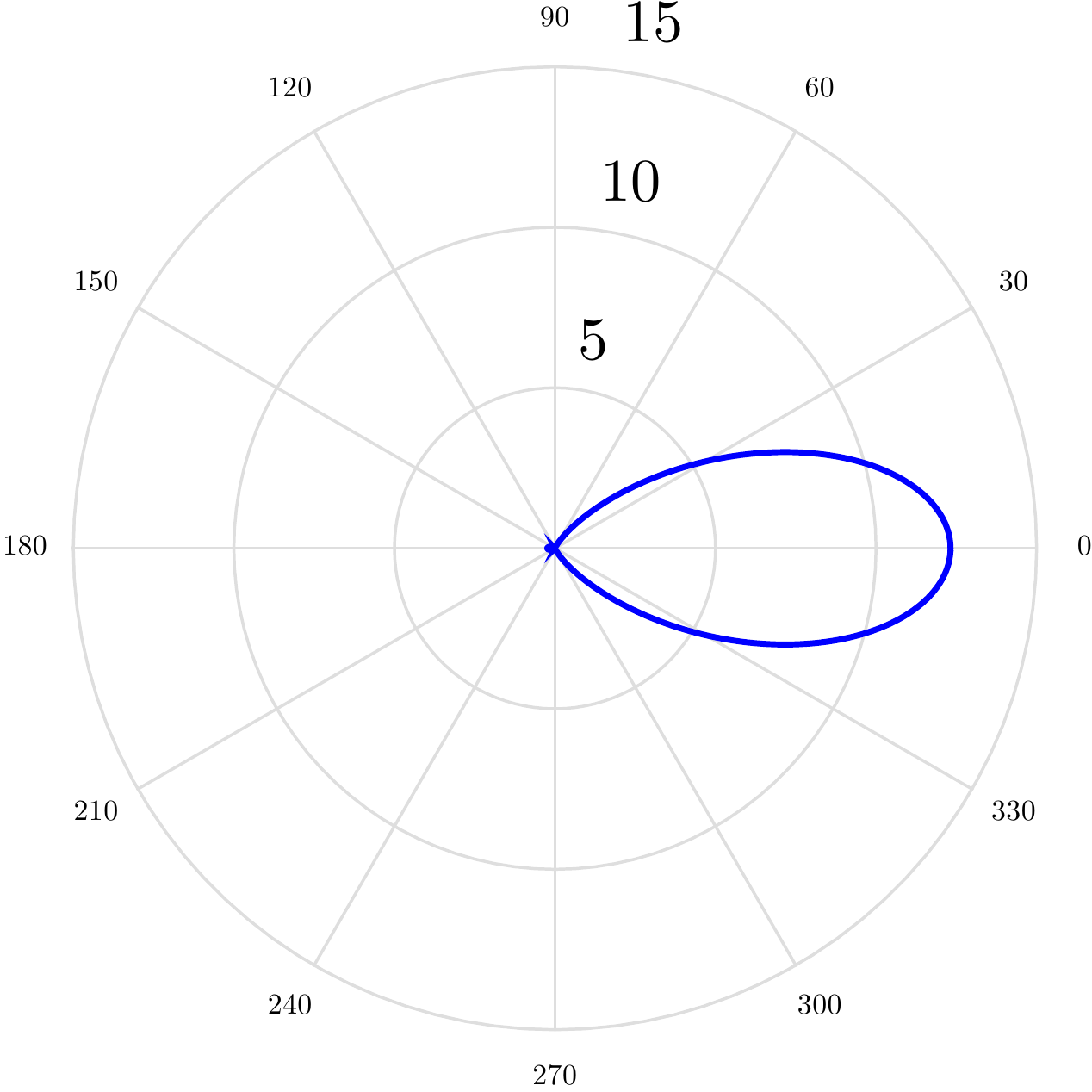}}
\end{center}
\vspace{-0.2in}
\caption{The result of the gain optimization problem. Here, the magnitude in each figure is the gain $G_{\mathbf{J}}(\theta,\phi)$ for the azimuth angle $\theta$. (a) $\phi=0,\pi$. (b) $\phi=\pi/2,3\pi/2$.}\label{figure12}
\end{figure}

As a result, the maximization is achieved for $N=5$, and the resultant beam pattern is given in Fig.~\ref{figure12}.
In our work, the maximum gain, the directivity and the half-gain beamwidth is about $12.31$, $12.36$ and $60^\circ$. In comparison with our work, the corresponding results given in the demonstration in~\cite{krasnok2014experimental} are $8$, $10$ and $35^\circ$. We expect the narrower beamwidth was achieved since the notched sphere is used in~\cite{krasnok2014experimental}.

\subsection{Comparison with previous works}
%
%\begin{table}
%\begin{center}
%\begin{tabular}{c | c | c | c}
%									&Gruber\& Marengo (2008)\cite{gruber2008new}								&Jensen\& Wallce (2008)\cite{jensen2008capacity} 					&Our work 														\\
%\hline\hline
%Constraint on the signal				&radiation power\&current strength 										&radiation power 												&source power consumption											\\
%\hline
%Noise statistics						&i.i.d. field noise															& background noise 									&thermal noise													\\
%									&													&						 									&(fluctuation-dissipation theorem)									\\
%\hline
%Spatial constraint 					&free-space sphere													&arbitary shape and material 												&dielectric sphere    											\\		\hline
%Channel model							&forward   													&forward, reverse										&forward, reverse									\\
%\hline
%Reactive near-field region 				&$kr\lesssim 1$ 														&not considered 												&$kr\lesssim n$ for the order $n$ of subchannels								
%\end{tabular}
%\caption{Comparison with previous works~\cite{gruber2008new,jensen2008capacity}}\label{table2}
%\end{center}
%\end{table}

We compare our work with the existing works~\cite{poon2005degrees,hanlen2006wireless,xu2006electromagnetic,migliore2006role,gruber2008new,jensen2008capacity,poon2011degree}
on the continuous-space electromagnetic channels.
The previous works~\cite{poon2005degrees,hanlen2006wireless,xu2006electromagnetic,migliore2006role,poon2011degree} dealt with the spatial DoF of the channels.
Poon \emph{et al.}~\cite{poon2005degrees} considered the relationship between the size of the source region and the DoF of the channels by assuming the free-space source regions. 
Later, Poon and Tse~\cite{poon2011degree} extended \cite{poon2005degrees} to the polarimetric antenna arrays to seek the extra DoF from polarization diversity.
Hanlen and Fu~\cite{hanlen2006wireless} analyzed DoF of the channels by utilizing the scatter channel model they suggested. 
Xu and Janaswamy~\cite{xu2006electromagnetic} analyzed DoF when the two-dimensional scattering occurs.
Migliore~\cite{migliore2006role} considered how the DoF of electromagnetic channels and the effective DoF of multi-antenna channels are related. 
Note that the previous works~\cite{poon2005degrees,hanlen2006wireless,xu2006electromagnetic,migliore2006role,poon2011degree} restricted the current strength and the noise model used in~\cite{hanlen2006wireless,xu2006electromagnetic,migliore2006role} was the i.i.d. field fluctuation. 
On the other hand, we calculate the exact power consumption that considers the electromagnetic interaction between the field and the source inside the source region and use the exact noise model followed from FDT.
Then, we derive the spatial DoF by considering multiple parameters such as the efficiency, the quality factor and the loss tangent of the dielectric source region. 

Meanwhile, the existing works~\cite{gruber2008new,jensen2008capacity} used the radiation power constraint to analyze the channel capacity.
Gruber and Marengo~\cite{gruber2008new} used the constraint on both the radiation power and the current strength and analyzed the channel capacity when the free-space source region is assumed. 
Jensen and Wallace~\cite{jensen2008capacity} restricted the radiation power and used the background noise model to analyze the capacity of the channels with artificial loss.
If the source region is filled with the lossless medium, the radiation power is equal to the power consumption of the source. Therefore, in lossless case, the channel capacity in our work is equal to the capacity in~\cite{gruber2008new,jensen2008capacity}. 
However, if the source region is filled with lossy medium, the actual power consumption can be differerent from the radiation power, and thus, the capacity in our work is not the same as the capacity in~\cite{gruber2008new,jensen2008capacity}.

%---------------------------------------------------------------------------------------------------------------------------
\section{Conclusion}
In this paper, the capacity of continuous-space electromagnetic channels is analyzed, where the transceivers are confined in dielectric spheres. 
As a result,
we characterized the channel capacity as a function of the size and the physical property of the dielectric and show how the capacity is affected by the radiation efficiency. 
Also, we derived the Q factor and showed the relationship between the DoF of the electromagnetic channels and the Q factor.
In addition, we considered how the backscattered wave at the transmitter affects the power consumption.
Besides, we compared our results with the recent experimental work~\cite{krasnok2014experimental} by solving the gain-optimization problem.

Recently, a major problem in the multi-antenna communication is how to improve the efficiency of the communication when the communication devices have a large number of antennas while the physical size of such devices is limited. 
Our work can provide a useful guideline for such a scenario. There are some experimental results that can solve some of the practical problems mentioned in our paper. For example, the experimental demonstration~\cite{krasnok2014experimental} uses dielectric resonator sufficiently smaller than the wavelength and achieves high efficiency and non-Foster impedance matching in~\cite{sussman2009non} increases the bandwidth by using active circuit elements in antenna impedance matching.

\appendices

\section{Spherical Vector Waves}\label{sec_appendix_definitions}
We explicitly derive the spherical vector waves and their properties in this section. If there is no confusion, the argument $kr$ of the spherical Bessel functions will be omitted. From~\cite{poon2011degree}, we have
\begin{align}
\mathbf{V}_{nml}(k,\mathbf{r})
&=
\begin{cases}
\sqrt{n(n+1)}\mathbf{A}_{nm}^{(1)}(\theta,\phi)j_n^{(1)}(kr),
&\text{~if $l=1$},\\
\dfrac{n(n+1)}{2n+1}j_{n}^{(2)}(kr)\mathbf{A}_{nm}^{(2)}(\theta,\phi)
+
\sqrt{n(n+1)}j_{n}^{(3)}(kr)\mathbf{A}_{nm}^{(3)}(\theta,\phi),
&\text{~if $l=2$},
\end{cases}\\
\mathbf{V}_{nml}^\star(k,\mathbf{r})
&=
\begin{cases}
\sqrt{n(n+1)}\mathbf{A}_{nm}^{(1)}(\theta,\phi)j_n^{(1)}(kr)^*,
&\text{~if $l=1$},\\
\dfrac{n(n+1)}{2n+1}j_{n}^{(2)}(kr)^*\mathbf{A}_{nm}^{(2)}(\theta,\phi)
+
\sqrt{n(n+1)}j_{n}^{(3)}(kr)^*\mathbf{A}_{nm}^{(3)}(\theta,\phi),
&\text{~if $l=2$},
\end{cases}\\
\mathbf{W}_{nml}(k,\mathbf{r})
&=
\begin{cases}
\sqrt{n(n+1)}\mathbf{A}_{nm}^{(1)}(\theta,\phi)y_n^{(1)}(kr),
&\text{~if $l=1$},\\
\dfrac{n(n+1)}{2n+1}y_{n}^{(2)}(kr)\mathbf{A}_{nm}^{(2)}(\theta,\phi)
+
\sqrt{n(n+1)}y_{n}^{(3)}(kr)\mathbf{A}_{nm}^{(3)}(\theta,\phi),
&\text{~if $l=2$},
\end{cases}\\
\mathbf{W}_{nml}^\star(k,\mathbf{r})
&=
\begin{cases}
\sqrt{n(n+1)}\mathbf{A}_{nm}^{(1)}(\theta,\phi)y_n^{(1)}(kr)^*,
&\text{~if $l=1$},\\
\dfrac{n(n+1)}{2n+1}y_{n}^{(2)}(kr)^*\mathbf{A}_{nm}^{(2)}(\theta,\phi)
+
\sqrt{n(n+1)}y_{n}^{(3)}(kr)^*\mathbf{A}_{nm}^{(3)}(\theta,\phi),
&\text{~if $l=2$},
\end{cases}\\
\mathbf{U}_{nml}(k,\mathbf{r})
&=
\mathbf{V}_{nml}(k,\mathbf{r})
+
i\mathbf{W}_{nml}(k,\mathbf{r}),\\
\mathbf{U}_{nml}^\star(k,\mathbf{r})
&=
\mathbf{V}_{nml}^\star(k,\mathbf{r})
-
i\mathbf{W}_{nml}^\star(k,\mathbf{r}),
\end{align}
where for all $n,m$,
\begin{align}
\mathbf{A}_{nm}^{(i)}(\theta,\phi)
&\triangleq
\begin{cases}
\dfrac{\nabla Y_{nm}(\theta,\phi)\times \mathbf{r}}{\sqrt{n(n+1)}},
&\text{~if $i=1$},\\
\hat{\mathbf{r}}Y_{nm}(\theta,\phi),
&\text{~if $i=2$},\\
\dfrac{r\nabla Y_{nm}(\theta,\phi)}{\sqrt{n(n+1)}},
&\text{~if $i=3$},
\end{cases}
\end{align}
and for all $n$,
\begin{align}
j_n^{(i)}(z)
&\triangleq
\begin{cases}
j_n(z),
&\text{~if $i=1$},\\
j_{n-1}(z)+j_{n+1}(z),
&\text{~if $i=2$},\\
\dfrac{n+1}{2n+1}j_{n-1}(z)-\dfrac{n}{2n+1}j_{n+1}(z),
&\text{~if $i=3$}.
\end{cases}\\
y_n^{(i)}(z)
&\triangleq
\begin{cases}
y_n(z),
&\text{~if $i=1$},\\
y_{n-1}(z)+y_{n+1}(z),
&\text{~if $i=2$},\\
\dfrac{n+1}{2n+1}y_{n-1}(z)-\dfrac{n}{2n+1}y_{n+1}(z),
&\text{~if $i=3$}.
\end{cases}
\end{align}
Here, note that for all $n,m,n',m'$,
\begin{align}
\int \mathbf{A}_{nm}^{(i_1)}(\theta,\phi)^H \mathbf{A}_{n'm'}^{(i_2)}(\theta,\phi)d\Omega
=
\delta_{n,n'}\delta_{m,m'}\delta_{i_1,i_2},\forall i_1,i_2=1,2,3.\label{eq:appendix:ortho}
\end{align}
For the spherical waves $\mathbf{V}_{nml},\mathbf{V}_{nml}^\star,\mathbf{W}_{nml},\mathbf{W}_{nml}^\star$ with the wave number $k$, the inner products over the sphere with radius $r$ are derived as follows. Let the arguments $(k,\cdot)$ of the vector waves and $(k,r)$ of integrals be omitted and $V$ be the sphere with radius $r$ centered at origin. Then, we have
\begin{align}
\left\langle
\mathbf{V}_{nml},\mathbf{V}_{n'm'l'}
\right\rangle_{V}
&=
\left\langle
\mathbf{V}_{nml}^\star,\mathbf{V}_{n'm'l'}^\star
\right\rangle_{V}
=
n(n+1)
I_{n,l}^{jj^*}
\delta_{nml,n'm'l'}\\
\left\langle
\mathbf{V}_{nml}^\star,\mathbf{V}_{n'm'l'}
\right\rangle_{V}
&=
n(n+1)
I_{n,l}^{jj}
\delta_{nml,n'm'l'},\\
\left\langle
\mathbf{W}_{nml}^\star,\mathbf{V}_{n'm'l'}
\right\rangle_{V}
&=
n(n+1)
I_{n,l}^{yj}
\delta_{nml,n'm'l'},\\
\left\langle
\mathbf{W}_{nml}^\star,\mathbf{V}_{n'm'l'}^\star
\right\rangle_{V}
&=
n(n+1)
I_{n,l}^{yj^*}
\delta_{nml,n'm'l'}.
\end{align}
for all $n,m,l,n',m',l'$, where 
%Here, the arguments $kr$ of spherical Bessel functions are omitted if there is no confusion. 
\begin{align}
I_{n,1}^{jj}(k,r)%%%%%%%%%%
&\triangleq
\int_0^r j_n(kr')^2r'^2dr'=
\frac{r^3}{2}(j_n^2-j_{n-1}j_{n+1}),\label{appendixA:eq:i1}\\
I_{n,1}^{yj}(k,r)%%%%%%%%%%
&\triangleq
\int_0^r y_n(kr')j_n(kr')r'^2dr'
=\frac{r^3}{2}\left(j_ny_n-\frac{j_{n-1}y_{n+1}+y_{n-1}j_{n+1}}{2}\right)-\frac{2n+1}{4k^3},\label{appendixA:eq:i2}\\
I_{n,1}^{jj^*}(k,r)%%%%%%%%%%
&\triangleq
\int_0^r
\left\lvert j_n(kr')\right\rvert^2r'^2 dr'
\begin{cases}
=
\dfrac{r^2}{k^2-(k^*)^2}(k^*j_{n-1}^*j_n-kj_{n-1}j_n^*),
&~\text{if $k''\ne0$},\\
\rightarrow
I_{n,1}^{jj}(k',r),
&~\text{as $k''\rightarrow0$}.
\end{cases}\label{appendixA:eq:i3}\\
I_{n,1}^{yj^*}(k,r)%%%%%%%%%%
&\triangleq
\int_0^r y_n(kr')j_n(kr')^*r'^2dr'
\begin{cases}
=
\dfrac{r^2}{k^2-(k^*)^2}(k^*j_{n-1}^*y_n-ky_{n-1}j_n^*)+\dfrac{k^{-n}(k^*)^n}{k(k^2-(k^*)^2)},
&~\text{if $k''\ne0$},\\
%=\dfrac{1}{k^2-(k^*)^2}\{r^2(k^*j_{n-1}^*y_n-ky_{n-1}j_n^*)+k^{-(n+1)}(k^*)^n\},
%&~\text{if $k''\ne0$},\\
\rightarrow
I_{n,1}^{yj}(k',r), 
&~\text{as $k''\rightarrow0$},
\end{cases}\label{appendixA:eq:i4}
\end{align}
and
\begin{align}
I_{n,2}(k,r)\triangleq\frac{n+1}{2n+1}I_{n-1,1}(k,r)+\frac{n}{2n+1}I_{n+1,1}(k,r)\label{appendixA:eq:i5}, .
\end{align}
Note that in \eqref{appendixA:eq:i1}, \eqref{appendixA:eq:i2}, \eqref{appendixA:eq:i3}, \eqref{appendixA:eq:i4} and \eqref{appendixA:eq:i5}, the supersciprt $jj, yj, jj^*, yj^*$ on $I_{n,l}$ and the argument $kr$ of the spherical Bessel functions are omitted. 
Also, \eqref{appendixA:eq:i5} is followed from the
properties on $j_n^{(i)}$'s and  $y_n^{(i)}$'s, $i=2,3$ such that
\begin{align}
\frac{n(n+1)}{(2n+1)^2}j_n^{(2)}(z)^2+j_n^{(3)}(z)^2
&=
\frac{n+1}{2n+1}j_{n-1}(z)^2
+
\frac{n}{2n+1}j_{n+1}(z)^2,\\
\frac{n(n+1)}{(2n+1)^2}|j_n^{(2)}(z)|^2+|j_n^{(3)}(z)|^2
&=
\frac{n+1}{2n+1}|j_{n-1}(z)|^2
+
\frac{n}{2n+1}|j_{n+1}(z)|^2,\\
\frac{n(n+1)}{(2n+1)^2}y_n^{(2)}(z)^2+y_n^{(3)}(z)^2
&=
\frac{n+1}{2n+1}y_{n-1}(z)^2
+
\frac{n}{2n+1}y_{n+1}(z)^2,\\
\frac{n(n+1)}{(2n+1)^2}|y_n^{(2)}(z)|^2+|y_n^{(3)}(z)|^2
&=
\frac{n+1}{2n+1}|y_{n-1}(z)|^2
+
\frac{n}{2n+1}|y_{n+1}(z)|^2,\\
\frac{n(n+1)}{(2n+1)^2}y_n^{(2)}(z)j_n^{(2)}(z)+y_n^{(3)}(z)j_n^{(3)}(z)
&=
\frac{n+1}{2n+1}y_{n-1}(z)j_{n-1}(z)
+
\frac{n}{2n+1}y_{n+1}(z)j_{n+1}(z),\\
\frac{n(n+1)}{(2n+1)^2}y_n^{(2)}(z)j_n^{(2)}(z)^*+y_n^{(3)}(z)j_n^{(3)}(z)^*
&=
\frac{n+1}{2n+1}y_{n-1}(z)j_{n-1}(z)^*
+
\frac{n}{2n+1}y_{n+1}(z)j_{n+1}(z)^*.
\end{align}
By using the definition of the inner product, the normalization coefficients, where $V$ is the sphere with radius $R_1$ and $S$ is the sphere with radius $R_2$, are defined as follows:
\begin{align}
\mathcal{N}_{V,\mathbf{V}_{nml}}(k)^2
&=
\mathcal{N}_{V,\mathbf{V}_{nml}^\star}(k)^2
=
\left\langle
\mathbf{V}_{nml}(k,\cdot),\mathbf{V}_{nml}(k,\cdot)
\right\rangle_V
=
\left\langle
\mathbf{V}_{nml}^\star(k,\cdot),\mathbf{V}_{nml}^\star(k,\cdot)
\right\rangle_V
=
n(n+1)I_{n,l}^{jj^*}(k,R_1),\\
\mathcal{N}_{S,\mathbf{U}_{nml}}(k)^2
&=
\mathcal{N}_{S,\mathbf{U}_{nml}^\star}(k)^2=
\begin{cases}
n(n+1)\left\lvert h_n^{(1)}(kR_2)\right\rvert^2
,&\text{~if~}l=1,\\
n(n+1)
\left[
\dfrac{n+1}{2n+1}\left\lvert h_{n-1}^{(1)}(kR_2)\right\rvert^2
+
\dfrac{n}{2n+1}\left\lvert h_{n+1}^{(1)}(kR_2)\right\rvert^2
\right]
,&\text{~if~}l=2.
\end{cases}
\end{align}

\section{Fields inside the dielectric sphere}\label{sec_appendix_field}

By using the decomposition of DGF in Section. \ref{sec_preliminaries_dyadic}, DGF for $\mathbf{r},\mathbf{r}'\in V$ can be represented as
\begin{align}
\obf{G}(\mathbf{r},\mathbf{r}')
&=
ik_1
\sum_{nl}\frac{1}{n(n+1)}\obf{g}_{nl}(\mathbf{r},\mathbf{r}')-\frac{\hat{\mathbf{r}}\hat{\mathbf{r}}^T}{ k_1^2}\delta(\mathbf{r}-\mathbf{r}'),
\end{align}
where for $\mathbf{r},\mathbf{r}'\in V$,
\begin{align}
\obf{g}_{nl}(\mathbf{r},\mathbf{r}')=
\begin{cases}
\sum_{m}
\left[i\mathbf{W}_{nml}(k_1,\mathbf{r})+(1+\mathcal{R}_{n,l})\mathbf{V}_{nml}(k_1,\mathbf{r})\right]
\mathbf{V}_{nml}^\star(k_1,\mathbf{r}')^H, \text{~if~}r\ge r',\\
\sum_{m}
\mathbf{V}_{nml}(k_1,\mathbf{r})
\left[-i\mathbf{W}_{nml}^\star(k_1,\mathbf{r}')+(1+\mathcal{R}_{n,l})^*\mathbf{V}_{nml}^\star(k_1,\mathbf{r}')\right]^H,\text{~if~}r\le r'.
\end{cases}
\end{align}
Let us assume the electric field $\mathbf{E}_{nml}$ is generated due to the source $\mathbf{J}_{nml}(\mathbf{r})\triangleq \mathbf{v}_{nml}^\star(k_1,\mathbf{r})$ in $V$, i.e.,
\begin{align}
\mathbf{E}_{nml}(\mathbf{r})
&=
-\frac{\omega\mu_0}{\sqrt{n(n+1)I_{n,l}^{jj^*}}}
\bigg\{
\frac{k_1}{n(n+1)}
\mathbf{C}_{nml}(\mathbf{r})
+
\frac{i}{k_1^2}[\hat{\mathbf{r}}^T \mathbf{V}_{nml}^\star(k_1,\mathbf{r})]\hat{\mathbf{r}}
\bigg\},
\end{align}
where $\mathbf{C}_{nml}(\mathbf{r})\triangleq\mathbf{C}_{1,nml}(\mathbf{r})+i\mathbf{C}_{2,nml}(\mathbf{r})$ for
\begin{align}
\mathbf{C}_{1,nml}(\mathbf{r})
&\triangleq
\mathbf{V}_{nml}(k_1,\mathbf{r})(1+\mathcal{R}_{n,l})\left\langle\mathbf{V}_{nml}^\star,\mathbf{V}_{nml}^\star\right\rangle_V
\\
\mathbf{C}_{2,nml}(\mathbf{r})
&\triangleq
\mathbf{W}_{nml}(k_1,\mathbf{r})
\left\langle\mathbf{V}_{nml}^\star,\mathbf{V}_{nml}^\star\right\rangle_{V(r)}+
\mathbf{V}_{nml}(k_1,\mathbf{r})
[
\left\langle\mathbf{W}_{nml}^\star,\mathbf{V}_{nml}^\star\right\rangle_V
-
\left\langle\mathbf{W}_{nml}^\star,\mathbf{V}_{nml}^\star\right\rangle_{V(r)}
]
.
\end{align}
Note that the argument $(k_1,\cdot)$ of the spherical vector waves are omitted in this section if there is no confusion and $V(r)$ is a sphere with radius $r$ that is centered at origin. 
Also, the argument of inner product is omitted if it is equal to $(k_1,R_1)$. 
By using the definitions and properties in Appendix~\ref{sec_appendix_definitions}, we have
\begin{align}
\mathbf{C}_{1,nml}(\mathbf{r})
&=
\begin{cases}%%%%%%%%%%
\left(n(n+1)\right)^{3/2}
\mathbf{A}_{nm}^{(1)}(\theta,\phi)
\left(
(1+\mathcal{R}_{n,1})j_n^{(1)}I_{n,1}^{jj^*}(k_1,R_1)
\right),
&\text{~if $l=1$},\\
\dfrac{\left(n(n+1)\right)^2}{2n+1}
\mathbf{A}_{nm}^{(2)}(\theta,\phi)
\left(
(1+\mathcal{R}_{n,2})j_n^{(2)}I_{n,2}^{jj^*}(k_1,R_1)
\right)\\
+
\left(n(n+1)\right)^{3/2}
\mathbf{A}_{nm}^{(3)}(\theta,\phi)
\left(
(1+\mathcal{R}_{n,2})j_n^{(3)}I_{n,2}^{jj^*}(k_1,R_1)
\right),
&\text{~if $l=2$}.
\end{cases}\label{eq:appendix:ccc1}
\end{align}
and
\begin{align}
\mathbf{C}_{2,nml}(\mathbf{r})
&=
\begin{cases}%%%%%%%%%%
\left(n(n+1)\right)^{3/2}
\mathbf{A}_{nm}^{(1)}(\theta,\phi)
\left(
y_n^{(1)}
I_{n,1}^{jj^*}
+
j_n^{(1)}
[I_{n,1}^{yj^*}(k_1,R_1)-I_{n,1}^{yj^*}]
\right),
&\text{~if $l=1$},\\
\dfrac{\left(n(n+1)\right)^2}{2n+1}
\mathbf{A}_{nm}^{(2)}(\theta,\phi)
\left(
y_n^{(2)}
I_{n,2}^{jj^*}
+
j_n^{(2)}
[I_{n,2}^{yj^*}(k_1,R_1)-I_{n,2}^{yj^*}]
\right)\\
+
\left(n(n+1)\right)^{3/2}
\mathbf{A}_{nm}^{(3)}(\theta,\phi)
\left(
y_n^{(3)}
I_{n,2}^{jj^*}
+
j_n^{(3)}
[I_{n,2}^{yj^*}(k_1,R_1)-I_{n,2}^{yj^*}]
\right),
&\text{~if $l=2$}.
\end{cases}\\
&=
\begin{cases}%%%%%%%%%%
\left(n(n+1)\right)^{3/2}
\mathbf{A}_{nm}^{(1)}(\theta,\phi)
\left(
j_n^{(1)}I_{n,1}^{yj^*}(k_1,R_1)
+
[y_n^{(1)}I_{n,1}^{jj^*}-j_n^{(1)}I_{n,1}^{yj^*}]
\right),
&\text{~if $l=1$},\\
\dfrac{\left(n(n+1)\right)^2}{2n+1}
\mathbf{A}_{nm}^{(2)}(\theta,\phi)
\left(
j_n^{(2)}I_{n,2}^{yj^*}(k_1,R_1)
+
[y_n^{(2)}I_{n,2}^{jj^*}-j_n^{(2)}I_{n,2}^{yj^*}]
\right)\\
+
\left(n(n+1)\right)^{3/2}
\mathbf{A}_{nm}^{(3)}(\theta,\phi)
\left(
j_n^{(3)}I_{n,2}^{yj^*}(k_1,R_1)
+
[y_n^{(3)}I_{n,2}^{jj^*}-j_n^{(3)}I_{n,2}^{yj^*}]
\right),
&\text{~if $l=2$},
\end{cases}\label{eq:appendix:ccc2}
\end{align}
where the argument $(k_1,r)$ of $I_{n,l}$'s and the argument $k_1r$ of the spherical Bessel functions are omitted for simplicity. 
Here, we derive
\begin{align}
y_n^{(1)}I_{n,1}^{jj^*}-j_n^{(1)}I_{n,1}^{yj^*}%%%%%%%%%%8
%&=
%\frac{(k^*)^n}{k(k^2-(k^*)^2)}((k^*)^{-n}j_n^*-k^{-n}j_n),\\
&=
\frac{(j_n^{(1)})^*-E_{n,1}(k_1)j_n^{(1)}}{4ik_1k_1'k_1''}\label{eq:appendix:minus1}\\
y_n^{(2)}I_{n,2}^{jj^*}-j_n^{(2)}I_{n,2}^{yj^*}%%%%%%%%%%13
&=
\frac{k_1^{-2}(k_1^*)^2(j_n^{(2)})^*-E_{n,2}(k_1)j_n^{(2)}}{4ik_1k_1'k_1''},\label{eq:appendix:minus2}\\
y_n^{(3)}I_{n,2}^{jj^*}-j_n^{(3)}I_{n,2}^{yj^*}%%%%%%%%%%14
&=
\frac{(j_n^{(3)})^*-E_{n,2}(k_1)j_{n}^{(3)}}{4ik_1k_1'k_1''}.\label{eq:appendix:minus3}
\end{align}
for 
\begin{align}
E_{n,l}(k)
\triangleq
\begin{cases}
k^{-n}(k^*)^n,
&\text{~if~}l=1,\\
\dfrac{n+1}{2n+1}E_{n-1,1}(k)+\dfrac{n}{2n+1}E_{n+1,1}(k),
&\text{~if~}l=2,
\end{cases}
\end{align}
by using differentiation, recurrence formula and Wronskian properties of spherical Bessel functions such that 
\begin{align}
\frac{\partial j_n(z)}{\partial z}
&=
\frac{j_{n-1}(z)-j_{n+1}(z)}{2}-\frac{j_n(z)}{2z}\\
j_{n-1}(z)+j_{n+1}(z)
&=
(2n+1)z^{-1}j_n(z),\\
j_{n+1}(z)y_n(z)-j_n(z)y_{n+1}(z)
&
=z^{-2},\\
j_{n+2}(z)y_n(z)-j_n(z)y_{n+2}(z)
&
=(2n+3)z^{-3},\\
j_{n+3}(z)y_n(z)-j_n(z)y_{n+3}(z)
&
=(2n+3)(2n+5)z^{-4}-z^{-2}.
\end{align}
By using \eqref{eq:appendix:ccc1}, \eqref{eq:appendix:ccc2}, \eqref{eq:appendix:minus1}, \eqref{eq:appendix:minus2} and \eqref{eq:appendix:minus3},  the electric field is derived as
\begin{align}
\mathbf{E}_{nml}(\mathbf{r})
&=
\begin{cases}
-\dfrac{\omega\mu_0}{\sqrt{I_{n,1}^{jj^*}(k_1,R_1)}}
\mathbf{A}_{nm}^{(1)}(\theta,\phi)
\bigg(
\mathcal{D}_{n,1}j_n^{(1)}+\dfrac{(j_n^{(1)})^*-\mathcal{E}_{n,1}j_n^{(1)}}{4k_1'k_1''}
\bigg)
,
&\text{~if $l=1$},\\
-\dfrac{\omega\mu_0}{\sqrt{I_{n,2}^{jj^*}(k_1,R_1)}}
\bigg\{
\dfrac{\sqrt{n(n+1)}}{2n+1}\mathbf{A}_{nm}^{(2)}(\theta,\phi)
\bigg(
\mathcal{D}_{n,2}j_n^{(2)}+\dfrac{(j_n^{(2)})^*-\mathcal{E}_{n,2}j_n^{(2)}}{4k_1'k_1''}
\bigg)\\
+
\mathbf{A}_{nm}^{(3)}(\theta,\phi)
\bigg(
\mathcal{D}_{n,2}j_n^{(3)}+\dfrac{(j_n^{(3)})^*-\mathcal{E}_{n,2}j_n^{(3)}}{4k_1'k_1''}
\bigg)
\bigg\}
,
&\text{~if $l=2$},
\end{cases}\\
&=
\begin{cases}
-\dfrac{\omega\mu_0}{\sqrt{I_{n,1}^{jj^*}(k_1,R_1)}}
\mathbf{A}_{nm}^{(1)}(\theta,\phi)
\bigg(
\mathcal{F}_{n,1}j_n^{(1)}+\dfrac{(j_n^{(1)})^*}{4k_1'k_1''}
\bigg)
,
&\text{~if $l=1$},\\
-\dfrac{\omega\mu_0}{\sqrt{I_{n,2}^{jj^*}(k_1,R_1)}}
\bigg\{
\dfrac{\sqrt{n(n+1)}}{2n+1}\mathbf{A}_{nm}^{(2)}(\theta,\phi)
\bigg(
\mathcal{F}_{n,2}j_n^{(2)}+\dfrac{(j_n^{(2)})^*}{4k_1'k_1''}
\bigg)\\
+
\mathbf{A}_{nm}^{(3)}(\theta,\phi)
\bigg(
\mathcal{F}_{n,2}j_n^{(3)}+\dfrac{(j_n^{(3)})^*}{4k_1'k_1''}
\bigg)
\bigg\}
,
&\text{~if $l=2$},
\end{cases}\\
&=
-\dfrac{\omega\mu_0}{\mathcal{N}_{V,\mathbf{V}_{nml}}(k_1)}
\bigg(
\mathcal{F}_{n,l}\mathbf{V}_{nml}(k_1,\mathbf{r})+\dfrac{1}{4k_1'k_1''}\mathbf{V}_{nml}(k_1^*,\mathbf{r})
\bigg),
\end{align}
where $\mathcal{D}_{n,l}, \mathcal{E}_{n,l}$ and $\mathcal{F}_{n,l}$ are defined in \textbf{Theorem~\ref{sec_single_equivalent_theorem}} 
and the arguments $k_1$ of $\mathcal{D}_{n,l}, \mathcal{E}_{n,l}$ and $\mathcal{F}_{n,l}$ are omitted for simplicity. Also, note that $\mathbf{V}_{nml}^\star(k_1,\mathbf{r})=\mathbf{V}_{nml}(k_1^*,\mathbf{r})$ since $j_{n}^*(k_1r)=j_n(k_1^*r)$~\cite{abramowitz1966handbook}. 
In addition, by using $\mathbf{H}_{nml}=\nabla\times\mathbf{E}_{nml}/(i\omega\mu_0)$ and $\nabla\times\mathbf{V}_{nml}(k,\mathbf{r})=k\mathbf{V}_{nm,3-l}(k,\mathbf{r})$ for any $k$~\cite[p.36]{kristensson2014spherical}, the magnetic field $\mathbf{H}_{nml}$ in $V$ due to the defined source is 
\begin{align}
\mathbf{H}_{nml}(\mathbf{r})
&=
\dfrac{i}{\mathcal{N}_{V,\mathbf{V}_{nml}}(k_1)}
\bigg(
k_1\mathcal{F}_{n,l}\mathbf{V}_{nm,3-l}(k_1,\mathbf{r})+\dfrac{k_1^*}{4k_1'k_1''}\mathbf{V}_{nm,3-l}(k_1^*,\mathbf{r})
\bigg),\forall n,m,l. 
\end{align}

\section{Noise statistics}\label{sec_appendix_noise}
For the noise statistics at the reverse channel of the single-user case, we have to derive
\begin{align}
-i
\int_V\int_V \mathbf{v}_{nml}(k_1,\mathbf{r})^H\obf{G}(\mathbf{r},\mathbf{r}')\mathbf{v}_{nml}(k_1,\mathbf{r}')d\mathbf{r}d\mathbf{r}',\label{sec_appendix_noise_eq_1}
\end{align}
where the decomposition of DGF is given in Appendix~\ref{sec_appendix_field}.
Note that we calculate the double integral by first doing the integral over $\mathbf{r}$ and then over $\mathbf{r}'$. 
Let us first calculate
\begin{align}
-i
\int_V
\mathbf{v}_{nml}(k_1,\mathbf{r})^H\obf{G}(\mathbf{r},\mathbf{r}')d \mathbf{r}
=
\frac{1}{\sqrt{\mathcal{N}_{V,\mathbf{V}_{nml}}(k_1)}}
\left\{
\frac{k_1}{n(n+1)}\mathbf{R}_{nml}(\mathbf{r}')
+
\frac{i}{k_1^2}[\mathbf{V}_{nml}(k_1,\mathbf{r}')^H\hat{\mathbf{r}}]\hat{\mathbf{r}}^T
\right\},
\end{align}
where
\begin{align}
\mathbf{R}_{nml}(\mathbf{r}')
&\triangleq
\mathbf{R}_{1,nml}(\mathbf{r}')
+
i\mathbf{R}_{2,nml}(\mathbf{r}')
\end{align}
for
\begin{align}
\mathbf{R}_{1,nml}(\mathbf{r}')
&\triangleq
(1+\mathcal{R}_{n,l})\left\langle\mathbf{V}_{nml},\mathbf{V}_{nml}\right\rangle_V\mathbf{V}_{nml}^\star(k_1,\mathbf{r}')^H,
\\
\mathbf{R}_{2,nml}(\mathbf{r}')
&\triangleq
\left\langle\mathbf{V}_{nml},\mathbf{V}_{nml}\right\rangle_{V(r')}
\mathbf{W}_{nml}^\star(k_1,\mathbf{r}')^H+
[
\left\langle\mathbf{V}_{nml},\mathbf{W}_{nml}\right\rangle_V
-
\left\langle\mathbf{V}_{nml},\mathbf{W}_{nml}\right\rangle_{V(r')}
]
\mathbf{V}_{nml}^\star(k_1,\mathbf{r}')^H,
\end{align}
where $V(r)$ is a sphere with radius $r$ that is centered at origin. 
By using the definitions and properties in Appendix~\ref{sec_appendix_definitions}, we have
\begin{align}
\mathbf{R}_{1,nml}(\mathbf{r})
&=
\begin{cases}%%%%%%%%%%
\left(n(n+1)\right)^{3/2}
\left(
(1+\mathcal{R}_{n,1})j_n^{(1)}I_{n,1}^{jj^*}(k_1,R_1)
\right)
\mathbf{A}_{nm}^{(1)}(\theta,\phi)^H,
&\text{~if $l=1$},\\
\dfrac{\left(n(n+1)\right)^2}{2n+1}
\left(
(1+\mathcal{R}_{n,2})j_n^{(2)}I_{n,2}^{jj^*}(k_1,R_1)
\right)
\mathbf{A}_{nm}^{(2)}(\theta,\phi)^H\\
+
\left(n(n+1)\right)^{3/2}
\left(
(1+\mathcal{R}_{n,2})j_n^{(3)}I_{n,2}^{jj^*}(k_1,R_1)
\right)
\mathbf{A}_{nm}^{(3)}(\theta,\phi)^H,
&\text{~if $l=2$}.
\end{cases}
\end{align}
and
\begin{align}
\mathbf{R}_{2,nml}(\mathbf{r})
&=
\begin{cases}%%%%%%%%%%
\left(n(n+1)\right)^{3/2}
\left(
j_n^{(1)}I_{n,1}^{yj^*}(k_1,R_1)
+
[y_n^{(1)}I_{n,1}^{jj^*}-j_n^{(1)}I_{n,1}^{yj^*}]
\right)
\mathbf{A}_{nm}^{(1)}(\theta,\phi)^H,
&\text{~if $l=1$},\\
\dfrac{\left(n(n+1)\right)^2}{2n+1}
\left(
j_n^{(2)}I_{n,2}^{yj^*}(k_1,R_1)
+
[y_n^{(2)}I_{n,2}^{jj^*}-j_n^{(2)}I_{n,2}^{yj^*}]
\right)
\mathbf{A}_{nm}^{(2)}(\theta,\phi)^H\\
+
\left(n(n+1)\right)^{3/2}
\left(j_n^{(3)}I_{n,2}^{yj^*}(k_1,R_1)+[y_n^{(3)}I_{n,2}^{jj^*}-j_n^{(3)}I_{n,2}^{yj^*}]\right)
\mathbf{A}_{nm}^{(3)}(\theta,\phi)^H,
&\text{~if $l=2$},
\end{cases}
\end{align}
where the argument $(k_1,r')$ of $I_{n,l}$'s and the argument $k_1r'$ of the spherical Bessel functions are omitted for simplicity. As a result, by using \eqref{eq:appendix:minus1}, \eqref{eq:appendix:minus2}, \eqref{eq:appendix:minus3} and the orthogonality of $\mathbf{A}_{nm}^{(1)}, \mathbf{A}_{nm}^{(2)}, \mathbf{A}_{nm}^{(3)}$ in \eqref{eq:appendix:ortho}, \eqref{sec_appendix_noise_eq_1} is equal to 
\begin{align}
\frac{\mathcal{F}_{n,l}I_{n,l}^{jj}(k_1,R_1)}{I_{n,l}^{jj^*}}
+
\frac{1}{4k_1'k_1''}.
\end{align}

\section{Proof on the efficiency for the lossless case}\label{sec_appendix_lossless_efficiency}

Let us assume that $l=1$.
From the definition of the scattering coefficients, we have
\begin{align}
1+\mathcal{R}_{n,1}
=
-i
\frac{\mathcal{A}_n}{\mathcal{B}_n},
\mathcal{T}_{n,1}
=
\frac{i}{\mathcal{B}_n},
\end{align}
where $\mathcal{A}_n\triangleq\hat{Y}_n(\mathcal{C}z)\hat{H}_n^{(1)'}(z)-\mathcal{C}\hat{Y}_n'(\mathcal{C}z)\hat{H}_n^{(1)}(z), 
\mathcal{B}_n\triangleq\hat{J}_n(\mathcal{C}z)\hat{H}_n^{(1)'}(z)-\mathcal{C}\hat{J}_n'(\mathcal{C}z)\hat{H}_n^{(1)}(z)$
for $\hat{Y}_n^{(1)}(\rho)\triangleq \rho y_n(\rho)$, $\hat{Y}_n'(\rho)\triangleq \frac{d}{d\rho}\hat{Y}_n(\rho)$.
Then,
\begin{align}
\frac{\mathcal{C}|\mathcal{T}_{n,1}|^2}{\Re\{1+\mathcal{R}_{n,1}\}}
=
\frac{2i\mathcal{C}}{\mathcal{A}_n\mathcal{B}_n^*-\mathcal{A}_n^*\mathcal{B}_n}
=
\frac{\mathcal{C}}{\Im\{\mathcal{A}_n\mathcal{B}_n^*\}}.
\end{align}
Here, $\mathcal{A}_n\mathcal{B}_n^*$ is equal to
\begin{align}
&\hat{Y}_n(\mathcal{C}z)\hat{J}_n(\mathcal{C}z)|\hat{H}_n^{(1)'}(z)|^2+|\mathcal{C}|^2\hat{Y}_n'(\mathcal{C}z)\hat{J}_n'(\mathcal{C}z)|\hat{H}_n^{(1)}(z)|^2\\
&-
\mathcal{C}\hat{J}_n(\mathcal{C}z)\hat{Y}_n'(\mathcal{C}z)\hat{H}_n^{(1)}(z)(\hat{H}_n^{(1)'}(z))^*-\mathcal{C}\hat{J}_n'(\mathcal{C}z)\hat{Y}_n(\mathcal{C}z)\hat{H}_n^{(1)'}(z)(\hat{H}_n^{(1)}(z))^*.
\end{align}
Since $z$ is assumed to be a real number, we have
\begin{align}
\hat{H}_n^{(1)}(z) (\hat{H}_n^{(1)'}(z))^*
%&=(\hat{J}_n(z)+i\hat{Y}(z))(\hat{J}_n'(z)-i\hat{Y}_n'(z))\\
%&=(\hat{J}_n(z)\hat{J}_n'(z)+\hat{Y}_n(z)\hat{Y}_n'(z))\\
%&+i(\hat{J}_n'(z)\hat{Y}_n(z)-\hat{J}_n(z)\hat{Y}_n'(z))\\
&=(\hat{J}_n(z)\hat{J}_n'(z)+\hat{Y}_n(z)\hat{Y}_n'(z))+i,
\end{align}
which follows by using
$\hat{J}_n'(z)\hat{Y}_n(z)-\hat{J}_n(z)\hat{Y}_n'(z)=1$.
%\begin{align}
%\hat{J}_n'(z)\hat{Y}_n(z)-\hat{J}_n(z)\hat{Y}_n'(z)
%%&=
%%\frac{\partial(zj_n(z))}{\partial z}zy_n(z)-zj_n(z)\frac{\partial(zy_n(z))}{\partial z}\\
%%&=
%%\left(j_n(z)+z\frac{\partial j_n(z)}{\partial z}\right)zy_n(z)
%%-
%%zj_n(z)\left(y_n(z)+z\frac{\partial y_n(z)}{\partial z}\right)\\
%%&=z^2(j_n'(z)y_n(z)-j_n(z)y_n'(z))\\
%&=1,
%\end{align}
%where the last equality holds by using the Wronskian property. 
Also, since $\mathcal{C}\triangleq k_1/k_0$ is a real number, the imaginary part of $\mathcal{A}_n\mathcal{B}_n^*$ is equal to
\begin{align}
\Im\{-\mathcal{C}\hat{J}_n(\mathcal{C}z)\hat{Y}_n'(\mathcal{C}z)(\hat{J}_n\hat{J}_n'+\hat{Y}_n\hat{Y}_n'+i)
-\mathcal{C}\hat{J}_n'(\mathcal{C}z)\hat{Y}_n(\mathcal{C}z)(\hat{J}_n\hat{J}_n'+\hat{Y}_n\hat{Y}_n'+i)^*\}=\mathcal{C},
\end{align}
where $\hat{J}_n, \hat{Y}_n$ without arguement have their argument as $z$ and the Wronskian property is used. 
%It is equal to
%\begin{align}
%&-c(\hat{J}_n(cz)\hat{Y}_n'(cz)-\hat{J}_n'(cz)\hat{Y}_n(cz))\\
%&=c(\hat{J}_n'(cz)\hat{Y}_n(cz)-\hat{J}_n(cz)\hat{Y}_n'(cz))\\
%&=c((j_n(cz)+czj_n'(cz))czy_n(cz)-czj_n(cz)(y_n(cz)+czy_n(cz)))\\
%&=c^3z^2(j_n'(cz)y_n(cz)-j_n(cz)y_n'(cz))\\
%&=c^3z^2(cz)^{-2}=c
%\end{align}
%for $j_n'(z)\triangleq\frac{\partial j_n(z)}{\partial z}$,  $y_n'(z)\triangleq\frac{\partial y_n(z)}{\partial z}$, where the laet equality holds by using the Wronskain property.
Therefore
\begin{align}
\frac{\mathcal{C}|\mathcal{T}_{n,1}|^2}{\Re\{1+\mathcal{R}_{n,1}\}}
=1.
\end{align}
The proof on the case $l=2$ can be done similarly.

\section{}\label{Appendix:scattering}\label{appendix:scattering}
The integral representation of the spherical Bessel functions are given as follows:
\begin{align}
h_n^{(1)}(\rho)
&=
-\frac{(\rho/2)^n}{n!}\int_{1}^{i\infty}e^{i\rho t}(1-t^2)^n dt,\\
h_n^{(2)}(\rho)
&=
\frac{(\rho/2)^n}{n!}\int_{-1}^{i\infty}e^{i\rho t}(1-t^2)^n dt,\\
j_n(\rho)
&=\frac{1}{2}\left\{h_n^{(1)}(\rho)+h_n^{(2)}(\rho)\right\}
=
\frac{1}{2}\frac{(\rho/2)^n}{n!}\int_{-1}^{1}e^{i\rho t}(1-t^2)^n dt.
\end{align}
For $n\gg1$ and $\rho\le n$, the approximation for the spherical Bessel functions is given as follows by using the method of steepest descent for integrals in~\cite{olver2014asymptotics}:
\begin{align}
j_n(\rho)
&\simeq
\frac{1}{2\rho}e^{\sqrt{n^2-\rho^2}}\left(\frac{n-\sqrt{n^2-\rho^2}}{\rho}\right)^n
\left(\frac{n-\sqrt{n^2-\rho^2}}{\sqrt{n^2-\rho^2}}\right)^{1/2},\\
n_n(\rho)
&
\simeq
-\frac{1}{\rho}e^{-\sqrt{n^2-\rho^2}}\left(\frac{n+\sqrt{n^2-\rho^2}}{\rho}\right)^n\left(\frac{n+\sqrt{n^2-\rho^2}}{\sqrt{n^2-\rho^2}}\right)^{1/2}.
\end{align}

First, consider the case when $\beta\triangleq k_0R_2/n<1$, which is in the reactive near-field region.
Since $k_0R_2=n\beta\le n$, for $n\gg1$, 
\begin{align}
j_n(k_0R_2)
&\simeq
C_1
\left[\frac{1}{n}\left(\frac{e^{\sqrt{1-\beta^2}}(1-\sqrt{1-\beta^2})}{\beta}\right)^n\right]
\triangleq
C_1
\left[\frac{1}{n}	(f_{1}(\beta))^n\right]
,\\
n_n(k_0R_2)
&\simeq
C_2
\left[\frac{1}{n}\left(\frac{1+\sqrt{1-\beta^2}}{\beta e^{\sqrt{1-\beta^2}}}\right)^n\right]
\triangleq
C_2
\left[\frac{1}{n}	(f_{2}(\beta))^n\right],
\end{align}
where $C_1\triangleq\frac{1}{2\beta}
\left(\frac{1-\sqrt{1-\beta^2}}{\sqrt{1-\beta^2}}\right)^{1/2}$ and $C_2\triangleq-\frac{1}{\beta}\left(\frac{1+\sqrt{1-\beta^2}}{\sqrt{1-\beta^2}}\right)^{1/2}$.
If $\beta=1$, $f_{1}(\beta)=f_{2}(\beta)=1$.
In addition, for $0<\beta<1$, both are positive and satisfy
\begin{align}
\frac{\partial f_{1}(\beta)}{\partial \beta}
=
\frac{(\beta^2-1+\sqrt{1-\beta^2})e^{\sqrt{1-\beta^2}}}{\beta}>0,
\frac{\partial f_{2}(\beta)}{\partial \beta}
=
\frac{(\beta^2-1+\sqrt{1-\beta^2})e^{\sqrt{1-\beta^2}}}{\beta}<0,
\end{align}
which shows that $0<f_1(\beta)<1, f_2(\beta)>1$.
Therefore, in the reactive near-field region, i.e., $0<\beta<1$, $|n_n(k_0R_2)|\rightarrow\infty, |j_n(k_0R_2)|\rightarrow 0$
as $n\rightarrow\infty$. By using such behavior, we have
\begin{align}
\frac{P_L}{1+P_S}
=
\frac{3(2n+1)\left(n_n^2(k_0R_2)+j_n^2(k_0R_2)\right)}{32+6(2n+1)\left(n_n^2(k_0R_2)-j_n^2(k_0R_2)\right)}
\rightarrow
\frac{1}{2}
\end{align}
as $n$ goes to infinity.

Second, consider the case when $\beta>1$, which is the outside of the reactive region. For this regime, the spherical Hankel function can be approximated as
\begin{align}
h_n^{(1)}(\rho)
\simeq
\frac{1}{\rho}
e^{i\sqrt{\rho^2-n^2}}
\left(\frac{n-i\sqrt{\rho^2-n^2}}{\rho}\right)^n
\left(\frac{n-i\sqrt{\rho^2-n^2}}{i\sqrt{\rho^2-n^2}}\right)^{1/2}
\end{align}
for $\rho\ge n\gg 1$, and thus, its absolute value satisfies
\begin{align}
|h_n^{(1)}(\rho)|
&\simeq
\frac{1}{\rho}
\left|\frac{n-i\sqrt{\rho^2-n^2}}{\rho}\right|^n
\left|\frac{n-i\sqrt{\rho^2-n^2}}{i\sqrt{\rho^2-n^2}}\right|^{1/2}
=\left(\frac{1}{\rho\sqrt{\rho^2-n^2}}\right)^{1/2}
\end{align}
for $\rho\ge n\gg 1$. For $k_0R_2=n\beta\ge n$, 
\begin{align}
|h_n^{(1)}(k_0R_2)|
&\simeq
\left(\frac{1}{n^2\beta\sqrt{\beta^2-1}}\right)^{1/2}
=
\frac{C_3}{n},
\end{align}
where $C_3\triangleq\left(\frac{1}{\beta\sqrt{\beta^2-1}}\right)^{1/2}$.
Therefore,
\begin{align}
\frac{P_L}{1+P_S}
=
\frac{3(2n+1)|h_n^{(1)}(k_0R_2)|^2}{32+6(2n+1)\left(n_n^2(k_0R_2)-j_n^2(k_0R_2)\right)}
\rightarrow
0
\end{align}
as $n\rightarrow\infty$.

\bibliographystyle{IEEEtran}
\bibliography{IEEEabrv,References_J1}

\end{document}